\documentclass[10pt,journal]{IEEEtran} 
\usepackage[noadjust]{cite}
\usepackage[bookmarks=false]{hyperref}
\hypersetup{
    colorlinks,
    linkcolor={blue!18!black},
    citecolor={blue!18!black},
    urlcolor={blue!80!black}
}
\usepackage[font=footnotesize]{subcaption}
\usepackage[font=footnotesize]{caption}
\usepackage{caption}
\usepackage[utf8]{inputenc}
\usepackage{graphicx,url}

\usepackage{booktabs}
\graphicspath{ {./images/} }
\usepackage{multirow}
\usepackage{hhline}
\usepackage{xspace}
\usepackage{graphicx}
\usepackage{amssymb,amsmath,mathtools}
\usepackage{amsthm}
\newtheorem{proposition}{Proposition} %
\newtheorem{remark}{Remark}
\newtheorem{lemma}{Lemma}

\usepackage{bm}                 %
\usepackage{pifont}
\newtheorem{definition}{Definition}

\theoremstyle{plain}

\usepackage{color,soul}
\usepackage[dvipsnames]{xcolor}
\usepackage{bbm}
\usepackage{stmaryrd}
\usepackage{textcomp}
\usepackage{lipsum,lineno}
\usepackage{float}
\makeatletter
\newif\if@restonecol
\makeatother
\usepackage{amsmath}
\usepackage{kbordermatrix}
\usepackage{mathrsfs}
\usepackage[shortlabels]{enumitem}
\newcommand\numeq[1]%
{\stackrel{\scriptscriptstyle(\mkern-1.5mu#1\mkern-1.5mu)}{=}}
\newcommand\numeqq[1]%
{\stackrel{\scriptscriptstyle(\mkern-1.5mu#1\mkern-1.5mu)}{\triangleq}}
\newcommand\numleq[1]%
{\stackrel{\scriptscriptstyle(\mkern-1.5mu#1\mkern-1.5mu)}{\leq}}
\newcommand\numgeq[1]%
{\stackrel{\scriptscriptstyle(\mkern-1.5mu#1\mkern-1.5mu)}{\geq}}
\newcommand\numimp[1]%
{\stackrel{\scriptscriptstyle(\mkern-1.5mu#1\mkern-1.5mu)}{\implies}}
\newcommand\norm[1]{\lVert#1\rVert}

\usepackage[titlenumbered,ruled,linesnumbered]{algorithm2e}
\SetAlCapNameFnt{\footnotesize}
\SetAlCapFnt{\footnotesize}
\usepackage[shortlabels]{enumitem}

\let\oldnl\nl
\newcommand{\nonl}{\renewcommand{\nl}{\let\nl\oldnl}}
\usepackage{algpseudocode}

\usepackage{fancyhdr}
\usepackage{tabularx}
\usepackage{dsfont}
\usepackage{listings}
\usepackage[most]{tcolorbox}
\lstdefinestyle{mystyle}{
    backgroundcolor=\color{backcolour},
    commentstyle=\color{codegreen},
    keywordstyle=\color{magenta},
    numberstyle=\tiny\color{codegray},
    stringstyle=\color{codepurple},
    basicstyle=\ttfamily\footnotesize,
    breakatwhitespace=false,
    breaklines=true,
    captionpos=b,
    keepspaces=true,
    showspaces=false,
    showstringspaces=false,
    showtabs=false,
    tabsize=2,
    xleftmargin=50pt,
    xrightmargin=50pt
  }
\usepackage{tikz,pgfplots}
\usepackage{pgfplots}
\usepackage{pgfplotstable}
\definecolor{gray2}{HTML}{ededed}
\definecolor{gray3}{HTML}{F5F5F5}
\definecolor{RoyalAzure}{rgb}{0.0, 0.22, 0.66}
\definecolor{lightgray}{gray}{0.9}
\definecolor{lightgray}{gray}{0.9}
\definecolor{lightgreen}{rgb}{0.88, 1, 0.88}
\definecolor{lightred}{rgb}{1, 0.88, 0.88}
\definecolor{lightblue}{rgb}{0.88, 0.94, 1}
\definecolor{lightorange}{rgb}{1, 0.94, 0.88}
\usepackage{colortbl}
\usetikzlibrary{shapes.geometric,backgrounds,patterns, trees}
\usetikzlibrary{3d,decorations.text,shapes.arrows,positioning,fit,backgrounds}
\usetikzlibrary{positioning, decorations.pathmorphing, shapes}
\usetikzlibrary{decorations.pathreplacing}
\usetikzlibrary{shapes.geometric,backgrounds,patterns, trees}
\usetikzlibrary{spy}
\usetikzlibrary{arrows.meta,
                bending,
                intersections,
                quotes,
                shapes.geometric}
              \usetikzlibrary{automata, positioning}
              \usepgfplotslibrary{fillbetween}
\usetikzlibrary{shapes,arrows}
\usetikzlibrary{arrows.meta}
\usetikzlibrary{positioning}
\tikzset{set/.style={draw,circle,inner sep=0pt,align=center}}
\usetikzlibrary{automata, positioning}
  \usetikzlibrary{shapes,shadows}
  \tikzstyle{abstractbox} = [draw=black, fill=white, rectangle,
  inner sep=10pt, style=rounded corners, drop shadow={fill=black,
  opacity=1}]
\tikzstyle{abstracttitle} =[fill=white]
\usetikzlibrary{calc,positioning,shapes.geometric}
\usetikzlibrary{arrows.meta,arrows}

\DeclareMathOperator*{\argmin}{arg\,min}

\usetikzlibrary{matrix}
\tikzstyle{cblue}=[circle, draw, thin,fill=cyan!20, scale=0.8]
\tikzstyle{qgre}=[rectangle, draw, thin,fill=green!20, scale=0.8]
\tikzstyle{rpath}=[ultra thick, red, opacity=0.4]
\tikzstyle{legend_isps}=[rectangle, rounded corners, thin,
                       fill=gray!20, text=blue, draw]

\tikzstyle{legend_overlay}=[rectangle, rounded corners, thin,
                           top color= white,bottom color=green!25,
                           minimum width=2.5cm, minimum height=0.8cm,
                           pinegreen]
\tikzstyle{legend_phytop}=[rectangle, rounded corners, thin,
                          top color= white,bottom color=cyan!25,
                          minimum width=2.5cm, minimum height=0.8cm,
                          royalblue]
\tikzstyle{legend_general}=[rectangle, rounded corners, thin,
                          top color= white,bottom color=lavander!25,
                          minimum width=2.5cm, minimum height=0.8cm,
                          violet]
                          \usetikzlibrary{matrix}

\colorlet{myRed}{red!20}
\tikzset{
  rows/.style 2 args={/utils/temp/.style={row ##1/.append style={nodes={#2}}},
    /utils/temp/.list={#1}},
  columns/.style 2 args={/utils/temp/.style={column ##1/.append style={nodes={#2}}},
    /utils/temp/.list={#1}}}
\usetikzlibrary{backgrounds,calc,shadings,shapes.arrows,shapes.symbols,shadows}
\definecolor{switch}{HTML}{006996}

\makeatletter
\pgfkeys{/pgf/.cd,
  parallelepiped offset x/.initial=2mm,
  parallelepiped offset y/.initial=2mm
}
\pgfdeclareshape{parallelepiped}
{
  \inheritsavedanchors[from=rectangle] 
  \inheritanchorborder[from=rectangle]
  \inheritanchor[from=rectangle]{north}
  \inheritanchor[from=rectangle]{north west}
  \inheritanchor[from=rectangle]{north east}
  \inheritanchor[from=rectangle]{center}
  \inheritanchor[from=rectangle]{west}
  \inheritanchor[from=rectangle]{east}
  \inheritanchor[from=rectangle]{mid}
  \inheritanchor[from=rectangle]{mid west}
  \inheritanchor[from=rectangle]{mid east}
  \inheritanchor[from=rectangle]{base}
  \inheritanchor[from=rectangle]{base west}
  \inheritanchor[from=rectangle]{base east}
  \inheritanchor[from=rectangle]{south}
  \inheritanchor[from=rectangle]{south west}
  \inheritanchor[from=rectangle]{south east}
  \backgroundpath{
    \southwest \pgf@xa=\pgf@x \pgf@ya=\pgf@y
    \northeast \pgf@xb=\pgf@x \pgf@yb=\pgf@y
    \pgfmathsetlength\pgfutil@tempdima{\pgfkeysvalueof{/pgf/parallelepiped
      offset x}}
    \pgfmathsetlength\pgfutil@tempdimb{\pgfkeysvalueof{/pgf/parallelepiped
      offset y}}
    \def\ppd@offset{\pgfpoint{\pgfutil@tempdima}{\pgfutil@tempdimb}}
    \pgfpathmoveto{\pgfqpoint{\pgf@xa}{\pgf@ya}}
    \pgfpathlineto{\pgfqpoint{\pgf@xb}{\pgf@ya}}
    \pgfpathlineto{\pgfqpoint{\pgf@xb}{\pgf@yb}}
    \pgfpathlineto{\pgfqpoint{\pgf@xa}{\pgf@yb}}
    \pgfpathclose
    \pgfpathmoveto{\pgfqpoint{\pgf@xb}{\pgf@ya}}
    \pgfpathlineto{\pgfpointadd{\pgfpoint{\pgf@xb}{\pgf@ya}}{\ppd@offset}}
    \pgfpathlineto{\pgfpointadd{\pgfpoint{\pgf@xb}{\pgf@yb}}{\ppd@offset}}
    \pgfpathlineto{\pgfpointadd{\pgfpoint{\pgf@xa}{\pgf@yb}}{\ppd@offset}}
    \pgfpathlineto{\pgfqpoint{\pgf@xa}{\pgf@yb}}
    \pgfpathmoveto{\pgfqpoint{\pgf@xb}{\pgf@yb}}
    \pgfpathlineto{\pgfpointadd{\pgfpoint{\pgf@xb}{\pgf@yb}}{\ppd@offset}}
  }
}

\makeatletter
\tikzset{anchor/.append code=\let\tikz@auto@anchor\relax,
  add font/.code=%
    \expandafter\def\expandafter\tikz@textfont\expandafter{\tikz@textfont#1},
  left delimiter/.style 2 args={append after command={\tikz@delimiter{south east}
    {south west}{every delimiter,every left delimiter,#2}{south}{north}{#1}{.}{\pgf@y}}}}
\tikzstyle{sms} = [rectangle callout, draw,very thick, rounded corners, minimum height=20pt]
\makeatletter
\tikzset{anchor/.append code=\let\tikz@auto@anchor\relax,
  add font/.code=%
    \expandafter\def\expandafter\tikz@textfont\expandafter{\tikz@textfont#1},
  left delimiter/.style 2 args={append after command={\tikz@delimiter{south east}
    {south west}{every delimiter,every left delimiter,#2}{south}{north}{#1}{.}{\pgf@y}}}}
\tikzstyle{sms} = [rectangle callout, draw,very thick, rounded corners, minimum height=20pt]
\usetikzlibrary{positioning,calc}
\tikzstyle{block} = [rectangle, draw,
text width=10.5em, text centered, rounded corners, minimum height=4em]
\tikzstyle{line} = [draw, -latex]
\tikzset{
  mybackground51/.style={execute at end picture={
      \begin{scope}[on background layer]
        \draw[black, rounded corners=2ex, fill=gray2] (current bounding box.south west)
        rectangle (current bounding box.north east);
        \node[draw,fill=white,ellipse,anchor=west,inner sep=1pt,minimum width=1ex] at (current bounding box.north
        west){#1};
      \end{scope}
    }},
}
\tikzset{
  mybackground9/.style={execute at end picture={
        \begin{scope}[on background layer]
          \draw[black,fill=black!5,rounded corners=6ex] (current bounding box.south west)
                    rectangle (current bounding box.north east);
          \node[draw,fill=white,ellipse,anchor=west,inner sep=1pt,minimum width=4ex] at (current bounding box.north
                   west){#1};
        \end{scope}
    }},
}

\tikzset{
  mybackground13/.style={execute at end picture={
        \begin{scope}[on background layer]
          \draw[black, fill=gray2, rounded corners=4ex] (current bounding box.south west)
                    rectangle (current bounding box.north east);
          \node[draw,fill=white,ellipse,anchor=west,inner sep=1pt,minimum width=4ex] at (current bounding box.north
                   west){#1};
        \end{scope}
    }},
}
\tikzset{
  mybackground14/.style={execute at end picture={
        \begin{scope}[on background layer]
          \draw[black, rounded corners=2ex] (current bounding box.south west)
                    rectangle (current bounding box.north east);
          \node[draw,fill=white,ellipse,anchor=west,inner sep=1pt,minimum width=4ex] at (current bounding box.north
                   west){#1};
        \end{scope}
    }},
}

\tikzset{
  mybackground6/.style={execute at end picture={
        \begin{scope}[on background layer]
          \draw[black,rounded corners=1ex, line width=0.15mm] (current bounding box.south west)
                    rectangle (current bounding box.north east);
          \node[draw,fill=white,ellipse,anchor=west,inner sep=1pt,minimum width=4ex] at (current bounding box.north
                   west){#1};
        \end{scope}
    }},
}

\tikzset{
  mybackground11/.style={execute at end picture={
        \begin{scope}[on background layer]
          \draw[black, fill=Black!80!Sepia!9, rounded corners=6ex] (current bounding box.south west)
                    rectangle (current bounding box.north east);
          \node[draw,fill=white,ellipse,anchor=west,inner sep=1pt,minimum width=4ex] at (current bounding box.north
                   west){#1};
        \end{scope}
    }},
}

\tikzset{
  mybackground15/.style={execute at end picture={
        \begin{scope}[on background layer]
          \draw[black, fill=Black!80!Sepia!9, rounded corners=3ex] (current bounding box.south west)
                    rectangle (current bounding box.north east);
          \node[draw,fill=white,ellipse,anchor=west,inner sep=1pt,minimum width=4ex] at (current bounding box.north
                   west){#1};
        \end{scope}
    }},
}

\tikzset{
  mybackground12/.style={execute at end picture={
        \begin{scope}[on background layer]
          \draw[black, fill=Black!40!Emerald!30, rounded corners=3ex, line width=0.3mm] (current bounding box.south west)
                    rectangle (current bounding box.north east);
        \end{scope}
    }},
}
\tikzset{
  mybackground18/.style={execute at end picture={
      \begin{scope}[on background layer]
        \draw[black, fill=gray3, rounded corners=3.5ex] (current bounding box.south west)
        rectangle (current bounding box.north east);
        \node[draw,fill=white,ellipse,anchor=west,inner sep=1pt,minimum width=4ex] at (current bounding box.north
        west){#1};
      \end{scope}
    }}
}
\tikzset{
  mybackground180/.style={execute at end picture={
      \begin{scope}[on background layer]
        \draw[black, rounded corners=3.5ex] (current bounding box.south west)
        rectangle (current bounding box.north east);
        \node[draw,fill=white,ellipse,anchor=west,inner sep=1pt,minimum width=4ex] at (current bounding box.north
        west){#1};
      \end{scope}
    }}
}
\tikzset{
  mybackground19/.style={execute at end picture={
      \begin{scope}[on background layer]
        \draw[black, rounded corners=3.5ex] (current bounding box.south west)
        rectangle (current bounding box.north east);
        \node[draw,fill=white,ellipse,anchor=west,inner sep=1pt,minimum width=4ex] at (current bounding box.north
        west){#1};
      \end{scope}
    }}
}
\tikzset{
  mybackground58/.style={execute at end picture={
        \begin{scope}[on background layer]
          \draw[black, fill=blue!40!black!5, rounded corners=1ex] (current bounding box.south west)
                    rectangle (current bounding box.north east);
          \node[draw,fill=white,ellipse,anchor=west,inner sep=1pt,minimum width=4ex, rounded corners=1ex] at (current bounding box.north
                   west){#1};
        \end{scope}
    }},
}
\tikzset{l3 switch/.style={
    parallelepiped,fill=switch, draw=white,
    minimum width=0.75cm,
    minimum height=0.75cm,
    parallelepiped offset x=1.75mm,
    parallelepiped offset y=1.25mm,
    path picture={
      \node[fill=white,
        circle,
        minimum size=6pt,
        inner sep=0pt,
        append after command={
          \pgfextra{
            \foreach \angle in {0,45,...,360}
            \draw[-latex,fill=white] (\tikzlastnode.\angle)--++(\angle:2.25mm);
          }
        }
      ]
       at ([xshift=-0.75mm,yshift=-0.5mm]path picture bounding box.center){};
    }
  },
  ports/.style={
    line width=0.3pt,
    top color=gray!20,
    bottom color=gray!80
  },
  rack switch/.style={
    parallelepiped,fill=white, draw,
    minimum width=1.25cm,
    minimum height=0.25cm,
    parallelepiped offset x=2mm,
    parallelepiped offset y=1.25mm,
    xscale=-1,
    path picture={
      \draw[top color=gray!5,bottom color=gray!40]
      (path picture bounding box.south west) rectangle
      (path picture bounding box.north east);
      \coordinate (A-west) at ([xshift=-0.2cm]path picture bounding box.west);
      \coordinate (A-center) at ($(path picture bounding box.center)!0!(path
        picture bounding box.south)$);
      \foreach \x in {0.275,0.525,0.775}{
        \draw[ports]([yshift=-0.05cm]$(A-west)!\x!(A-center)$)
          rectangle +(0.1,0.05);
        \draw[ports]([yshift=-0.125cm]$(A-west)!\x!(A-center)$)
          rectangle +(0.1,0.05);
       }
      \coordinate (A-east) at (path picture bounding box.east);
      \foreach \x in {0.085,0.21,0.335,0.455,0.635,0.755,0.875,1}{
        \draw[ports]([yshift=-0.1125cm]$(A-east)!\x!(A-center)$)
          rectangle +(0.05,0.1);
      }
    }
  },
  server/.style={
    parallelepiped,
    fill=white, draw,
    minimum width=0.35cm,
    minimum height=0.75cm,
    parallelepiped offset x=3mm,
    parallelepiped offset y=2mm,
    xscale=-1,
    path picture={
      \draw[top color=gray!5,bottom color=gray!40]
      (path picture bounding box.south west) rectangle
      (path picture bounding box.north east);
      \coordinate (A-center) at ($(path picture bounding box.center)!0!(path
        picture bounding box.south)$);
      \coordinate (A-west) at ([xshift=-0.575cm]path picture bounding box.west);
      \draw[ports]([yshift=0.1cm]$(A-west)!0!(A-center)$)
        rectangle +(0.2,0.065);
      \draw[ports]([yshift=0.01cm]$(A-west)!0.085!(A-center)$)
        rectangle +(0.15,0.05);
      \fill[black]([yshift=-0.35cm]$(A-west)!-0.1!(A-center)$)
        rectangle +(0.235,0.0175);
      \fill[black]([yshift=-0.385cm]$(A-west)!-0.1!(A-center)$)
        rectangle +(0.235,0.0175);
      \fill[black]([yshift=-0.42cm]$(A-west)!-0.1!(A-center)$)
        rectangle +(0.235,0.0175);
    }
  },
}

\usetikzlibrary{calc, shadings, shadows, shapes.arrows}
\tikzset{cross/.style={cross out, draw=black, minimum size=2*(#1-\pgflinewidth), inner sep=0pt, outer sep=0pt},
cross/.default={1pt}}
\tikzset{%
  interface/.style={draw, rectangle, rounded corners, font=\LARGE\sffamily},
  ethernet/.style={interface, fill=yellow!50},
  serial/.style={interface, fill=green!70},
  speed/.style={sloped, anchor=south, font=\large\sffamily},
  route/.style={draw, shape=single arrow, single arrow head extend=4mm,
    minimum height=1.7cm, minimum width=3mm, white, fill=switch!20,
    drop shadow={opacity=.8, fill=switch}, font=\tiny}
}
\newcommand*{\shift}{1.3cm}
\newcommand*{\router}[1]{
\begin{tikzpicture}
  \coordinate (ll) at (-3,0.5);
  \coordinate (lr) at (3,0.5);
  \coordinate (ul) at (-3,2);
  \coordinate (ur) at (3,2);
  \shade [shading angle=90, left color=switch, right color=white] (ll)
    arc (-180:-60:3cm and .75cm) -- +(0,1.5) arc (-60:-180:3cm and .75cm)
    -- cycle;
  \shade [shading angle=270, right color=switch, left color=white!50] (lr)
    arc (0:-60:3cm and .75cm) -- +(0,1.5) arc (-60:0:3cm and .75cm) -- cycle;
  \draw [thick] (ll) arc (-180:0:3cm and .75cm)
    -- (ur) arc (0:-180:3cm and .75cm) -- cycle;
  \draw [thick, shade, upper left=switch, lower left=switch,
    upper right=switch, lower right=white] (ul)
    arc (-180:180:3cm and .75cm);
  \node at (0,0.5){\color{blue!60!black}\Huge #1};
  \begin{scope}[yshift=2cm, yscale=0.28, transform shape]
    \node[route, rotate=45, xshift=\shift] {\strut};
    \node[route, rotate=-45, xshift=-\shift] {\strut};
    \node[route, rotate=-135, xshift=\shift] {\strut};
    \node[route, rotate=135, xshift=-\shift] {\strut};
  \end{scope}
\end{tikzpicture}}

\makeatletter
\pgfdeclareradialshading[tikz@ball]{cloud}{\pgfpoint{-0.275cm}{0.4cm}}{%
  color(0cm)=(tikz@ball!75!white);
  color(0.1cm)=(tikz@ball!85!white);
  color(0.2cm)=(tikz@ball!95!white);
  color(0.7cm)=(tikz@ball!89!black);
  color(1cm)=(tikz@ball!75!black)
}
\tikzoption{cloud color}{\pgfutil@colorlet{tikz@ball}{#1}%
  \def\tikz@shading{cloud}\tikz@addmode{\tikz@mode@shadetrue}}
\makeatother

\tikzset{my cloud/.style={
     cloud, draw, aspect=2,
     cloud color={gray!5!white}
  }
}

\usetikzlibrary{backgrounds}
\usetikzlibrary{patterns}
\pgfmathdeclarefunction{gauss}{3}{%
  \pgfmathparse{1/(#3*sqrt(2*pi))*exp(-((#1-#2)^2)/(2*#3^2))}%
}
\pgfmathdeclarefunction{cdf}{3}{%
  \pgfmathparse{1/(1+exp(-0.07056*((#1-#2)/#3)^3 - 1.5976*(#1-#2)/#3))}%
}
\pgfmathdeclarefunction{fq}{3}{%
  \pgfmathparse{1/(sqrt(2*pi*#1))*exp(-(sqrt(#1)-#2/#3)^2/2)}%
}
\pgfmathdeclarefunction{fq0}{1}{%
  \pgfmathparse{1/(sqrt(2*pi*#1))*exp(-#1/2)}%
}

\usepackage{etoolbox}

\usepackage{wrapfig}

\makeatletter
\newcommand{\setword}[2]{%
  \phantomsection
  #1\def\@currentlabel{\unexpanded{#1}}\label{#2}%
}
\makeatother
\usepackage{lettrine}
 \definecolor{DBrown}{HTML}{9B8879}
 \definecolor{LBrown}{HTML}{C5B99F}
 \definecolor{backg}{HTML}{BCC534}
 \definecolor{latCol}{HTML}{E8B041}
 \definecolor{pot}{HTML}{185BD9}

\tikzset{%
  wireless/.pic={
      \draw [->] (0,0) -| (.5,#1);
    \foreach \r in {.1,.2,.3}
      \draw (.6,#1) ++ (60:\r) arc (60:-60:\r);
  },
  vdots/.pic={
    \foreach \i in {-.1,0,.1}
      \fill (.25,\i) circle [radius=.75pt];
  },
  block/.style={
    shape=rectangle,
    minimum width=2cm,
    minimum height=1cm,
    draw
  },
  Tx/.style 2 args={
    block,
    node contents=Tx,
    append after command={
      \pgfextra{\pgfnodealias{@}{\tikzlastnode}}
      (@.north #1) [yshift=-.125cm] pic [#2] {wireless=.5}
      (@.#1)                        pic [#2] {vdots}
      (@.south #1) [yshift= .125cm] pic [#2] {wireless=.5}
    }
  },
  MIMO Tx east/.style={Tx={east}{xscale=1}},
  MIMO Tx west/.style={Tx={west}{xscale=-1}},
  Tx2/.style 2 args={
    node contents=,
    append after command={
      \pgfextra{\pgfnodealias{@}{\tikzlastnode}}
      (@.north #1) [yshift=-.125cm] pic [#2] {wireless=.5}
      (@.south #1) [yshift= .125cm] pic [#2] {wireless=.5}
    }
  },
  MIMO2 Tx east/.style={Tx2={east}{xscale=1}},
  MIMO2 Tx west/.style={Tx2={west}{xscale=-1}}
}

\newcommand{\figref}[1]{\hyperref[#1]{Fig.~\ref*{#1}}}
\newcommand{\probbref}[1]{\hyperref[#1]{Prob.~\ref*{#1}}}
\newcommand{\Probbref}[1]{\hyperref[#1]{Problem~\ref*{#1}}}
\newcommand{\figsref}[1]{\hyperref[#1]{Figs.~\ref*{#1}}}
\newcommand{\Figref}[1]{\hyperref[#1]{Figure~\ref*{#1}}}
\newcommand{\tableref}[1]{\hyperref[#1]{Table~\ref*{#1}}}
\newcommand{\appendixref}[1]{\hyperref[#1]{Appendix~\ref*{#1}}}
\newcommand{\appendicesref}[1]{\hyperref[#1]{Appendices~\ref*{#1}}}
\newcommand{\theoremref}[1]{\hyperref[#1]{Thm.~\ref*{#1}}}
\newcommand{\theoremsref}[1]{\hyperref[#1]{Thms.~\ref*{#1}}}
\newcommand{\Theoremref}[1]{\hyperref[#1]{Theorem~\ref*{#1}}}
\newcommand{\lemmaref}[1]{\hyperref[#1]{Lemma~\ref*{#1}}}
\newcommand{\propref}[1]{\hyperref[#1]{Prop.~\ref*{#1}}}
\newcommand{\propsref}[1]{\hyperref[#1]{Props.~\ref*{#1}}}
\newcommand{\Propref}[1]{\hyperref[#1]{Proposition~\ref*{#1}}}
\newcommand{\corref}[1]{\hyperref[#1]{Cor.~\ref*{#1}}}
\newcommand{\Corref}[1]{\hyperref[#1]{Corollary~\ref*{#1}}}
\newcommand{\scenarioref}[1]{\hyperref[#1]{Scenario~\ref*{#1}}}
\newcommand{\Scenarioref}[1]{\hyperref[#1]{\textsc{scenario}~\ref*{#1}}}
\newcommand{\probref}[1]{\hyperref[#1]{Prob.~\ref*{#1}}}
\newcommand{\Probref}[1]{\hyperref[#1]{Problem~\ref*{#1}}}
\newcommand{\gameref}[1]{\hyperref[#1]{Game~\ref*{#1}}}
\newcommand{\chapterref}[1]{\hyperref[#1]{Chapter~\ref*{#1}}}
\newcommand{\sectionref}[1]{\hyperref[#1]{\S\ref*{#1}}}
\newcommand{\Algref}[1]{\hyperref[#1]{Algorithm ~\ref*{#1}}}
\newcommand{\myalgref}[1]{\hyperref[#1]{Alg.~\ref*{#1}}}
\newcommand{\Myalgref}[1]{\hyperref[#1]{Algorithm~\ref*{#1}}}
\newcommand{\defref}[1]{\hyperref[#1]{Def.~\ref*{#1}}}
\newcommand{\Defref}[1]{\hyperref[#1]{Definition~\ref*{#1}}}
\newcommand{\assumptionref}[1]{\hyperref[#1]{Assumption~\ref*{#1}}}
\newcommand{\assumptionsref}[1]{\hyperref[#1]{Assumptions~\ref*{#1}}}
\newcommand{\remarkref}[1]{\hyperref[#1]{Remark~\ref*{#1}}}
\newcommand{\exampleref}[1]{\hyperref[#1]{Ex.~\ref*{#1}}}

\usepackage{pifont}

\newtcolorbox{promptone}{
  colback=black!5!white,
  colframe=black!30!black!70,
  title=Instruction for generating a complete recovery plan,
  fonttitle=\bfseries,
  sharp corners
}

\newtcolorbox{prompttwo}{
  colback=black!5!white,
  colframe=black!30!black!70,
  title=Attack classification instruction,
  fonttitle=\bfseries,
  sharp corners
}

\newtcolorbox{promptthree}{
  colback=black!5!white,
  colframe=black!30!black!70,
  title=Action generation instruction,
  fonttitle=\bfseries,
  sharp corners
}

\newtcolorbox{promptfour}{
  colback=black!5!white,
  colframe=black!30!black!70,
  title=State prediction instruction,
  fonttitle=\bfseries,
  sharp corners
}

\newtcolorbox{questionone}{
  colback=black!5!white,
  colframe=black!30!black!70,
  title=The CPS Recovery Problem,
  fonttitle=\bfseries,
  sharp corners
}

\newtcolorbox{motivatingexample}{
  float,
  colback=black!5!white,
  colframe=black!30!black!70,
  title=Example: Retrieval-augmented generation (\textsc{rag}),
  fonttitle=\bfseries,
  sharp corners
}

\newtcolorbox{pomdpexample}{
  colback=black!5!white,
  colframe=black!30!black!70,
  title=Example: Recovery from a network intrusion,
  fonttitle=\bfseries,
  sharp corners
}

\newtcolorbox{motivatingexamplethree}{
  colback=black!5!white,
  colframe=black!30!black!70,
  title=Example: In-context learning,
  fonttitle=\bfseries,
  sharp corners
}

\newtcolorbox{problemtwo}{
  colback=black!5!white,
  colframe=black!30!black!70,
  title=Problem $2$ (event-based attestation strategy),
  fonttitle=\bfseries,
  sharp corners
}

\makeatletter
\tikzset{
    database/.style={
        path picture={
            \draw (0, 1.5*\database@segmentheight) circle [x radius=\database@radius,y radius=\database@aspectratio*\database@radius];
            \draw (-\database@radius, 0.5*\database@segmentheight) arc [start angle=180,end angle=360,x radius=\database@radius, y radius=\database@aspectratio*\database@radius];
            \draw (-\database@radius,-0.5*\database@segmentheight) arc [start angle=180,end angle=360,x radius=\database@radius, y radius=\database@aspectratio*\database@radius];
            \draw (-\database@radius,1.5*\database@segmentheight) -- ++(0,-3*\database@segmentheight) arc [start angle=180,end angle=360,x radius=\database@radius, y radius=\database@aspectratio*\database@radius] -- ++(0,3*\database@segmentheight);
        },
        minimum width=2*\database@radius + \pgflinewidth,
        minimum height=3*\database@segmentheight + 2*\database@aspectratio*\database@radius + \pgflinewidth,
    },
    database segment height/.store in=\database@segmentheight,
    database radius/.store in=\database@radius,
    database aspect ratio/.store in=\database@aspectratio,
    database segment height=0.1cm,
    database radius=0.25cm,
    database aspect ratio=0.35,
  }
\makeatother

\usepackage{listofitems} 
\usetikzlibrary{arrows.meta} 
\usepackage[outline]{contour} 
\contourlength{1.4pt}

\usepackage{xcolor}
\colorlet{myred}{red!80!black}
\colorlet{myblue}{blue!80!black}
\colorlet{mygreen}{green!60!black}
\colorlet{myorange}{orange!70!red!60!black}
\colorlet{mydarkred}{red!30!black}
\colorlet{mydarkblue}{blue!40!black}
\colorlet{mydarkgreen}{green!30!black}

\tikzset{
  >=latex, 
  node/.style={thick,circle,draw=myblue,minimum size=22,inner sep=0.5,outer sep=0.6},
  node in/.style={node,black!20!black,draw=mygreen!30!black,fill=black!20},
  node hidden/.style={node,black!20!black,draw=myblue!30!black,fill=black!20},
  node convol/.style={node,black!20!black,draw=myorange!30!black,fill=black!20},
  node out/.style={node,red!20!black,draw=myred!30!black,fill=black!20},
  connect/.style={thick,Blue!100}, 
  connect arrow/.style={-{Latex[length=4,width=3.5]},thick,mydarkblue,shorten <=0.5,shorten >=1},
  node 1/.style={node in}, 
  node 2/.style={node hidden},
  node 3/.style={node out}
}
\def\nstyle{int(\lay<\Nnodlen?min(2,\lay):3)} 

\newcommand{\cmark}{\textcolor{OliveGreen}{\ding{51}}} 
\newcommand{\xmark}{\textcolor{Red}{\ding{55}}}  
\newcommand{\qmark}{\textcolor{Blue}{\textbf{?}}}

\usepackage[most]{tcolorbox}
\usepackage{tikz}
\usepackage{xcolor}
\usepackage{fontawesome}
\usepackage{lipsum} 
\usepackage{enumitem}
\usepackage{setspace}
\usepackage{caption}
\usepackage{fancyvrb}

\definecolor{highlightA}{RGB}{255, 192, 192} 
\definecolor{highlightB}{RGB}{192, 255, 192} 
\definecolor{highlightC}{RGB}{255, 240, 150} 

\newtcolorbox{responsebox}[2][]{
  enhanced,
  sharp corners,
  colback=white,
  colframe=black,
  boxrule=0.5pt,
  width=\textwidth,
  title=#2,
  fonttitle=\bfseries,
  coltitle=black,
  attach boxed title to top left={yshift=-2mm, xshift=4mm},
  boxed title style={colback=white},
  #1
}

\definecolor{bluetwo}{RGB}{189, 213, 234}
\definecolor{bluethree}{RGB}{165, 193, 224}

\newcommand{\acro}[1]{\textsc{#1}\xspace}

\newcommand{\ssh}{\acro{ssh}}
\newcommand{\irc}{\acro{irc}}
\newcommand{\smtp}{\acro{smtp}}

\newcommand{\mysql}{\acro{mysql}}
\newcommand{\tcpp}{\acro{tcp}}
\newcommand{\xmas}{\acro{xmas}}

\newcommand{\udp}{\acro{udp}}
\newcommand{\syn}{\acro{syn}}
\newcommand{\mongo}{\acro{mongodb}}
\newcommand{\postgres}{\acro{postgres}}
\newcommand{\telnet}{\acro{telnet}}

\newcommand{\cassandra}{\acro{cassandra}}

\newcommand{\vulscan}{\acro{vulscan}}

\newcommand{\ftp}{\acro{ftp}}

\newcommand{\cve}{\acro{cve}}
\newcommand{\cwe}{\acro{cwe}}

\makeatletter
\patchcmd{\@makecaption}
  {\scshape}
  {}
  {}
  {}
\makeatother

\allowdisplaybreaks

\begin{document}
\bstctlcite{MyBSTcontrol}

\title{Incident Response Planning Using a Lightweight Large Language Model with Reduced Hallucination}

\author{\IEEEauthorblockN{Kim Hammar\IEEEauthorrefmark{2}, Tansu Alpcan\IEEEauthorrefmark{2}, and Emil C. Lupu\IEEEauthorrefmark{3}}\\
 \IEEEauthorblockA{\IEEEauthorrefmark{2}
   Department of Electrical and Electronic Engineering, University of Melbourne, Australia\\
 }
 \IEEEauthorblockA{\IEEEauthorrefmark{3}
   Department of Computing, Imperial College London, United Kingdom\\
 } 
 Email: \{kim.hammar,tansu.alpcan\}@unimelb.edu.au, and e.c.lupu@imperial.ac.uk\\
}
\maketitle
\begin{abstract}  
Timely and effective incident response is key to managing the growing frequency of cyberattacks. However, identifying the right response actions for complex systems is a major technical challenge. A promising approach to mitigate this challenge is to use the security knowledge embedded in large language models (\textsc{llm}s) to assist security operators during incident handling. Recent research has demonstrated the potential of this approach, but current methods are mainly based on prompt engineering of frontier \textsc{llm}s, which is costly and prone to hallucinations. We address these limitations by presenting a novel way to use an \textsc{llm} for incident response planning with reduced hallucination. Our method includes three steps: fine-tuning, information retrieval, and lookahead planning. We prove that our method generates response plans with a bounded probability of hallucination and that this probability can be made arbitrarily small at the expense of increased planning time under certain assumptions. Moreover, we show that our method is lightweight and can run on commodity hardware. We evaluate our method on logs from incidents reported in the literature. The experimental results show that our method a) achieves up to $22$\% shorter recovery times than frontier \textsc{llm}s and b) generalizes to a broad range of incident types and response actions.
\end{abstract}
\section{Introduction}
\lettrine[lines=2]{\textbf{I}}{ncident} response refers to the coordinated actions taken to contain, mitigate, and recover from cyberattacks. Today, incident response is largely a manual process carried out by security operators \cite{287145}. While this approach can be effective, it is often slow, labor-intensive, and requires significant skills. For example, a recent study reports that organizations take an average of 73 days to respond and recover from an incident \cite{ibm2024costofdatabreach}. Reducing this delay requires better decision-support tools to assist operators during incident handling. Currently, the standard approach to assisting operators relies on \textit{response playbooks} \cite{10.1145/3491102.3517559}, which comprise predefined rules for handling specific incidents. However, playbooks still rely on security experts for configuration and are therefore difficult to keep aligned with evolving threats and system architectures \cite{10646756}.

\begin{figure}
  \centering
  \scalebox{1.69}{
   \input{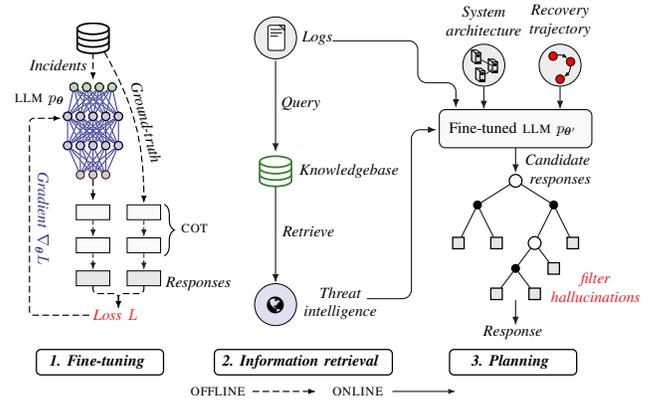}    
  }
  \caption{The three steps of our method for incident response planning: 1. fine-tuning of a (lightweight) large language model (\textsc{llm}); 2. retrieval of relevant threat intelligence; and 3. decision-theoretic planning and chain-of-thought (\textsc{cot} \cite{10.5555/3600270.3602070}) reasoning to select effective responses and filter hallucinations.}
  \label{fig:framework}
\end{figure}

To overcome these limitations, an emerging direction of research is to leverage the security knowledge encoded in large language models (\textsc{llm}s) to generate effective response actions \cite{castro2025largelanguagemodelsautonomous, rigaki2023cage, 10991969,yan2024dependingshouldmentoringllm,hays2024employingllmsincidentresponse,lin2025ircopilotautomatedincidentresponse,10540988}. These actions can then be used as suggestions to security operators. Although this approach remains largely confined to academic settings for now, it is beginning to see commercial adoption, as exemplified by \textsc{ibm}’s recent launch of an \textsc{llm}-based response service \cite{hussey2025instana}. Most of the \textsc{llm}-based methods proposed in the literature so far are based on prompt engineering of frontier \textsc{llm}s, such as \textsc{openai o3} \cite{openai2024gpt4technicalreport}. While this approach has shown promise, it is costly and relies on an external \textsc{llm} provider (e.g., \textsc{google} or \textsc{openai}), which limits flexibility. Another important concern with this approach is that frontier \textsc{llm}s are not specialized for incident response, which makes them particularly prone to \textit{hallucinations} \cite{NEURIPS2024_3c1e1fdf}, i.e., they may generate response actions that appear plausible but are incorrect or unrelated to the incident.

In this paper, we present a novel method that addresses these limitations and provides a principled way to use an \textsc{llm} as decision support for incident response; see \figref{fig:framework}. Our method includes three main steps: (\textit{i}) instruction fine-tuning of a lightweight \textsc{llm} to align it with the phases and objectives of incident response; (\textit{ii}) retrieval-augmented generation (\textsc{rag}) to ground the \textsc{llm} in current threat information and system knowledge; and (\textit{iii}) decision-theoretic planning and chain-of-thought (\textsc{cot}) reasoning to generate effective response actions.

We evaluate our method based on log data from incidents reported in the literature. The results show that our method surpasses the performance of frontier \textsc{llm}s (e.g., \textsc{gemini 2.5} \cite{comanici2025gemini25pushingfrontier, geminiteam2024geminifamilyhighlycapable}) by up to $22$\% while being far less resource-intensive. Moreover, we show that our method performs comparably to the \textsc{ppo} reinforcement learning method \cite{ppo}, despite not relying on incident-specific training like \textsc{ppo} does. We also present an ablation study assessing the contribution of the individual steps of our method. We show that all steps contribute to its performance, with fine-tuning and planning having the greatest impact. In addition to the empirical results, we present a theoretical analysis that establishes a probabilistic upper bound on the hallucination probability of our method.

Our contributions can be summarized as follows:
\begin{itemize}
 \item We develop a novel method for incident response that integrates a lightweight \textsc{llm} with instruction fine-tuning, information retrieval, and decision-theoretic planning.
 \item We derive a probabilistic upper bound on the hallucination probability of our method. Under certain assumptions, this bound can be made arbitrarily small at the expense of increased planning time.
 \item We evaluate our method on logs from incidents reported in the literature. The results show that our method a) achieves up to $22$\% shorter recovery times than frontier \textsc{llm}s; b) generalizes to a broad range of incidents and responses; and c) performs comparably to a reinforcement learning method that is pretrained for each incident.
 \item We release the first \textsc{llm} fine-tuned for incident response, together with a dataset of $68,000$ incidents and the corresponding responses. We also provide source code and a video demonstration of a decision-support system for incident response that implements our method \cite{llm_source_kim}.
 \end{itemize}

\section{Related Work}\label{sec:related_work}
Since the early 2000s, there has been broad interest in developing systems that can assist security operators during incident response \cite{10.1504/IJICS.2007.012248,tansu_response_2003}. Traditional decision-support systems are based on playbooks that map incident scenarios to sequences of response actions \cite{10.1504/IJICS.2007.012248,playbook_response}, such as those provided by \textsc{splunk} \cite{splunk_playbook}, \textsc{cisa} \cite{cisa_playbook}, and \textsc{oasis} \cite{oasis_playbook}. Although playbooks can be effective, they rely on security experts for configuration. As a consequence, they are difficult to keep up-to-date with evolving security threats and system architectures \cite{10646756}. Another common critique of playbooks is that they consist of generic response actions that are difficult for non-experts to interpret and execute effectively \cite{10.1145/3491102.3517559}. Several research efforts have aimed to address these limitations by \textit{automating} the generation of effective incident response strategies and functions. Four predominant approaches to such automation have emerged: decision-theoretic \cite{dsn24_hammar_stadler}, reinforcement learning \cite{singh2024hierarchicalmultiagentreinforcementlearning}, game-theoretic \cite{10955193,5270307,nework_security_alpcan}, and \textsc{llm}-based approaches \cite{castro2025largelanguagemodelsautonomous}.
 
The first three approaches share a common requirement: they need a perfect simulator (model) that captures how the system evolves in response to attacks and defensive actions. The simulator enables the computation of optimal response strategies (according to the model) through numerical optimization techniques. For example, a standard benchmark in this line of research is \textsc{cage-2} \cite{cage_challenge_2_announcement}, which simulates an advanced persistent threat on an enterprise network. State-of-the-art methods evaluated on this benchmark include dynamic programming \cite{tifs_25_HLALB}, reinforcement learning \cite{vyas2023automated}, and tree search \cite{hammar2024optimaldefenderstrategiescage2}, all of which rely on a simulator. While these approaches can be effective when high-fidelity simulators are available, such simulators are rarely available in practice. Furthermore, the resulting response strategies are limited in scope as they are trained on a narrow set of attack vectors and response options. For instance, the \textsc{cage-2} simulation is limited to around $20$ attacker actions and defensive countermeasures \cite{hammar2024optimaldefenderstrategiescage2}.

A promising approach to address this drawback is to use large language models (\textsc{llm}s) to automatically generate effective response actions based on system logs. This approach is not limited to a predefined set of actions and eliminates the need for a simulator. Early studies in this direction include \cite{castro2025largelanguagemodelsautonomous, rigaki2023cage, 10991969,yan2024dependingshouldmentoringllm,hays2024employingllmsincidentresponse,lin2025ircopilotautomatedincidentresponse,10540988}, and \cite{hussey2025instana}. Notably, the work in \cite{hussey2025instana} is a commercial product by \textsc{ibm}. While these works report encouraging results, they have three key limitations: they do not provide a theoretical analysis, they do not address the risk of hallucinations, and most of them require \textsc{api} access to frontier \textsc{llm}s.

Our method differs from prior work in several ways. It does not rely on a simulator or a manually-designed playbook, is lightweight enough to run on commodity hardware, has reduced hallucination, is accompanied by a theoretical analysis, and combines fine-tuning with retrieval-augmented generation (\textsc{rag}); see \tableref{tab:related_work}. Moreover, ours is the only \textsc{llm}-based method that is fully open-source (code, weights, and data).

 \begin{table}[H]
  \centering
  \scalebox{0.58}{
    \begin{tabular}{llllllll} \toprule
\rowcolor{lightgray}
      {\textit{Method}} & {\textit{Theory}} & {\textit{\textsc{rag}}} & {\textit{Fine-tuning}} & {\textit{Lightweight}} & {\textsc{llm}} & {\textit{Req. simulator}} & {\textit{Manual}}\\ \midrule
    \rowcolor{lightgreen}
      \textsc{ours} (\figref{fig:framework}) & \cmark & \cmark & \cmark & \cmark & \cmark & \xmark & \xmark\\
    \rowcolor{lightblue}      
      \cite{castro2025largelanguagemodelsautonomous},\cite{rigaki2023cage}--\cite{lin2025ircopilotautomatedincidentresponse} & \xmark & \xmark & \xmark & \xmark & \cmark & \xmark & \xmark\\
    \rowcolor{lightblue}      
      \cite{hussey2025instana} & \xmark & \qmark & \qmark & \qmark & \cmark & \xmark & \xmark\\
    \rowcolor{lightblue}      
      \cite{10540988} & \xmark & \cmark & \xmark & \xmark & \cmark & \xmark & \xmark\\
    \rowcolor{lightred}            
      \cite{tifs_25_HLALB,hammar2024optimaldefenderstrategiescage2,li2024conjectural} & \cmark & \xmark & \xmark & \cmark & \xmark & \cmark & \xmark\\
    \rowcolor{lightred}                  
      \cite{vyas2023automated},\cite{tabular_Q_andy,singh2024hierarchicalmultiagentreinforcementlearning,ramamurthy2025generalautonomouscybersecuritydefense,huang2025intentbasedontologydrivenautonomicsecurity} & \xmark & \xmark & \xmark & \cmark & \xmark & \cmark & \xmark\\      
      \cite{playbook_response,10.1145/3538969.3538976,10.1145/3688810} & \xmark & \xmark & \xmark & \cmark & \xmark & \xmark & \cmark\\
    \bottomrule\\
  \end{tabular}}
\caption{Comparison between our method and related approaches, which can be grouped into three categories: those relying on playbooks (white row), those relying on a simulator for numerical optimization (red rows), and those using \textsc{llm}s (blue rows). Compared to other \textsc{llm}-based approaches, our method (green row) is the only method that does not depend on frontier \textsc{llm}s, is lightweight enough to run on commodity hardware, has reduced hallucination probability, and is accompanied by a theoretical analysis.}\label{tab:related_work}
\end{table}

Lastly, we note that a growing body of research applies \textsc{llm}s to security use cases other than incident response, such as penetration testing \cite{pentest_gpt,rodriguez2025frameworkevaluatingemergingcyberattack}, security assistants \cite{DBLP:conf/ndss/DengLCBWLW025}, scanning \cite{DBLP:conf/ndss/StafeevRSKP25,299549}, threat hunting \cite{google_llm_recovery}, verification \cite{DBLP:conf/ndss/0012XW00S025}, piracy \cite{DBLP:conf/ndss/GohilDNSR25}, detection \cite{DBLP:conf/ndss/YangL0L25}, fuzzing \cite{10.1145/3597503.3639121,299896}, \textsc{api} design \cite{DBLP:conf/ndss/LiuY0L25}, network operations \cite{net_llm}, threat intelligence \cite{ARAZZI2025100765}, and decompilation \cite{DBLP:conf/ndss/HuL024}. Compared to these works, the main novelty of our method lies in its approach to reducing hallucinations.

\begin{figure*}
  \centering
  \scalebox{1.57}{
   \input{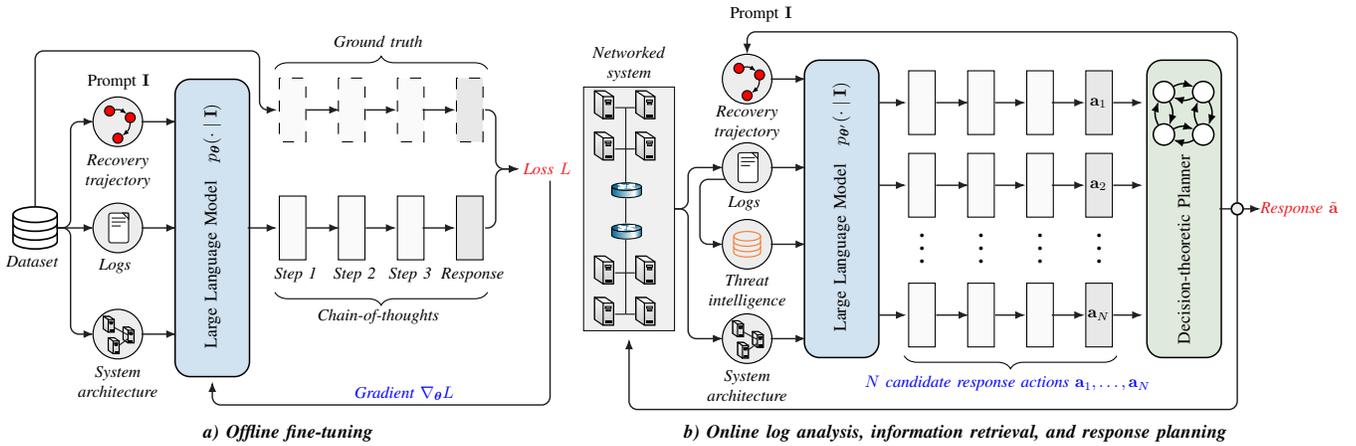}    
  }
  \caption{The two phases of our method. In the first phase [cf.~\textbf{a})], an \textsc{llm} is fine-tuned offline on a dataset of logs from 68,000 incidents paired with response plans and chain-of-thought reasoning steps \cite{10.5555/3600270.3602070}. In the second phase [cf.~\textbf{b})], system logs and threat intelligence are processed online by the fine-tuned \textsc{llm} and used to generate $N$ candidate responses. These responses are then evaluated via a planning procedure, which selects the most effective response.}
  \label{fig:deployment}
\end{figure*}
\section{The Incident Response Problem}
Incident response involves selecting a sequence of actions that restores a networked system to a secure and operational state after a cyberattack. These actions should analyze the scope of the attack, secure forensic evidence, contain and evict the attacker, harden the system to prevent recurrence, and restore critical services. Examples of response actions include redirecting network flows, updating access control policies, patching vulnerabilities, shutting down compromised systems, and restarting operational services. From a security engineering perspective \cite{ross_security}, incident response fits within the broader cyber resilience framework by operationalizing the response and recovery phase after a cyberattack \cite{Ganin2016}.

\Figref{fig:resilience} illustrates the phases of incident response. Following the attack is a \textit{response time} interval, which represents the delay between the attack and the first response. This phase is followed by a \textit{recovery time} interval, during which response actions are deployed. When selecting these actions, the goal is to restore the system to a secure and operational state as quickly as possible while minimizing operational costs. A key challenge to achieving this goal is that the information about the attack is often limited to partial indicators of compromise (e.g., log files and alerts), while the full scope and severity of the attack are unknown \cite{273869}. Another major difficulty is that even short delays in initiating the response can lead to significant costs. For example, in the event of a ransomware attack, a delay of just a few minutes may allow the malware to encrypt systems or spread laterally across the network \cite{wannacry_nhgs}.
\begin{figure}[H]
  \centering
  \scalebox{0.85}{
   \begin{tikzpicture}

\node[scale=0.8] (kth_cr) at (0,2.15)
{
  \begin{tikzpicture}

\draw[-{Latex[width=2mm]}, black, thick, line width=0.4mm] (0,0) to (11,0);
\draw[-{Latex[width=2mm]}, black, thick, line width=0.4mm] (0,0) to (0,3.2);

\draw[-, black, thick, line width=0.4mm, dashed, name path=normal] (0,2.5) to (11,2.5);
\draw[-, black, thick, line width=0.4mm, name path=recovery1] (0,2.48) to (3.5,2.48) to (3.5, 1.1) to (5, 1.1);

\draw[-, black, thick, line width=0.4mm, name path=recovery5] (3.5, 2.5) to (3.5, 1.1) to (6.5, 1.1) to[bend left=15] (10.5, 2.5) to (11,2.5);

\draw[-, black, thick, line width=0.4mm, dotted] (3.5,2.95) to (3.5,-0.5);
\draw[-, black, thick, line width=0.4mm, dotted] (10.5,2.95) to (10.5,-0.5);
\draw (3.5,0) node[cross,rotate=0, scale=6] {};

\path[
pattern=north west lines, pattern color=Red,
        intersection segments={
                of=normal and recovery5,
                sequence={R2--L2}
              }];

\draw[-{Latex[width=1.2mm]}, black, thick, line width=0.3mm, bend right=20] (8,2.8) to (5,1.9);

\node[inner sep=0pt,align=center, scale=1, rotate=0, opacity=1] (obs) at (9.3,2.8)
{
Operational cost
};

\draw[{Latex[width=1.2mm]}-{Latex[width=1.2mm]}, black, thick, line width=0.3mm] (3.3,2.4) to (3.3,1.1);
\draw[{Latex[width=1.2mm]}-{Latex[width=1.2mm]}, black, thick, line width=0.4mm] (0,1.1) to (3.3,1.1);
\draw[{Latex[width=1.2mm]}-{Latex[width=1.2mm]}, black, thick, line width=0.4mm] (6.6,1.1) to (10.5,1.1);
\draw[{Latex[width=1.2mm]}-{Latex[width=1.2mm]}, black, thick, line width=0.4mm] (3.5,0.9) to (6.4,0.9);

\draw[{Latex[width=1.2mm]}-{Latex[width=1.2mm]}, black, thick, line width=0.3mm] (3.3,0.1) to (3.3,1);

\node[inner sep=0pt,align=center, scale=1, rotate=0, opacity=1] (obs) at (4.7,-0.35)
{
  Attack event
};
%
%
\node[inner sep=0pt,align=center, scale=1, rotate=0, opacity=1] (obs) at (11.45,0)
{
  Time
};
\node[inner sep=0pt,align=center, scale=1, rotate=0, opacity=1] (obs) at (8.65,0.8)
{
  Recovery time $T$
};
\node[inner sep=0pt,align=center, scale=1, rotate=0, opacity=1] (obs) at (4.95,0.57)
{
  Response time
};
\node[inner sep=0pt,align=center, scale=1, rotate=0, opacity=1] (obs) at (1.7,1.35)
{
  Survivability
};
\node[inner sep=0pt,align=center, scale=1, rotate=0, opacity=1] (obs) at (2.45,0.58)
{
  Tolerance
};
\node[inner sep=0pt,align=center, scale=1, rotate=90, opacity=1] (obs) at (3,1.8)
{
  Loss
};
%
%
\node[inner sep=0pt,align=center, scale=1, rotate=90, opacity=1] (obs) at (-0.25,1.6)
{
  Service/security
};
%
%

  \end{tikzpicture}
};

\end{tikzpicture}        
  }
  \caption{Phases and performance metrics of the incident response problem.}
  \label{fig:resilience}
\end{figure}
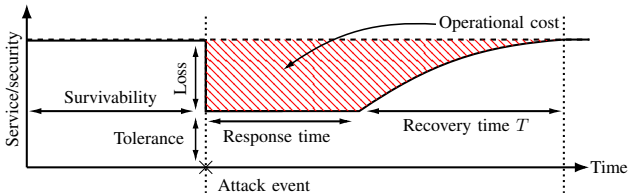
\section{Our Method for Incident Response Planning}
Motivated by the challenges described above, we develop a method for using an \textsc{llm} as decision support during incident handling, i.e., to help security operators identify and execute effective response actions quickly. Broadly speaking, our method takes as input a description of an incident (e.g., system logs, security alerts, and threat intelligence) and produces as output a sequence of recommended response actions. The main challenge in generating such recommendations is to ensure that the response actions are effective despite the possibility that the \textsc{llm} hallucinates. In the following subsections, we describe the steps we take to address this challenge.

\subsection{Overview of Our Approach and System Architecture}\label{sec:overview}
Our method consists of three main steps: (\textit{i}) supervised fine-tuning of a lightweight \textsc{llm} to align it with the objectives of incident response; (\textit{ii}) retrieval-augmented generation (\textsc{rag}) to ground the \textsc{llm} in current threat information and system knowledge; and (\textit{iii}) decision-theoretic planning to synthesize effective response actions. These steps can be divided into two phases: an offline phase for fine-tuning and an online phase for information retrieval and response generation; see \figref{fig:deployment}.

The first step of our method is to fine-tune a lightweight \textsc{llm} for incident response. We conduct this fine-tuning by training the \textsc{llm} on a labeled dataset of incident logs paired with corresponding response actions and reasoning steps. This training enables the \textsc{llm} to learn typical patterns of incident handling. For example, it learns the logical dependencies between different phases of the response process, such as containment and eviction. Another benefit of fine-tuning is that it can reduce hallucinations; see e.g., \cite{tonmoy2024comprehensivesurveyhallucinationmitigation}.
\begin{remark}
We call an \textsc{llm} \emph{lightweight} if it has significantly fewer parameters than a typical frontier \textsc{llm}. For the experimental results reported in this paper, we use the \textsc{deepseek-r1-14b} \textsc{llm}, which has $14$ billion parameters. This parameter count is small in comparison with that of the frontier \textsc{llm} \textsc{deepseek-r1}, which has $671$ billion parameters \cite{deepseekai2025deepseekr1incentivizingreasoningcapability}.
\end{remark}

Once fine-tuned, the \textsc{llm} can provide decision support for incident response by generating a sequence of recommended response actions when prompted with details about an incident. However, because the \textsc{llm} is trained on historical incident data, it cannot generate response actions that relate to newly discovered vulnerabilities or attack techniques. To address this limitation, we augment the system logs with additional threat information retrieved online. Specifically, we automatically extract \textit{indicators of compromise} from the logs (e.g., hostnames and vulnerability identifiers) and use them to retrieve relevant information from external sources, such as threat intelligence \textsc{api}s and vulnerability databases. We then append this information to the logs before prompting the \textsc{llm}. In addition to improving the quality of the response, several empirical studies have shown that such \textit{retrieval-augmented generation} also reduces the probability of hallucinations \cite{ayala-bechard-2024-reducing}.

Lastly, instead of directly selecting the response action generated by the fine-tuned \textsc{llm}, we use the \textsc{llm} to generate several candidate actions and select the one that is least likely to be hallucinated. In particular, we evaluate each candidate action by using the \textsc{llm} to simulate possible outcomes of the action, after which we select the action that leads to the shortest expected recovery time. This lookahead planning enforces a form of \textit{self-consistency} \cite{DBLP:conf/iclr/0002WSLCNCZ23}, where actions are validated against the \textsc{llm}’s predicted outcomes. Such validation has been shown in prior work to reduce hallucinations; see e.g., \cite{chen2023universalselfconsistencylargelanguage}, \cite{weng-etal-2023-large}, and \cite{10.5555/3666122.3668141}. We provide a theoretical justification for why this procedure can reduce hallucination in \sectionref{sec:hall_theory}.

Each of these three steps (fine-tuning, information retrieval, and planning) is detailed below, starting with fine-tuning.
\subsection{Instruction Fine-Tuning}\label{sec:fine_tuning}
Our goal with fine-tuning is to make the pre-trained \textsc{llm} generate appropriate response actions when prompted with system logs describing an incident. In this context, we view the pre-trained \textsc{llm} as a probabilistic model that takes as input a sequence of tokens $\mathbf{x}=\mathbf{x}_1, \mathbf{x}_2, \dots, \mathbf{x}_n$ and predicts the probability distribution over the subsequent token as
\begin{equation}
p_{\bm{\theta}}(\mathbf{x}_{n+1} \mid \mathbf{x}_1, \mathbf{x}_2, \hdots, \mathbf{x}_n), \label{eq:next_token_pred}
\end{equation}
where $\bm{\theta}$ denotes the model parameters.

The next-token prediction in (\ref{eq:next_token_pred}) allows us to generate response actions as follows. We start by concatenating a description of the incident (e.g., system logs) with an instruction to generate a response action. We then pass the resulting text through a tokenizer that converts it into a sequence of tokens $\mathbf{x}=\mathbf{x}_1, \hdots, \mathbf{x}_n$. Next, we feed these tokens into the \textsc{llm} to generate the next token by sampling from (\ref{eq:next_token_pred}). Subsequently, we append the generated token to the prompt and feed the entire sequence back into the \textsc{llm} to predict the next token. We repeat this process autoregressively until the \textsc{llm} generates a special end-of-sequence token, which is produced when the \textsc{llm} determines that the response action is complete, i.e., when the action has been fully specified.

\begin{remark}
We place no restrictions on the form of a response action. It may be a single command, a compound procedure, or any other textual description, depending on the incident.
\end{remark}  

To steer the \textsc{llm} toward generating effective responses, we fine-tune it using supervised learning on a dataset of $68,000$ instruction-answer pairs $\mathcal{D} = {(\mathbf{x}^i, \mathbf{y}^i)}_{i=1}^K$, where each instruction $\mathbf{x}^i$ consists of information related to an incident and a task for the \textsc{llm} to perform. The associated answer $\mathbf{y}^i$ describes the correct steps to complete the task, paired with a sequence of chain-of-thought (\textsc{cot} \cite{10.5555/3600270.3602070}) reasoning steps that explain the answer. We use two types of instructions: \textit{action-generation} instructions and \textit{state-prediction} instructions.

In the first case, the vector $\mathbf{x}^i$ represents an instruction to generate a response to an incident. In the latter case, $\mathbf{x}^i$ represents an instruction to assess the current status of the incident response process. For example, the instruction may be to determine whether the attack has been contained, whether the system has been hardened to prevent recurrence, or whether forensic evidence has been secured. See \appendixref{app:example_dataset} for the prompt templates and details on how we construct the dataset.
\begin{remark}
We do not fine-tune the \textsc{llm} for a specific incident scenario. Rather, the training dataset spans a diverse set of incidents, log formats, and response types, enabling the \textsc{llm} to generalize across a wide range of incident scenarios.
\end{remark}
Given the training dataset $\mathcal{D}$, we fine-tune the \textsc{llm} by iteratively sampling a batch of instruction-answer pairs $(\mathbf{x}^{1}, \mathbf{y}^{1}), \hdots, (\mathbf{x}^{M}, \mathbf{y}^{M})$ and updating its parameters via gradient descent based on the cross-entropy loss
\begin{align}
L = -\frac{1}{M}\sum_{i=1}^M\sum_{k=1}^{m_i}\ln p_{\bm{\theta}}\left(\mathbf{y}^{i}_k \mid \mathbf{x}^{i}, \mathbf{y}^{i}_{1},\hdots, \mathbf{y}^{i}_{k-1} \right),\label{cross_entropy_loss}
\end{align}
where $m_i$ is the length of the vector $\mathbf{y}^{i}$. We denote the fine-tuned parameter vector by $\bm{\theta}^{\prime}$ to distinguish it from $\bm{\theta}$.

\Figref{fig:fine_tune} displays the training loss curves when fine-tuning the \textsc{deepseek-r1-14b} \textsc{llm} \cite{deepseekai2025deepseekr1incentivizingreasoningcapability}. We run the experiment on $4\times$\textsc{rtx 8000} \textsc{gpu}s and compare a higher learning rate (blue) with a lower one (red). We observe that the higher learning rate results in convergence to a lower loss. Additional experimental details and hyperparameters can be found in \appendixref{appendix:hyperparameters}.

\begin{figure}[H]
  \centering
  \scalebox{0.81}{
   \input{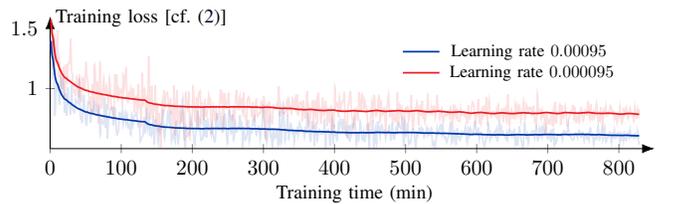}     
  }
  \caption{Loss curves when fine-tuning the \textsc{deepseek-r1-14b} \cite{deepseekai2025deepseekr1incentivizingreasoningcapability} \textsc{llm} under two different learning rates. The solid lines indicate the mean loss and the shaded lines represent the loss on specific batches of training examples; cf.~(\ref{cross_entropy_loss}).}
  \label{fig:fine_tune}
\end{figure}
\subsection{Retrieval-Augmented Response Generation (\textsc{rag})}\label{section:rag}
While the fine-tuned \textsc{llm} can generate effective response actions, its outputs depend on the distribution of incidents seen during training. This presents a limitation as the \textsc{llm} is trained on historical data that may not reflect the most recent threat landscape. To address this challenge, we use indicators of compromise (e.g., vulnerability identifiers or hostnames) in the system logs to retrieve relevant threat intelligence from external sources. By incorporating such information at the time of action generation, the \textsc{llm} can adapt its responses to reflect up-to-date threat information and system knowledge \cite{lewis2020retrieval}.

As an example, consider a scenario where the \textsc{llm} is trained on data available only up to 2020. Suppose that the \textsc{llm} is prompted with information about an incident that relates to a vulnerability discovered after $2020$, e.g., \textsc{cve-2021-44228} \cite{cve}. In this case, the \textsc{llm} may not have sufficient information to generate effective response actions, as illustrated below.
\begin{itemize}
\item \textsc{without \textsc{rag}}. Prompted only with the logs, the \textsc{llm} generates the action: ``\textit{isolate host}'' as it has no knowledge about the nature of the vulnerability \textsc{cve-2021-44228}.
\item \textsc{with \textsc{rag}}. The system retrieves information about specific mitigations for \textsc{cve-2021-44228}. When provided with this information, the \textsc{llm} generates a response action with targeted mitigations for \textsc{cve-2021-44228}, thereby reducing the time to recover from the incident.
\end{itemize}  

\subsection{Incident Response Planning}\label{sec:decision_planning}
Having fine-tuned the \textsc{llm} to produce response actions from incident logs, we now address the challenge of selecting the most effective action. Although the \textsc{llm} can produce effective actions in many cases, it may also hallucinate and generate ineffective actions. To reduce the risk of such hallucinations, our method includes a planning procedure where we use the \textsc{llm} to generate multiple candidate actions and then select the action least likely to be hallucinated, as described below.

\vspace{2mm}
\noindent\textit{\textbf{System model.}}
We formulate incident response planning as a stochastic shortest path problem. In this formulation, the response process evolves over a sequence of \textit{time steps} $t=0,1,\hdots,T$ and the goal is to generate a sequence of \textit{actions} $\mathbf{a}_0,\mathbf{a}_1,\hdots,\mathbf{a}_{T-1}$ that quickly recovers the system from the incident. In other words, the goal is to minimize the \textit{recovery time} $T$. To formalize this goal, we model the progress toward system recovery with a \textit{recovery state}. We define this state based on the \textsc{mitre d3fend} \cite{kaloroumakis2021d3fend} taxonomy as follows.
\begin{definition}[Recovery state]\label{def:recovery_state}
  The recovery state is a vector
\begin{align}
\mathbf{s}_t = (s^{\mathrm{I}}_t, s^{\mathrm{S}}_t, s^{\mathrm{F}}_t, s^{\mathrm{E}}_t, s^{\mathrm{H}}_t,  s^{\mathrm{R}}_t), \label{eq:state_def}
\end{align}
where each component is a binary variable indicating whether a specific stage of response is completed. In particular,
\begin{itemize}
\item \textbf{\textit{Containment}}: $s^{\mathrm{I}}_t=1$ if the attack has been isolated and stopped from spreading; $s^{\mathrm{I}}_t=0$ otherwise.
\item \textbf{\textit{Assessment}}: $s^{\mathrm{S}}_t=1$ if the scope and severity of the attack have been determined; $s^{\mathrm{S}}_t=0$ otherwise.
\item \textbf{\textit{Preservation}}: $s^{\mathrm{F}}_t=1$ if forensic evidence related to the incident has been preserved; $s^{\mathrm{F}}_t=0$ otherwise.
\item \textbf{\textit{Eviction}}: $s^{\mathrm{E}}_t=1$ if the attacker's access has been revoked and potential malicious code or processes have been removed from the system; $s^{\mathrm{E}}_t=0$ otherwise.
\item \textbf{\textit{Hardening}}: $s^{\mathrm{H}}_t=1$ if the system has been hardened to prevent recurrence of the same attack; $s^{\mathrm{H}}_t=0$ otherwise.
\item \textbf{\textit{Restoration}}: $s^{\mathrm{R}}_t=1$ if services have been restarted and user access has been restored; $s^{\mathrm{R}}_t=0$ otherwise.  
\end{itemize}  
\end{definition}
Given this definition of the recovery state, we define the \textit{recovery time} to be the number of time steps until the \textit{terminal recovery state} $\mathbf{s}=(1,1,1,1,1,1)$ is reached. Since both the system and the attacker may behave stochastically, we model the recovery time as a random variable, as defined below.
\begin{definition}[Recovery time]\label{def:recovery_time}
The recovery time $T$ is a random variable that takes on values in the set $\{1,2,\hdots\}$ and represents the time to reach the terminal state, i.e.,
\begin{align*}
T = \inf\{t \mid t > 0, \mathbf{s}_t=(1,1,1,1,1,1)\}.
\end{align*}
\end{definition}
To illustrate the preceding definitions, we show two possible state trajectories $\mathbf{s}_0,\mathbf{s}_1, \hdots, \mathbf{s}_T$ in \figref{fig:states}. As shown in the figure, several response actions may achieve the same effect on the recovery state. For example, the severity of the attack can be determined in several ways. Moreover, certain response actions can lead to shorter recovery times by skipping intermediate steps. For instance, in the event of a denial of service (\textsc{dos}) attack, containment and eviction can often be achieved simultaneously by appropriate filtering of the network traffic. This single action both isolates the attack ($s_t^{\mathrm{I}}=1$) and revokes attacker access ($s_t^{\mathrm{E}}=1$). In contrast, an advanced persistent threat (\textsc{apt}) typically requires multiple actions to complete these stages. For example, containment may involve isolating compromised hosts ($s_t^{\mathrm{I}}=1$) and eviction may require malware removal or credential rotation ($s_t^{\mathrm{E}}=1$). Thus, the recovery time $T$ for an \textsc{apt} is typically longer than for a \textsc{dos} attack.

\begin{figure}
  \centering
  \scalebox{1.35}{
   \begin{tikzpicture}

\node[draw,circle, minimum width=10mm, scale=0.45, fill=Blue!10](s0) at (0,0) {\large $\mathbf{s}_0$};
\node[draw,circle, minimum width=10mm, scale=0.45](s1) at (0,-1) {\small $\mathbf{s}_1,\mathbf{s}_2$};
\node[draw,circle, minimum width=10mm, scale=0.45](s2) at (-1,-2) {\large $\mathbf{s}_3$};
\node[draw,circle, minimum width=10mm, scale=0.45](s3) at (1,-2) {\large $\mathbf{s}'_3$};
\node[draw,circle, minimum width=10mm, scale=0.45](s4) at (0,-3) {\large $\mathbf{s}_4$};
\node[draw,circle, minimum width=10mm, scale=0.45](s5) at (0,-4) {\large $\mathbf{s}_5$};
\node[draw,circle, minimum width=10mm, scale=0.45](s6) at (0,-5) {\large $\mathbf{s}_6$};
\node[draw,circle, minimum width=10mm, scale=0.45, fill=Red!20](s8) at (0,-6) {\large $\mathbf{s}_T$};
\draw[-{Latex[length=1.3mm]}, line width=0.17mm] (s0) to (s1);
\draw[-{Latex[length=1.3mm]}, line width=0.17mm] (s1) to (s2);
\draw[-{Latex[length=1.3mm]}, line width=0.17mm] (s1) to (s3);
\draw[-{Latex[length=1.3mm]}, line width=0.17mm, looseness=5, out=110, in=170] (s1) to (s1);
\draw[-{Latex[length=1.3mm]}, line width=0.17mm] (s3) to (s4);
\draw[-{Latex[length=1.3mm]}, line width=0.17mm] (s2) to (s4);
\draw[-{Latex[length=1.3mm]}, line width=0.17mm] (s4) to (s5);
\draw[-{Latex[length=1.3mm]}, line width=0.17mm] (s5) to (s6);
\draw[-{Latex[length=1.3mm]}, line width=0.17mm] (s6) to (s8);
\draw[-{Latex[length=1.3mm]}, line width=0.17mm, bend right=35] (s5) to (s8);


%
%

\node[inner sep=0pt,align=center, scale=0.55, color=black] (hacker) at (1.8,-1) {
\textsc{isolate}
};
\node[inner sep=0pt,align=center, scale=0.55, color=black] (hacker) at (1.8,0) {
\textsc{initiate}
};
\node[inner sep=0pt,align=center, scale=0.55, color=black] (hacker) at (1.8,-2) {
\textsc{assess}
};
\node[inner sep=0pt,align=center, scale=0.55, color=black] (hacker) at (1.8,-3) {
\textsc{preserve}
};
\node[inner sep=0pt,align=center, scale=0.55, color=black] (hacker) at (1.8,-4) {
\textsc{evict}
};
\node[inner sep=0pt,align=center, scale=0.55, color=black] (hacker) at (1.8,-5) {
\textsc{harden}
};
\node[inner sep=0pt,align=center, scale=0.55, color=black] (hacker) at (1.8,-6) {
\textsc{restore}
};
\node[inner sep=0pt,align=center, scale=0.5, color=black] (hacker) at (2.55,0.22) {
$s^{\mathrm{I}}$
};
\node[inner sep=0pt,align=center, scale=0.5, color=black] (hacker) at (2.75,0.22) {
$s^{\mathrm{S}}$
};
\node[inner sep=0pt,align=center, scale=0.5, color=black] (hacker) at (2.94,0.22) {
$s^{\mathrm{F}}$
};
\node[inner sep=0pt,align=center, scale=0.5, color=black] (hacker) at (3.12,0.22) {
$s^{\mathrm{E}}$
};
\node[inner sep=0pt,align=center, scale=0.5, color=black] (hacker) at (3.31,0.22) {
$s^{\mathrm{H}}$
};
\node[inner sep=0pt,align=center, scale=0.5, color=black] (hacker) at (3.51,0.22) {
$s^{\mathrm{R}}$
};

\node[inner sep=0pt,align=center, scale=0.55, color=black] (hacker) at (3,-1) {
$(\text{\textcolor{Red}{$1$}},0,0,0,0,0)$
};
\node[inner sep=0pt,align=center, scale=0.55, color=black] (hacker) at (3,0) {
$(0,0,0,0,0,0)$
};
\node[inner sep=0pt,align=center, scale=0.55, color=black] (hacker) at (3,-2) {
$(1,\text{\textcolor{Red}{$1$}},0,0,0,0)$
};
\node[inner sep=0pt,align=center, scale=0.55, color=black] (hacker) at (3,-3) {
$(1,1,\text{\textcolor{Red}{$1$}},0,0,0)$
};
\node[inner sep=0pt,align=center, scale=0.55, color=black] (hacker) at (3,-4) {
$(1,1,1,\text{\textcolor{Red}{$1$}},0,0)$
};
\node[inner sep=0pt,align=center, scale=0.55, color=black] (hacker) at (3,-5) {
$(1,1,1,1,\text{\textcolor{Red}{$1$}},0)$
};
\node[inner sep=0pt,align=center, scale=0.55, color=black] (hacker) at (3,-6) {
$(1,1,1,1,1,\text{\textcolor{Red}{$1$}})$
};

\node[inner sep=0pt,align=center, scale=0.45, color=black] (hacker) at (0,0.42) {
\textit{logs indicates anomalous}\\\textit{activity on a host}
};
\node[inner sep=0pt,align=center, scale=0.45, color=black] (hacker) at (0.78,-0.4) {
$\mathbf{a}_0$: \textit{segment network}\\
\textit{to isolate host}
};

\node[inner sep=0pt,align=center, scale=0.45, color=black] (hacker) at (-0.91,-0.7) {
\textcolor{black}{$\mathbf{a}_1$: \textit{malware}}\\
\textit{scan \textcolor{Red}{(failed)}}
};

\node[inner sep=0pt,align=center, scale=0.45, color=black] (hacker) at (-1.43,-1.4) {
\textcolor{black}{$\mathbf{a}_2$: \textit{analyze logs}}\\
\textcolor{Blue}{\textit{(found malicious process)}}
};
\node[inner sep=0pt,align=center, scale=0.45, color=black] (hacker) at (1.43,-1.4) {
\textcolor{black}{$\mathbf{a}'_2$: \textit{analyze processes}}\\
\textcolor{Blue}{\textit{(found malicious process)}}
};

\node[inner sep=0pt,align=center, scale=0.45, color=black] (hacker) at (0,-2.2) {
$\mathbf{a}_3$: \textit{memory dump}
};
\node[inner sep=0pt,align=center, scale=0.45, color=black] (hacker) at (0.68,-3.5) {
$\mathbf{a}_4$: \textit{stop process}
};

\node[inner sep=0pt,align=center, scale=0.45, color=black] (hacker) at (0.6,-4.4) {
$\mathbf{a}_5$: \textit{upgrade}\\
\textit{affected software}
};

\node[inner sep=0pt,align=center, scale=0.45, color=black] (hacker) at (-1,-5) {
$\mathbf{a}'_5$: \textit{live-patch}\\
\textit{vulnerability}
};

\node[inner sep=0pt,align=center, scale=0.45, color=black] (hacker) at (0.72,-5.4) {
$\mathbf{a}_6$: \textit{restart service}
};

%


\end{tikzpicture}        
  }
  \caption{Two example evolutions of the recovery state $\mathbf{s}_t$; cf.~(\ref{eq:state_def}). The first recovery trajectory involves the actions $\mathbf{a}_0,\mathbf{a}_1,\mathbf{a}_2,\mathbf{a}_3, \mathbf{a}_4, \mathbf{a}_5, \mathbf{a}_6$ and the second trajectory involves the actions $\mathbf{a}_0, \mathbf{a}_1, \mathbf{a}^{\prime}_2, \mathbf{a}_3, \mathbf{a}_4, \mathbf{a}^{\prime}_5$.}
  \label{fig:states}
\end{figure}
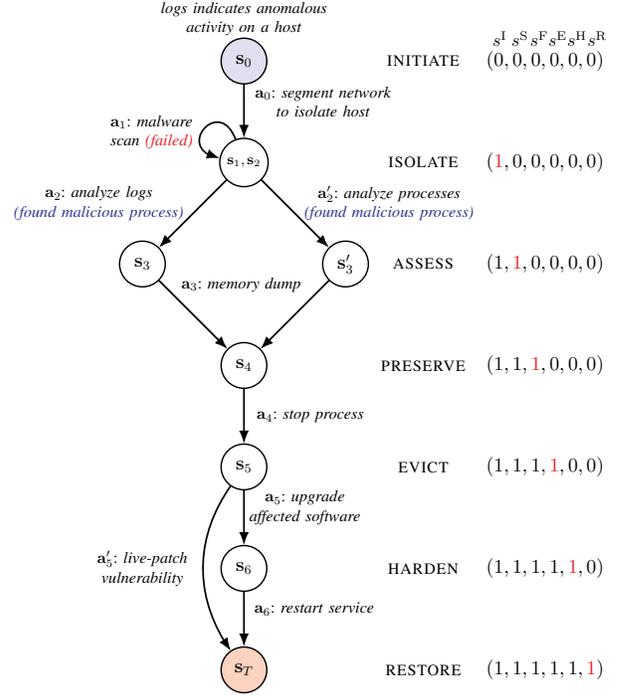
\vspace{2mm}
\noindent\textit{\textbf{Response generation.}}
Because the recovery state contains information about the attacker, it is generally not known with certainty. However, the \textsc{llm} can predict the state based on the available system logs and threat intelligence, which we denote by $\mathbf{I}$. Such predictions allow us to generate a response plan through auto-regressive sampling as follows. We start by generating the first action as $\mathbf{a}_0 \sim p_{\bm{\theta}^{\prime}}(\cdot \mid \mathbf{s}_0, \mathbf{I})$, where the initial state is $\mathbf{s}_0=(0,0,0,0,0,0)$. Subsequently, we evaluate the effect of the action by predicting the next recovery state as $\tilde{\mathbf{s}}_1 \sim p_{\bm{\theta}^{\prime}}(\cdot \mid \mathbf{s}_0, \mathbf{a}_0, \mathbf{I})$. We then repeat the same procedure to generate the next action as $\mathbf{a}_1 \sim p_{\bm{\theta}^{\prime}}(\cdot \mid \tilde{\mathbf{s}}_1, \mathbf{I})$. This iterative procedure continues until the \textsc{llm} predicts that the terminal recovery state $\tilde{\mathbf{s}}_{t} = (1,1,1,1,1,1)$ has been reached.
\begin{remark}
To instruct the \textsc{llm} whether to generate an action or to predict the state, we append an instruction to the prompt. For brevity, we do not explicitly denote this instruction in the equations. Our prompt templates are available at \cite{llm_source_kim}.
\end{remark}  

The expected time to recover from the incident when using response actions generated by the \textsc{llm} depends on the current recovery state $\mathbf{s}_t$ (which captures the effects of previous actions) and the type of incident, as characterized by the vector $\mathbf{I}$. We formally define this recovery time-to-go as follows.
\begin{definition}[Recovery time-to-go]\label{def:time_to_go}
Given an incident described by $\mathbf{I}$, the expected recovery time-to-go from the state $\mathbf{s}$ when executing actions generated by the \textsc{llm} $p_{\bm{\theta}^{\prime}}$ is
\begin{align*}
J(\mathbf{s}) =
\begin{cases}
    0 & \text{if } \mathbf{s}=(1,1,1,1,1,1), \\
    \mathbb{E}_{\mathbf{a}_t \sim p_{\bm{\theta}^{\prime}}(\cdot \mid \tilde{\mathbf{s}}_t, \mathbf{I})}\left\{T \mid \mathbf{s}_0=\mathbf{s}, \mathbf{I}\right\} & \text{otherwise.}
\end{cases}
\end{align*}
\end{definition}
Given this definition, we say that a response action is \textit{hallucinated} if it has no effect on the expected recovery time-to-go. In other words, it does not contribute any progress toward recovery. This notion is formally defined below.
\begin{definition}[Hallucinated response action]\label{def:hallucination}
A response action $\mathbf{a}_t$ is hallucinated if it leads to a recovery state with the same expected recovery time-to-go as the current state, i.e.,
\begin{align*}
J(\mathbf{s}_t)-\mathbb{E}_{\mathbf{s}_{t+1}}\left\{J(\mathbf{s}_{t+1}) \mid \mathbf{a}_t, \mathbf{s}_t, \mathbf{I}\right\} = 0, && \text{for all }\mathbf{s}_t \in \tilde{\mathcal{S}},
\end{align*}
where $\tilde{\mathcal{S}}$ denotes the set of all states except $(1,1,1,1,1,1)$.
\end{definition}
This definition implies that hallucinations can be avoided by iteratively generating actions until one is found that reduces the expected recovery time-to-go. However, this approach to reducing hallucinations is not feasible in practice, as computing the recovery time requires knowledge of the attacker’s behavior. An alternative approach is to involve a security operator to manually assess whether each action is hallucinated. While feasible, this approach defeats the purpose of using an \textsc{llm} in the first place, namely to assist security operators.

To circumvent these limitations, we adopt a different approach, known as \textit{self-verification} \cite{weng-etal-2023-large}. Following this approach, we \textit{estimate} the recovery time-to-go of response actions using the \textsc{llm} itself. This verification enforces a form of \textit{self-consistency} \cite{DBLP:conf/iclr/0002WSLCNCZ23}, where actions are validated against the \textsc{llm}’s predicted outcomes. Such validations have been shown to reduce hallucinations (see e.g., \cite{chen2023universalselfconsistencylargelanguage} and \cite{10.5555/3666122.3668141}) and form the basis for our planning algorithm, as described below. (For a theoretical justification for why this approach can reduce the probability of hallucinated response actions, see \sectionref{sec:hall_theory}.)

\vspace{2mm}
\noindent\textit{\textbf{Planning algorithm.}}  At each time $t$ of the response, we use the \textsc{llm} to generate $N$ candidate actions $\mathcal{A}^N_t=\{\mathbf{a}^1_t,\hdots,\mathbf{a}^N_t\}$. Then, for each action $\mathbf{a}^i_t$, we use the \textsc{llm} to simulate a \textit{recovery trajectory} $\tilde{\mathbf{s}}_{t+1},\mathbf{a}_{t+1}\hdots,$ by sampling actions and updating the state until $\tilde{\mathbf{s}}_T=(1,1,1,1,1,1)$. We then use the length of the simulated trajectory as an estimate of the expected recovery time-to-go. We define this estimate as
\begin{align*}
\tilde{Q}(\tilde{\mathbf{s}}_t, \mathbf{a}_t^i) \approx 1 + \sum_{\mathbf{s}_{t+1} \in \mathcal{S}}p_{\bm{\theta}^{\prime}}(\mathbf{s}_{t+1} \mid \tilde{\mathbf{s}}_t, \mathbf{a}_t^i, \mathbf{I})\Tilde{J}(\mathbf{s}_{t+1}),
\end{align*}
where $\Tilde{J}$ is the estimated recovery time-to-go function.

Finally, we select the action with the shortest expected recovery time-to-go according to the estimate, i.e.,
\begin{align}
\tilde{\mathbf{a}}_t \in \argmin_{\mathbf{a}^i_t \in \mathcal{A}^N_t}\tilde{Q}(\tilde{\mathbf{s}}_t, \mathbf{a}^{i}_t). \label{eq:rollout}
\end{align}
This planning procedure is illustrated conceptually in \figref{fig:lookahead_planning} and the pseudocode is listed in \myalgref{alg:our_method}. In the next section, we analyze the theoretical properties of this procedure and establish conditions under which it reduces hallucination. We also derive a bound on its hallucination probability.
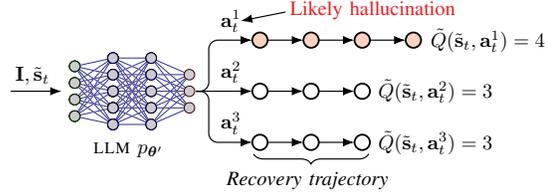
\begin{figure}
  \centering
  \scalebox{1.7}{
   \begin{tikzpicture}

\draw[-{Latex[length=0.9mm]}, line width=0.09mm, rounded corners=2pt] (-0.95,0) to (-0.55,0);

\node[inner sep=0pt,align=center, scale=0.45, color=black] (hacker) at (-0.75,0.12) {
$\mathbf{I},\tilde{\mathbf{s}}_t$
};

\node[scale=0.13] (impact) at (0,0) {
\begin{tikzpicture}[x=2.3cm,y=1.0cm]
  \readlist\Nnod{4,5,5,3}
  
  \foreachitem \N \in \Nnod{
    \def\lay{\Ncnt}
    \pgfmathsetmacro\prev{int(\Ncnt-1)}
    \message{\lay,}
    \foreach \i [evaluate={\y=\N/2-\i; \x=\lay; \n=\nstyle;
                           \nprev=int(\prev<\Nnodlen?min(2,\prev):3);}] in {1,...,\N}{
      
      \coordinate (N\lay-\i) at (\x,\y);
      
      \ifnum\lay>1
        \foreach \j in {1,...,\Nnod[\prev]}{
          \draw[connect,white,line width=1.2] (N\prev-\j) -- (N\lay-\i);
          \draw[connect] (N\prev-\j) -- (N\lay-\i);
          \node[node \nprev,minimum size=18] at (N\prev-\j) {};
        }
        \ifnum \lay=\Nnodlen
          \node[node \n,minimum size=18] at (N\lay-\i) {};
        \fi
      \fi
      
    }
  }
  \end{tikzpicture}
};

\node[draw,circle, minimum width=10mm, scale=0.12, fill=Red!20](s00) at (1,0.4) {};
\node[draw,circle, minimum width=10mm, scale=0.12, fill=Red!20](s01) at (1.4,0.4) {};
\node[draw,circle, minimum width=10mm, scale=0.12, fill=Red!20](s02) at (1.8,0.4) {};
\node[draw,circle, minimum width=10mm, scale=0.12, fill=Red!20](s03) at (2.2,0.4) {};

\node[draw,circle, minimum width=10mm, scale=0.12](s10) at (1,0) {};
\node[draw,circle, minimum width=10mm, scale=0.12](s11) at (1.4,0) {};
\node[draw,circle, minimum width=10mm, scale=0.12](s12) at (1.8,0) {};

\node[draw,circle, minimum width=10mm, scale=0.12](s20) at (1,-0.4) {};
\node[draw,circle, minimum width=10mm, scale=0.12](s21) at (1.4,-0.4) {};
\node[draw,circle, minimum width=10mm, scale=0.12](s22) at (1.8,-0.4) {};

\draw[-{Latex[length=0.9mm]}, line width=0.09mm, rounded corners=2pt] (0.5,0) to (s10);
\draw[-{Latex[length=0.9mm]}, line width=0.09mm, rounded corners=2pt] (0.5,0) to (0.6, 0) to (0.6, 0.4) to (s00);
\draw[-{Latex[length=0.9mm]}, line width=0.09mm, rounded corners=2pt] (0.5,0) to (0.6, 0) to (0.6, -0.4) to (s20);

\draw[-{Latex[length=0.9mm]}, line width=0.09mm, rounded corners=2pt] (s00) to (s01);
\draw[-{Latex[length=0.9mm]}, line width=0.09mm, rounded corners=2pt] (s01) to (s02);
\draw[-{Latex[length=0.9mm]}, line width=0.09mm, rounded corners=2pt] (s02) to (s03);

\draw[-{Latex[length=0.9mm]}, line width=0.09mm, rounded corners=2pt] (s10) to (s11);
\draw[-{Latex[length=0.9mm]}, line width=0.09mm, rounded corners=2pt] (s11) to (s12);

\draw[-{Latex[length=0.9mm]}, line width=0.09mm, rounded corners=2pt] (s20) to (s21);
\draw[-{Latex[length=0.9mm]}, line width=0.09mm, rounded corners=2pt] (s21) to (s22);

\node[inner sep=0pt,align=center, scale=0.45, color=black] (hacker) at (0.8,0.15) {
$\mathbf{a}^2_t$
};
\node[inner sep=0pt,align=center, scale=0.45, color=black] (hacker) at (0.8,0.55) {
$\mathbf{a}^1_t$
};
\node[inner sep=0pt,align=center, scale=0.45, color=black] (hacker) at (0.8,-0.25) {
$\mathbf{a}^3_t$
};

\node[inner sep=0pt,align=center, scale=0.45, color=black] (hacker) at (2.81,0.4) {
$\tilde{Q}(\tilde{\mathbf{s}}_t, \mathbf{a}^{1}_t) = 4$
};

\node[inner sep=0pt,align=center, scale=0.45, color=black] (hacker) at (2.4,0) {
$\tilde{Q}(\tilde{\mathbf{s}}_t, \mathbf{a}^{2}_t) = 3$
};
\node[inner sep=0pt,align=center, scale=0.45, color=black] (hacker) at (2.4,-0.4) {
$\tilde{Q}(\tilde{\mathbf{s}}_t, \mathbf{a}^{3}_t) = 3$
};

\node[inner sep=0pt,align=center, scale=0.45, color=black] (hacker) at (0,-0.46) {
\textsc{llm} $p_{\bm{\theta}^{\prime}}$
};

\node[inner sep=0pt,align=center, scale=0.45, color=black] (hacker) at (1.4,-0.7) {
\textit{Recovery trajectory}
};

\node[inner sep=0pt,align=center, scale=0.45, color=black] (hacker) at (1.91,0.63) {
\textcolor{Red}{Likely hallucination}
};
\draw[-{Latex[length=0.9mm]}, line width=0.09mm, rounded corners=2pt] (1.2, 0.62) to (0.85, 0.55);

\node[scale=0.9, rotate=90] (s3) at (1.4,-0.53)
{
\begin{tikzpicture}
\draw[decorate, decoration={brace, mirror, amplitude=2.5pt}, line width=0.1mm] (2,1) -- (2,0);
\end{tikzpicture}
};

\end{tikzpicture}    
  }
  \caption{Our planning procedure to circumvent hallucinations. At each time step of the response, we prompt the \textsc{llm} with the incident description $\mathbf{I}$ and the current (predicted) recovery state $\tilde{\mathbf{s}}_t$ to generate $N$ candidate actions (here $N=3$). We then estimate the expected recovery time-to-go [denoted by $\Tilde{Q}(\tilde{\mathbf{s}}_t, \mathbf{a}^i_t)$] of each action $\mathbf{a}^i_t$ by using the \textsc{llm} to simulate possible recovery trajectories. If an action is hallucinated (per \defref{def:hallucination}), then it does not make progress toward recovery and thus leads to a longer recovery trajectory. Therefore, we select the action that leads to the shortest predicted recovery trajectory (either $\mathbf{a}^2_t$ or $\mathbf{a}^3_t$ in this example); cf.~(\ref{eq:rollout}). This planning allows us to circumvent hallucinations under certain conditions, see \sectionref{sec:filter} for details.}
  \label{fig:lookahead_planning}
\end{figure}

\begin{algorithm}
\footnotesize
\SetNoFillComment
\SetKwProg{myalg}{Method}{}{}
\SetKwBlock{recov}{Procedure \normalfont \textsc{recovery-time} ($\tilde{\mathbf{s}}$, $\mathbf{a}, \mathbf{I}$)}{end}
\DontPrintSemicolon
\caption{Incident response planning with an \textsc{llm}.}\label{alg:our_method}
\textbf{Input:} \textsc{llm} $p_{\bm{\theta}'}$, system logs $\mathbf{I}$, \# actions $N$, \# samples $M$.\;
\textbf{Output:} A response plan $\rho$, i.e., a sequence of response actions.\;
Initialize $\tilde{\mathbf{s}}_0\leftarrow(0,0,0,0,0,0), \rho \leftarrow \emptyset, t\leftarrow 0$.\;
\While{$\tilde{\mathbf{s}}_t \neq (1,1,1,1,1,1)$}{
Sample $\mathbf{a}^1_t,\hdots,\mathbf{a}^N_t$ from $p_{\bm{\theta}^{\prime}}(\cdot \mid \tilde{\mathbf{s}}_t, \mathbf{I})$.\;
\For{$i=1,2,\hdots,N$}{
$\Tilde{Q}(\tilde{\mathbf{s}}_t, \mathbf{a}^i_t) = \frac{1}{M}\sum_{k=1}^{M}\text{\textsc{recovery-time}}(\tilde{\mathbf{s}}_t, \mathbf{a}^i_t, \mathbf{I})$.\;    
}  
Select action $\tilde{\mathbf{a}}_t \in \argmin_{\mathbf{a}^i_t}\Tilde{Q}(\tilde{\mathbf{s}}_t, \mathbf{a}^i_t)$, $\rho \leftarrow \rho \cup \{(t, \tilde{\mathbf{a}}_t)\}$.\;
Update the state as $\tilde{\mathbf{s}}_{t+1} \sim p_{\bm{\theta}^{\prime}}(\cdot \mid \tilde{\mathbf{s}}_{t},\tilde{\mathbf{a}}_{t}, \mathbf{I})$, $t \leftarrow t+1$.\;
}
\Return $\rho$.\;
\recov{   
  Predict the state as $\tilde{\mathbf{s}}^{\prime} \sim p_{\bm{\theta}^{\prime}}(\cdot \mid \tilde{\mathbf{s}},\mathbf{a}, \mathbf{I})$.\;
  \If{$\tilde{\mathbf{s}}^{\prime} = (1,1,1,1,1,1)$}{
      \Return $1$.\;
    }\Else{
      Sample $\mathbf{a}^{\prime}$ from $p_{\bm{\theta}^{\prime}}(\cdot \mid \tilde{\mathbf{s}}^{\prime}, \mathbf{I})$.\;
      \Return $1 + \text{\textsc{recovery-time}}(\tilde{\mathbf{s}}^{\prime}, \mathbf{a}^{\prime}, \mathbf{I})$.\;
    }
  
  }
\normalsize
 \end{algorithm}
\section{Analysis of the Hallucination Probability}\label{sec:hall_theory}
To analyze the probability that our method generates a hallucinated response action, we distinguish between two cases: (\textit{i}) at least one of the $N$ candidate actions is non-hallucinated; and (\textit{ii}) all $N$ actions are hallucinated; cf.~(\ref{eq:rollout}). In the following, we establish a sufficient condition under which hallucinations are avoided in case (\textit{i}), and derive a probabilistic upper bound on the probability that case (\textit{ii}) occurs.

\subsection{Sufficient Conditions for Filtering Hallucinations}\label{sec:filter}
The purpose of the minimization (\ref{eq:rollout}) is to filter \textit{hallucinated actions}, i.e., actions that do not affect the recovery time. This filtering is effective when the lookahead simulations [cf.~lines $13$--$22$ in \myalgref{alg:our_method}] accurately reflect that hallucinated actions have no beneficial impact on the expected recovery time. However, because these simulations rely on the \textsc{llm} to predict action outcomes, the filtering is inherently imperfect. Consequently, the effectiveness of the planning step [cf.~(\ref{eq:rollout})] depends on two key factors: (\textit{i}) the degree to which hallucinated and non-hallucinated actions can be distinguished based on their impact on expected recovery time; and (\textit{ii}) the accuracy of the \textsc{llm}'s predictions of the resulting recovery state $\mathbf{s}_t$; cf.~\defref{def:recovery_state}.

To quantify these two factors, let $\mathcal{A}$ denote the set of all possible response actions (as defined by the vocabulary of the \textsc{llm}) and let $\mathcal{A}(\mathbf{s},\mathbf{I})$ be the subset of non-hallucinated actions for the incident described by $\mathbf{I}$, given the recovery state $\mathbf{s}$. Moreover, let $\delta$ denote the minimal change in the recovery time-to-go when taking a non-hallucinated action, i.e.,
\begin{align*}
\delta = \min\left\{J(\mathbf{s}_t)-\mathbb{E}_{\mathbf{s}_{t+1}}\left\{J(\mathbf{s}_{t+1}) \mid \mathbf{a}, \mathbf{s}_t, \mathbf{I}\right\} \big\vert\text{ } \mathbf{a}\in\mathcal{A}(\mathbf{s},\mathbf{I})\right\}.
\end{align*}
In view of \defref{def:hallucination}, we have $\delta > 0$.

Similarly, let $\eta$ denote the total variation between the \textsc{llm}'s predictions and the true system dynamics (denoted by $P$), i.e.,
\begin{align*}
\sum_{\mathbf{s}^{\prime} \in \mathcal{S}} \big\vert p_{\bm{\theta}^{\prime}}(\mathbf{s}^{\prime} \mid \mathbf{s}, \mathbf{a}, \mathbf{I}) - P(\mathbf{s}^{\prime} \mid \mathbf{s}, \mathbf{a},\mathbf{I}) \big\vert \leq \eta, && \forall \mathbf{s}\in\tilde{\mathcal{S}},\mathbf{a}\in\mathcal{A},
\end{align*}
where $\mathcal{S}$ is the set of all recovery states and $\tilde{\mathcal{S}}$ is the set of non-terminal recovery states, i.e., $\tilde{\mathcal{S}}=\mathcal{S} \setminus \{(1,1,1,1,1,1)\}$. Note that the parameter $\eta$ is upper bounded by $2$, i.e., $0 \leq \eta \leq 2$.

Given the parameters $\delta$ and $\eta$, we have the following result.
\begin{proposition}\label{theorem:hallucination_bound}
Assuming that a) the number of sample trajectories $M$ in \myalgref{alg:our_method} is sufficiently large so that the empirical mean approximates the true expectation and b) that both the expected recovery time and the \textsc{llm}'s predicted recovery time are finite, i.e., $\norm{J}_{\infty}< \infty$ and $\norm{\Tilde{J}}_{\infty} < \infty$. If at least one action in the set $\mathcal{A}^N_t$ [cf.~(\ref{eq:rollout})] is non-hallucinated and
\begin{align*}
\delta > 2\eta\norm{J}_{\infty}\left(\norm{\tilde{J}}_{\infty} + 1\right),
\end{align*}
then the action selected by \myalgref{alg:our_method} will be non-hallucinated.
\end{proposition}
We present the proof of \propref{theorem:hallucination_bound} in \appendixref{app:hallucination_bound_proof}. This proposition provides a sufficient condition under which the minimization (\ref{eq:rollout}) effectively filters hallucinated actions. The main condition of the proposition is that $\delta$ (which captures the degree to which hallucinations and non-hallucinations can be distinguished) is sufficiently large in comparison with the inaccuracy of the \textsc{llm}'s predictions, as quantified by $\eta$. While these parameters are likely unknown in practice, they can be estimated offline. For instance, $\delta$ can be estimated based on a curated set of incidents with known recovery outcomes. Similarly, $\eta$ can be estimated by comparing the \textsc{llm}'s predictions with traces of historical incidents.

If at least one action in the set $\mathcal{A}^N_t$ [cf.~(\ref{eq:rollout})] would always be non-hallucinated, \propref{theorem:hallucination_bound} would imply a condition that provides a guarantee of avoiding hallucinations. However, in practice, it is possible that all actions in $\mathcal{A}^N_t$ are hallucinated, in which case the lookahead minimization (\ref{eq:rollout}) will not help. We quantify the probability of this event in the next subsection.
\subsection{Upper Bound on the Hallucination Probability}
To complement the above condition for filtering hallucinations, we now analyze the \textit{hallucination probability}. The main difficulty in this analysis is that the \textsc{llm}'s propensity to hallucinate is not known a priori. For this reason, we base our analysis on empirical observations of its behavior.

To obtain such empirical observations, we start by using the \textsc{llm} to generate $L$ sample actions. We then verify how many of those actions are hallucinated to estimate the \textsc{llm}'s hallucination probability $h$. We denote this estimate by $\overline{h}$. Due to sampling variability, this estimate may differ substantially from the probability $h$. To address this possibility, we establish a bound that quantifies how likely it is for the estimate $\overline{h}$ to deviate from the hallucination probability $h$ by more than a given threshold $\epsilon$, as stated in the following proposition.
\begin{proposition}\label{cor:hallucination_bound}
Let $h$ denote the true (but unknown) hallucination probability of the \textsc{llm} and let $\overline{h}$ denote the empirical probability based on $L$ samples. We have
\begin{align*}
P(h \geq \overline{h} + \epsilon) \leq e^{-2\epsilon^2L},
\end{align*}
where $\epsilon > 0$ is a configurable parameter.  
\end{proposition}
\begin{proof}
We model the process of generating $L$ actions and verifying which of them are hallucinated as $L$ independent and identically distributed Bernoulli trials, represented by the random variables $X_1,X_2,\hdots,X_L$. We have $X_i=1$ if the $i$th sampled action is hallucinated; $X_i=0$ otherwise. Hence, $\overline{h} = \frac{1}{L}\sum_{i=1}^LX_i$. Applying Hoeffding's inequality, we have
\begin{align*}
P(h \geq  \overline{h}+ \epsilon) \leq e^{-2\epsilon^2L}.
\end{align*}
\end{proof}
This proposition implies that the probability that all actions in the set $\mathcal{A}^N_t$ [cf.~(\ref{eq:rollout})] are hallucinated (i.e., $h^N$) can be controlled with a certain confidence when the conditions of \propref{theorem:hallucination_bound} hold by increasing $N$. Moreover, the confidence increases exponentially with the number of samples $L$ used to estimate the hallucination probability, as shown in \figref{fig:bound_2}.
\begin{figure}[H]
  \centering
  \scalebox{0.77}{
   \input{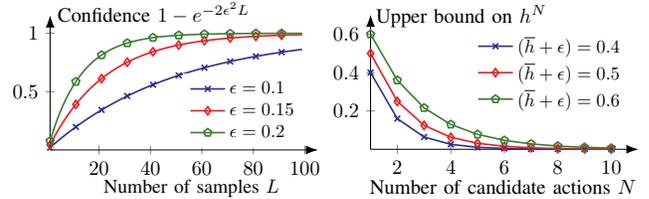}        
  }
  \caption{Illustration of Prop. \ref{cor:hallucination_bound}. Here $L$ is the number of samples for estimating the hallucination probability and $N$ is the number of candidate actions; cf.~(\ref{eq:rollout}).}  
  \label{fig:bound_2}
\end{figure}

\begin{table*}
  \centering
  \scalebox{0.78}{
    \begin{tabular}{llll} \toprule
\rowcolor{lightgray}
      {\textit{Dataset}} & {\textit{System}} & {\textit{Attacks}} & {\textit{Logs}}\\ \midrule
      CTU-Malware-2014 \cite{GARCIA2014100} & \textsc{windows xp sp2} servers & Various malwares and ransomwares, e.g., \textsc{cryptodefense} \cite{8418627}. & \textsc{snort} alerts \cite{snort} \\
      CIC-IDS-2017 \cite{icissp18} & \textsc{windows} and \textsc{linux} servers & Denial-of-service, web attacks, heartbleed, \textsc{sql} injection, etc. & \textsc{snort} alerts \cite{snort} \\
      AIT-IDS-V2-2022 \cite{ait_ids_1} & \textsc{linux} and \textsc{windows} servers/hosts & Multi-stage attack with reconnaissance, cracking, and escalation. & \textsc{wazuh} alerts \cite{wazuh} \\
      CSLE-IDS-2024 \cite{dsn24_hammar_stadler} & \textsc{linux} servers & \textsc{sambacry}, \textsc{shellshock}, exploit of \textsc{cve-2015-1427}, etc. & \textsc{snort} alerts \cite{snort}\\
    \bottomrule\\
  \end{tabular}}
  \caption{The datasets of cyberattacks and logs used for the experimental evaluation.}\label{tab:dataset_types}
\end{table*}
\section{Summary of Our Method}
In summary, our method for using an \textsc{llm} as decision support during incident handling consists of three main steps:
\begin{enumerate}
\item \textit{Offline instruction fine-tuning of a lightweight \textsc{llm}}.
  \begin{itemize}
  \item We fine-tune the \textsc{llm} via supervised learning on a dataset of logs from $68,000$ incidents paired with response plans and chain-of-thought reasoning steps.
  \end{itemize}
\item \textit{Online information retrieval}.
  \begin{itemize}
  \item Before prompting the \textsc{llm} with system logs to generate candidate response actions, we enrich the logs with threat intelligence retrieved from external sources.
  \end{itemize}
\item \textit{Online lookahead planning via \myalgref{alg:our_method}}.
  \begin{itemize}
  \item Instead of directly executing the action generated by the fine-tuned \textsc{llm}, we generate several candidate actions and select the one that leads to the shortest predicted recovery time, which reduces the probability of hallucinations under certain conditions; cf.~\propref{theorem:hallucination_bound}.
  \end{itemize}  
\end{enumerate}

\section{Experimental Evaluation of Our Method}
In this section, we present an experimental evaluation of our method. We start by comparing its performance with that of frontier \textsc{llm}s based on log data from incidents reported in the literature. We then compare the performance of our method with that of a popular reinforcement learning method, namely proximal policy optimization (\textsc{ppo}) \cite{ppo}. Our main evaluation metric is the recovery time $T$, as defined in \defref{def:recovery_time}\footnote{To penalize unnecessary responses in the evaluation, we increment the recovery time by two instead of one if an action includes unnecessary steps.}.

We instantiate our method with the \textsc{deepseek-r1-14b} \textsc{llm} \cite{deepseekai2025deepseekr1incentivizingreasoningcapability}, which we fine-tune using the procedure described in \sectionref{sec:fine_tuning}. Further, we implement the \textsc{rag} pipeline described in \sectionref{section:rag} using the open threat exchange (\textsc{otx}) \textsc{api} \cite{alienvaultOTX}. Finally, we instantiate the planning procedure described in \myalgref{alg:our_method} with $N=3$ candidate actions and $M=3$ samples; cf. (\ref{eq:rollout}). Additional experimental details and hyperparameters are provided in Appendices \ref{appendix:hyperparameters}--\ref{app:example_dataset}. We provide source code, model parameters, and a video demonstration of our method in \cite{llm_source_kim}.
\subsection{Comparison with Frontier \textsc{llm}s}
We compare our method with three frontier \textsc{llm}s: \textsc{deepseek-r1} \cite{deepseekai2025deepseekr1incentivizingreasoningcapability}, \textsc{gemini 2.5} \cite{comanici2025gemini25pushingfrontier}, and \textsc{openai o3} \cite{openai2024gpt4technicalreport}. Compared to these models, the main difference is that our method is significantly more lightweight; see \tableref{tab:parameter_counts}.
\begin{table}[H]
  \centering
  \scalebox{0.7}{
    \begin{tabular}{lll} \toprule
\rowcolor{lightgray}
      {\textit{Method}} & {\textit{Number of parameters}} & {\textit{Context window size}}  \\ \midrule
      \rowcolor{lightblue}
      \textsc{our method} & $14$ billion & $128,000$  \\
      \textsc{deepseek-r1} \cite{deepseekai2025deepseekr1incentivizingreasoningcapability} & $671$ billion \cite {deepseekai2025deepseekr1incentivizingreasoningcapability} & $128,000$\\      
      \textsc{gemini 2.5} \cite{comanici2025gemini25pushingfrontier} & unknown ($\geq 100$ billion)  & $1$ million\\
      \textsc{openai o3} \cite{openai2024gpt4technicalreport} & unknown ($\geq 100$ billion) & $200,000$\\
    \bottomrule\\
  \end{tabular}}
  \caption{Comparison between our method and frontier \textsc{llm}s in terms of the number of model parameters and context window size.}\label{tab:parameter_counts}
\end{table}

\vspace{2mm}
\noindent\textit{\textbf{Evaluation datasets.}}
The evaluation is based on log data from $25$ incidents across $4$ different datasets published in the literature, namely CTU-Malware-2014 \cite{GARCIA2014100}, CIC-IDS-2017 \cite{icissp18}, AIT-IDS-V2-2022 \cite{ait_ids_1}, and CSLE-IDS-2024 \cite{dsn24_hammar_stadler}; see \tableref{tab:dataset_types} and \figref{fig:eval_dataset}. We also include $5$ false-positive incidents. Each incident contains log data and a brief system description. Given this data, the task of the \textsc{llm} is to generate effective response actions, which we compare against the ground truth. The prompt templates and the datasets are available in \cite{llm_source_kim}.

We provide a (condensed) example of an incident and the ground truth response from the CTU-Malware-2014 dataset \cite{GARCIA2014100} below. In this example, the (ground truth) response plan consists of $T=6$ response actions. Hence, the shortest possible recovery time an \textsc{llm} can achieve when evaluated on this example is $6$. However, if the \textsc{llm} generates a plan that includes unnecessary response actions, then the recovery time will be longer than $6$. It is also possible that the generated actions fail to fully recover the system from the incident. We report such cases separately in the evaluation results.

\begin{tcolorbox}[colback=lightgray!20, colframe=lightgray!30, sharp corners, boxrule=0pt, title=\small \textcolor{black}{\underline{Example incident from the CTU-Malware-2014 dataset \cite{GARCIA2014100}}.}]
\footnotesize  \textbf{System description (condensed)}: Two subnetworks (A and B) are connected via a switch that is also connected to the Internet. All servers run \textsc{windows xp sp2}. Their \textsc{ip}s and configurations are...\\
  
\footnotesize \textbf{Snort alert logs (condensed)}:
\vspace{-0.2cm}
\begin{verbatim}
[120:3:2] (http_inspect) NO CONTENT-LENGTH..
[1:31033:6] MALWARE Win.Trojan.Cryptodefence..
{TCP} 147.32.84.165:1057 -> 222.88.205.195:443
[129:5:1] Bad segment, adjusted size..
[139:1:1] (spp_sdf) SDF..
\end{verbatim}
\footnotesize \textbf{Incident summary}: Server 147.32.84.165 is infected with the \textsc{win.trojan.cryptodefence} ransomware. Alerts show the server is making outbound command and control (\textsc{c2}) connections to 222.88.205.195. This indicates that the ransomware is active and may be preparing to encrypt files or has already begun doing so.\\

\textbf{Response actions (condensed)}:

\footnotesize 1. Disconnect the Ethernet cable of the infected server at 147.32.84.165 to sever its network connection. Concurrently, configure a rule on the main switch/firewall to block all outbound traffic to the \textsc{c2} server 222.88.205.195.\\

2. Analyze the central switch to scan all network traffic from both subnetworks A and B for any other hosts attempting to make connections to the malicious \textsc{ip} 222.88.205.195.\\

3. Before altering the infected server, create a complete bit-for-bit forensic image of its hard drive. This preserves the ransomware executable, encrypted files, and other evidence for future analysis.\\

4. Wipe the hard drive of 147.32.84.165. If other infected machines were discovered, they must also be taken offline and wiped.\\

5. Upgrade all servers from \textsc{windows xp sp2} (which is obsolete) to a modern operating system that receives security patches.\\

6. Restore the server's data from a trusted backup. Once the server is rebuilt with a modern operating system, reconnect it to the network and closely monitor for any anomalous activity.
\normalfont
\end{tcolorbox}

\begin{figure}
  \centering
  \scalebox{0.9}{
    \begin{tikzpicture}
  \node[scale=1] (box) at (0,5.4) {
  \begin{tikzpicture}
    \definecolor{grad1}{RGB}{0,150,150}
    \definecolor{grad2}{RGB}{0,100,200}
    \definecolor{grad3}{RGB}{80,0,150}
    \definecolor{grad4}{RGB}{255,100,100}

    \node[scale=0.8] (box) at (0,0) {
    \begin{tikzpicture}

    \node[inner sep=0pt,align=center,scale=1,color=black] at (-0.7,  0.00) {\textsc{impact}};
    \draw[fill=bluethree,draw=black,line width=0.23mm] (0,-0.15) rectangle (7.5, 0.15);
    \node[inner sep=0pt,align=center,scale=0.9,color=black]  at (3.75,  0.00) {5};

    \draw[fill=bluethree,draw=black,line width=0.23mm] (0,-0.25) rectangle (6,-0.55);
    \node[inner sep=0pt,align=center,scale=0.9,color=black]  at (3,-0.40) {4};
    \node[inner sep=0pt,align=center,scale=1,color=black]     at (-1.38,-0.40) {\textsc{initial access}};

    \draw[fill=bluethree,draw=black,line width=0.23mm]            (0,-0.65) rectangle (6,-0.95);
    \node[inner sep=0pt,align=center,scale=0.9,color=black] at (3,-0.80) {4};
    \node[inner sep=0pt,align=center,scale=1,color=black]   at (-2.08,-0.80) {\textsc{command and control}};

    \draw[fill=bluethree,draw=black,line width=0.23mm]            (0,-1.05) rectangle (4.5,-1.35);
    \node[inner sep=0pt,align=center,scale=0.9,color=black] at (2.25,-1.20) {3};
    \node[inner sep=0pt,align=center,scale=1,color=black]   at (-1.07,-1.20) {\textsc{execution}};

    \draw[fill=bluethree,draw=black,line width=0.23mm] (0,-1.45) rectangle (4.5,-1.75);
    \node[inner sep=0pt,align=center,scale=0.9,color=black]   at (2.25,-1.60) {3};
    \node[inner sep=0pt,align=center,scale=1,color=black]     at (-1.15,-1.60) {\textsc{collection}};

    \draw[fill=bluethree,draw=black,line width=0.23mm] (0,-1.85) rectangle (4.5,-2.15);
    \node[inner sep=0pt,align=center,scale=0.9,color=black]   at (2.25,-2.00) {3};
    \node[inner sep=0pt,align=center,scale=1,color=black]     at (-1.78,-2.00) {\textsc{lateral movement}};

    \draw[fill=bluethree,draw=black,line width=0.23mm] (0,-2.25) rectangle (3,-2.55);
    \node[inner sep=0pt,align=center,scale=0.9,color=black]   at (1.5,-2.40) {2};
    \node[inner sep=0pt,align=center,scale=1,color=black]     at (-1.95,-2.40) {\textsc{privilege escalation}};

    \draw[fill=bluethree,draw=black,line width=0.23mm] (0,-2.65) rectangle (3,-2.95);
    \node[inner sep=0pt,align=center,scale=0.9,color=black]   at (1.5,-2.80) {2};
    \node[inner sep=0pt,align=center,scale=1,color=black]     at (-1.25,-2.80) {\textsc{exfiltration}};

    \draw[fill=bluethree,draw=black,line width=0.23mm] (0,-3.05) rectangle (1.5,-3.35);
    \node[inner sep=0pt,align=center,scale=0.9,color=black]   at (0.75,-3.2) {1};
    \node[inner sep=0pt,align=center,scale=1,color=black]     at (-1.5,-3.2) {\textsc{reconnaissance}};

    \end{tikzpicture}
    };
  \end{tikzpicture}
  };
\end{tikzpicture}
  }
  \caption{Number of occurrences of different \textsc{mitre att\&ck tactics} \cite{strom2018mitre} among the incidents in the evaluation datasets.}
  \label{fig:eval_dataset}
\end{figure}
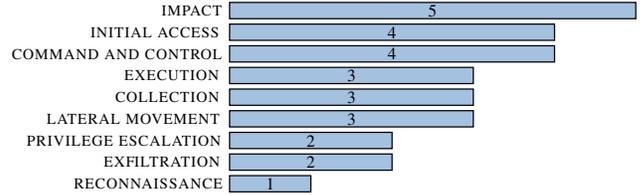

\vspace{2mm}
\noindent\textit{\textbf{Evaluation results.}}
The results are summarized in \tableref{tab:evaluation_1} and \figref{fig:eval_bars_1}. Across all evaluation datasets, our method achieves the shortest average recovery time. On average, the recovery time of our method is $13.46$ compared to $16.21$ for the next best method. Among the frontier \textsc{llm}s, we observe that \textsc{gemini 2.5} performs best on average, whereas the difference between \textsc{openai o3} and \textsc{deepseek-r1} is not statistically significant.

\begin{figure}
  \centering
  \scalebox{0.81}{
   \input{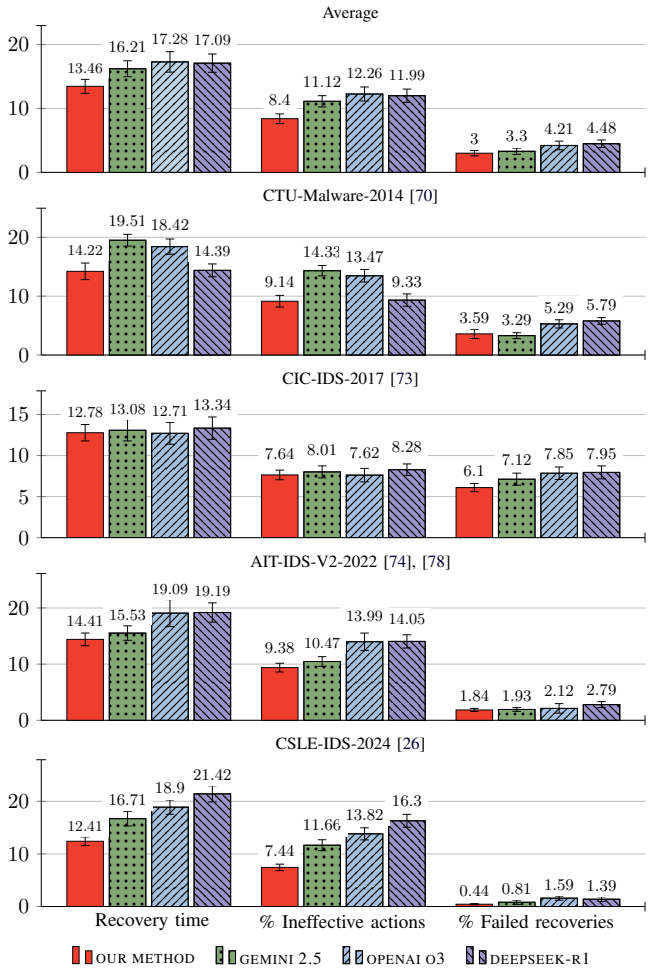}            
  }
  \caption{Evaluation results ($\downarrow$ better): comparison between our method and frontier \textsc{llm}s. Bar colors relate to different methods; bar groups indicate performance metrics; numbers and error bars indicate the mean and the standard deviation from $5$ evaluations with different random seeds.}
  \label{fig:eval_bars_1}
\end{figure}

%
\begin{table}
  \centering
  \scalebox{0.7}{
    \begin{tabular}{llll} \toprule
\rowcolor{lightgray}
      \textit{Method} & \textit{Recovery time} & \textit{\% Ineffective actions} & \textit{\% Failed recovery} \\ \midrule
      \multicolumn{4}{c}{\textbf{Average}} \\\midrule
      \rowcolor{lightblue}
      \textsc{our method} & $\bm{13.46 \pm 1.09}$ & $\bm{8.40 \pm 0.75}$ & $\bm{3.00 \pm 0.41}$\\
      \textsc{gemini 2.5} \cite{comanici2025gemini25pushingfrontier} & $16.21 \pm 1.25$ & $11.12 \pm 0.88$ & $3.30 \pm 0.49$\\
      \textsc{openai o3} \cite{openai2024gpt4technicalreport} & $17.28 \pm 1.60$ & $12.26 \pm 1.10$ & $4.21 \pm 0.68$\\
      \textsc{deepseek-r1} \cite{deepseekai2025deepseekr1incentivizingreasoningcapability} & $17.09 \pm 1.43$ & $11.99 \pm 1.04$ & $4.48 \pm 0.59$\\
      & \\
      \multicolumn{4}{c}{\textbf{CTU-Malware-2014 \cite{GARCIA2014100}}} \\\midrule
      \rowcolor{lightblue}
      \textsc{our method} & $\bm{14.22 \pm 1.41}$ & $\bm{9.14 \pm 0.99}$ & $3.59 \pm 0.78$\\
      \textsc{gemini 2.5} \cite{comanici2025gemini25pushingfrontier} & $19.51 \pm 1.00$ & $14.33 \pm 0.88$ & $\bm{3.29 \pm 0.50}$\\
      \textsc{openai o3} \cite{openai2024gpt4technicalreport} & $18.42 \pm 1.30$ & $13.47 \pm 1.07$ & $5.29 \pm 0.70$\\
      \textsc{deepseek-r1} \cite{deepseekai2025deepseekr1incentivizingreasoningcapability} & $14.39 \pm 1.09$ & $9.33 \pm 1.05$ & $5.79 \pm 0.57$\\
      & \\
      \multicolumn{4}{c}{\textbf{CIC-IDS-2017 \cite{icissp18}}} \\\midrule
      \rowcolor{lightblue}
      \textsc{our method} & $12.78 \pm 1.00$ & $7.64 \pm 0.59$ & $\bm{6.11 \pm 0.50}$\\
      \textsc{gemini 2.5} \cite{comanici2025gemini25pushingfrontier} & $13.08 \pm 1.30$ & $8.01 \pm 0.73$ & $7.13 \pm 0.74$\\
      \textsc{openai o3} \cite{openai2024gpt4technicalreport} & $\bm{12.71 \pm 1.33}$ & $\bm{7.62 \pm 0.82}$ & $7.86 \pm 0.76$\\
      \textsc{deepseek-r1} \cite{deepseekai2025deepseekr1incentivizingreasoningcapability} & $13.34 \pm 1.36$ & $8.28 \pm 0.70$ & $7.95 \pm 0.79$\\
                        & \\
      \multicolumn{4}{c}{\textbf{AIT-IDS-V2-2022 \cite{ait_ids_1,9866880}}} \\\midrule
      \rowcolor{lightblue}
      \textsc{our method} & $\bm{14.41 \pm 1.13}$ & $\bm{9.38 \pm 0.78}$ & $\bm{1.84 \pm 0.26}$\\
      \textsc{gemini 2.5} \cite{comanici2025gemini25pushingfrontier} & $15.53 \pm 1.29$ & $10.47 \pm 0.86$ & $1.94 \pm 0.33$\\
      \textsc{openai o3} \cite{openai2024gpt4technicalreport} & $19.09 \pm 2.40$ & $13.99 \pm 1.54$ & $2.12 \pm 0.85$\\
      \textsc{deepseek-r1} \cite{deepseekai2025deepseekr1incentivizingreasoningcapability} & $19.19 \pm 1.72$ & $14.05 \pm 1.18$ & $2.79 \pm 0.54$\\
                        & \\
      \multicolumn{4}{c}{\textbf{CSLE-IDS-2024 \cite{dsn24_hammar_stadler}}} \\\midrule
      \rowcolor{lightblue}
      \textsc{our method} & $\bm{12.41 \pm 0.82}$ & $\bm{7.44 \pm 0.62}$ & $\bm{0.44 \pm 0.11}$\\
      \textsc{gemini 2.5} \cite{comanici2025gemini25pushingfrontier} & $16.71 \pm 1.40$ & $11.66 \pm 1.04$ & $0.81 \pm 0.40$\\
      \textsc{openai o3} \cite{openai2024gpt4technicalreport} & $18.90 \pm 1.36$ & $13.82 \pm 1.15$ & $1.59 \pm 0.39$\\
      \textsc{deepseek-r1} \cite{deepseekai2025deepseekr1incentivizingreasoningcapability} & $21.42 \pm 1.53$ & $16.30 \pm 1.23$ & $1.39 \pm 0.45$\\
    \bottomrule\\
  \end{tabular}}
  \caption{Evaluation results: comparison between our method and frontier \textsc{llm}s. Rows relate to different methods; columns indicate performance metrics ($\downarrow$ better); blue rows relate to our method (see \figref{fig:framework}); the best results are highlighted in bold; numbers indicate the mean and the standard deviation from $5$ evaluations with different random seeds.}
  \label{tab:evaluation_1}
\end{table}

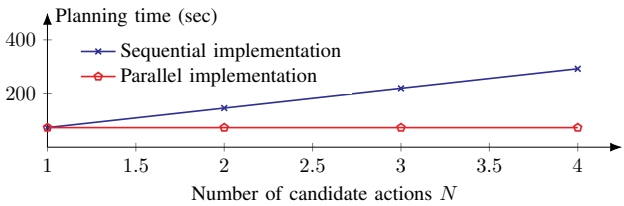
\begin{figure}
  \centering
  \scalebox{0.77}{
   \begin{tikzpicture}

\pgfplotstableread{
1 73
2 146
3 219
4 292
}\sequential

\pgfplotstableread{
1 73
2 73
3 73
4 73
}\parallel

\pgfplotsset{/dummy/workaround/.style={/pgfplots/axis on top}}

\node[scale=1] (kth_cr) at (0,0)
{
\begin{tikzpicture}
  \begin{axis}
[
        xmin=1,
        xmax=4.25,
        width=11.5cm,
        height=3.9cm,
        ymax=500,
        ymin=0,
        axis y line=center,
        axis x line=bottom,
        scaled y ticks=false,
        yticklabel style={
        /pgf/number format/fixed,
        /pgf/number format/precision=5
      },
        xlabel style={below right},
        ylabel style={above left},
        axis line style={-{Latex[length=2mm]}},
        legend style={at={(0.53,0.84)}},
        legend columns=1,
        legend style={
          draw=none,
      anchor=north east,
      /tikz/every even column/.style={anchor=west, column sep=5pt},
      /tikz/every odd column/.style={anchor=west},          
            /tikz/column 2/.style={
                column sep=5pt,
              }
              }
              ]
            \addplot[Blue,name path=l1, thick, mark=x, mark repeat=1] table [x index=0, y index=1, domain=0:1] {\sequential};
            \addplot[Red,name path=l1, thick, mark=pentagon, mark repeat=1] table [x index=0, y index=1, domain=0:1] {\parallel};
            \legend{Sequential implementation, Parallel implementation}
            \end{axis}
\node[inner sep=0pt,align=center, scale=1, rotate=0, opacity=1] (obs) at (1.6,2.25)
{
Planning time (sec)
};
\node[inner sep=0pt,align=center, scale=1, rotate=0, opacity=1] (obs) at (4.85,-0.8)
{
  Number of candidate actions $N$
};
\end{tikzpicture}
};

\end{tikzpicture}        
  }
  \caption{Time required (per time step) to execute \myalgref{alg:our_method} for varying number of candidate actions $N$. The average planning times were computed based on $5$ executions with \textsc{rtx 8000} \textsc{gpu}s.}
  \label{fig:planning_scale}
\end{figure}

\vspace{2mm}

\noindent\textit{\textbf{Scalability analysis.}} \Figref{fig:planning_scale} shows the compute time per time step of \myalgref{alg:our_method} for varying number of candidate actions $N$. We observe that the planning time increases linearly with $N$ when the actions are evaluated sequentially. However, by parallelizing the computation across multiple \textsc{gpu}s, the planning time remains nearly constant as $N$ increases.

\vspace{2mm}

\noindent\textit{\textbf{Hallucination analysis.}} \Figref{fig:empirical_hallucination} shows the empirical hallucination probability of our method based on $L=30$ response actions sampled from the \textsc{llm} when prompted with log data from the evaluation datasets. The figure also shows the theoretical upper bound expressed in \propref{cor:hallucination_bound}.
\begin{figure}[H]
  \centering
  \scalebox{0.77}{
   \begin{tikzpicture}

\pgfplotstableread{
1 0.14
2 0.1
3 0.06
4 0.0
}\empirical

\pgfplotstableread{
1 0.4228725924844031
2 0.17395886667609692
3 0.07253435695589375
4 0.030249866009240515
}\bound

\pgfplotsset{/dummy/workaround/.style={/pgfplots/axis on top}}

\node[scale=1] (kth_cr) at (0,0)
{
\begin{tikzpicture}
  \begin{axis}
[
        xmin=1,
        xmax=4.3,
        width=11.5cm,
        height=3.9cm,
        ymax=0.5,
        ymin=0,
        axis y line=center,
        axis x line=bottom,
        scaled y ticks=false,
        yticklabel style={
        /pgf/number format/fixed,
        /pgf/number format/precision=5
      },
        xlabel style={below right},
        ylabel style={above left},
        axis line style={-{Latex[length=2mm]}},
        legend style={at={(1,0.88)}},
        legend columns=1,
        legend style={
          draw=none,
      anchor=north east,
      /tikz/every even column/.style={anchor=west, column sep=5pt},
      /tikz/every odd column/.style={anchor=west},          
            /tikz/column 2/.style={
                column sep=5pt,
              }
              }
              ]
            \addplot[Blue,name path=l1, thick, mark=x, mark repeat=1] table [x index=0, y index=1, domain=0:1] {\empirical};
            \addplot[Red,name path=l1, thick, mark=pentagon, mark repeat=1] table [x index=0, y index=1, domain=0:1] {\bound};
            \legend{Empirical ($L=30$ samples), Theoretical upper bound (confidence $0.99$)}
            \end{axis}
\node[inner sep=0pt,align=center, scale=1, rotate=0, opacity=1] (obs) at (2,2.3)
{
Hallucination probability 
};
\node[inner sep=0pt,align=center, scale=1, rotate=0, opacity=1] (obs) at (4.85,-0.65)
{
  Number of candidate actions $N$
};
\end{tikzpicture}
};
\end{tikzpicture}    
  }
  \caption{The empirical hallucination probability of our method for varying number of candidate actions $N$, as well as the theoretical upper bound on the hallucination probability $h^N$ (assuming the conditions of \propref{theorem:hallucination_bound} hold) with confidence $0.99$, i.e., the right-hand side of the bound in \propref{cor:hallucination_bound} is $0.01$.}
  \label{fig:empirical_hallucination}
\end{figure}
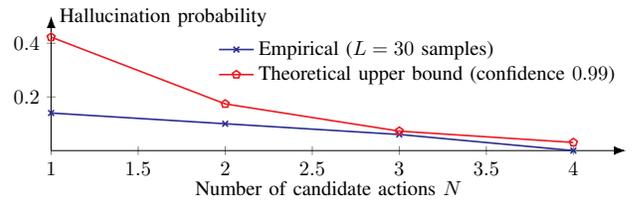

We observe in \figref{fig:empirical_hallucination} that the theoretical bound holds uniformly over the empirical probabilities. However, while the theoretical bound decreases exponentially with the number of candidate actions $N$, the empirical hallucination probability exhibits an approximately linear decline. This discrepancy suggests that the conditions of \propref{theorem:hallucination_bound} are not fully satisfied. This is expected, as we used only $M=3$ samples to estimate the expected recovery times. We chose this value of $M$ to ensure that the planning could be completed within a reasonable time using our limited hardware (\textsc{quadro rtx 8000} \textsc{gpus}). Increasing $M$ is likely to yield more accurate estimates and further reduce the probability of hallucinations.

\vspace{2mm}
\noindent\textit{\textbf{Ablation study.}}
To evaluate the relative importance of each step of our method (i.e., fine-tuning, \textsc{rag}, and planning), we evaluate our method with and without each step. The results are summarized in \tableref{tab:ablation_study_results} and \figref{fig:eval_bars_2}. We observe consistent performance degradations when each step is removed. The most substantial degradation occurs when fine-tuning is removed, which causes the average recovery time to increase from $13$ to $25$. Planning also has a significant impact. Without planning, the average recovery time jumps from $13$ to $21$. Retrieval augmented generation (\textsc{rag}) contributes as well, though its effects on the performance are more modest.

\begin{figure}[H]
  \centering
  \scalebox{0.81}{
   \input{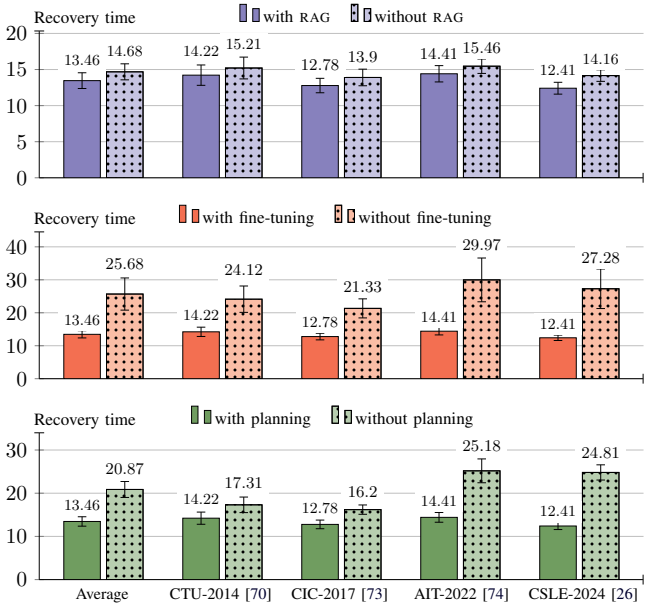}            
  }
  \caption{Ablation-study results for the recovery time metric ($\downarrow$ better). Bars relate to our method with and without different steps; bar groups indicate the evaluation dataset; numbers and error bars indicate the mean and the standard deviation from $5$ evaluations with different random seeds.}
  \label{fig:eval_bars_2}
\end{figure}

\begin{table}
  \centering
  \scalebox{0.7}{
    \begin{tabular}{llll} \toprule
\rowcolor{lightgray}
      \textit{Method} & \textit{Recovery time} & \textit{\% Ineffective actions} & \textit{\% Failed recovery} \\ \midrule
      \multicolumn{4}{c}{\textbf{Average}} \\\midrule
      \rowcolor{lightblue}
      \textsc{our method} & $\bm{13.46 \pm 1.09}$ & $\bm{8.40 \pm 0.75}$ & $\bm{3.00 \pm 0.41}$\\
      without \textsc{rag} & $14.68 \pm 1.11$ & $9.58 \pm 0.96$ & $4.13 \pm 0.44$\\
      without fine-tuning & $25.68 \pm 4.88$ & $20.48 \pm 3.51$ & $12.64 \pm 2.42$\\
      without planning & $20.87 \pm 1.85$ & $15.91 \pm 1.79$ & $7.75 \pm 0.79$\\
      & \\
      \multicolumn{4}{c}{\textbf{CTU-Malware-2014 \cite{GARCIA2014100}}} \\\midrule
      \rowcolor{lightblue}
      \textsc{our method} & $\bm{14.22 \pm 1.41}$ & $\bm{9.14 \pm 0.99}$ & $\bm{3.59 \pm 0.78}$\\
      without \textsc{rag} & $15.21 \pm 1.51$ & $10.20 \pm 1.27$ & $5.90 \pm 0.84$\\
      without fine-tuning & $24.12 \pm 4.00$ & $19.19 \pm 2.87$ & $12.12 \pm 2.50$\\
      without planning & $17.31 \pm 1.80$ & $12.35 \pm 1.61$ & $10.93 \pm 1.03$\\
      & \\
      \multicolumn{4}{c}{\textbf{CIC-IDS-2017 \cite{icissp18}}} \\\midrule
      \rowcolor{lightblue}
      \textsc{our method} & $\bm{12.78 \pm 1.00}$ & $\bm{7.64 \pm 0.59}$ & $\bm{6.11 \pm 0.50}$\\
      without \textsc{rag} & $13.90 \pm 1.15$ & $8.67 \pm 0.88$ & $7.71 \pm 0.61$\\
      without fine-tuning & $21.33 \pm 2.89$ & $16.08 \pm 2.26$ & $13.87 \pm 1.88$\\
      without planning & $16.20 \pm 1.08$ & $11.64 \pm 1.15$ & $9.86 \pm 0.56$\\
                        & \\
      \multicolumn{4}{c}{\textbf{AIT-IDS-V2-2022 \cite{ait_ids_1,9866880}}} \\\midrule
      \rowcolor{lightblue}
      \textsc{our method} & $\bm{14.41 \pm 1.13}$ & $\bm{9.38 \pm 0.78}$ & $\bm{1.84 \pm 0.26}$\\
      without \textsc{rag} & $15.46 \pm 1.00$ & $10.33 \pm 0.87$ & $2.16 \pm 0.20$\\
      without fine-tuning & $29.97 \pm 6.64$ & $24.52 \pm 4.35$ & $14.98 \pm 2.83$\\
      without planning & $25.18 \pm 2.76$ & $20.01 \pm 2.47$ & $5.39 \pm 1.02$\\
                        & \\
      \multicolumn{4}{c}{\textbf{CSLE-IDS-2024 \cite{dsn24_hammar_stadler}}} \\\midrule
      \rowcolor{lightblue}
      \textsc{our method} & $\bm{12.41 \pm 0.82}$ & $\bm{7.44 \pm 0.62}$ & $\bm{0.44 \pm 0.11}$\\
      without \textsc{rag} & $14.16 \pm 0.80$ & $9.09 \pm 0.84$ & $0.76 \pm 0.10$\\
      without fine-tuning & $27.28 \pm 6.00$ & $22.12 \pm 4.56$ & $10.71 \pm 2.50$\\
      without planning & $24.81 \pm 1.76$ & $19.64 \pm 1.92$ & $4.82 \pm 0.54$\\
    \bottomrule\\
  \end{tabular}}
 \caption{Ablation-study results. Rows relate to different methods; columns indicate performance metrics ($\downarrow$ better); blue rows relate to our method (see \figref{fig:framework}); the best results are highlighted in bold; numbers indicate the mean and the standard deviation from $5$ evaluations with different random seeds.}
  \label{tab:ablation_study_results}
\end{table}
\begin{figure*}
  \centering
  \scalebox{0.77}{
   \input{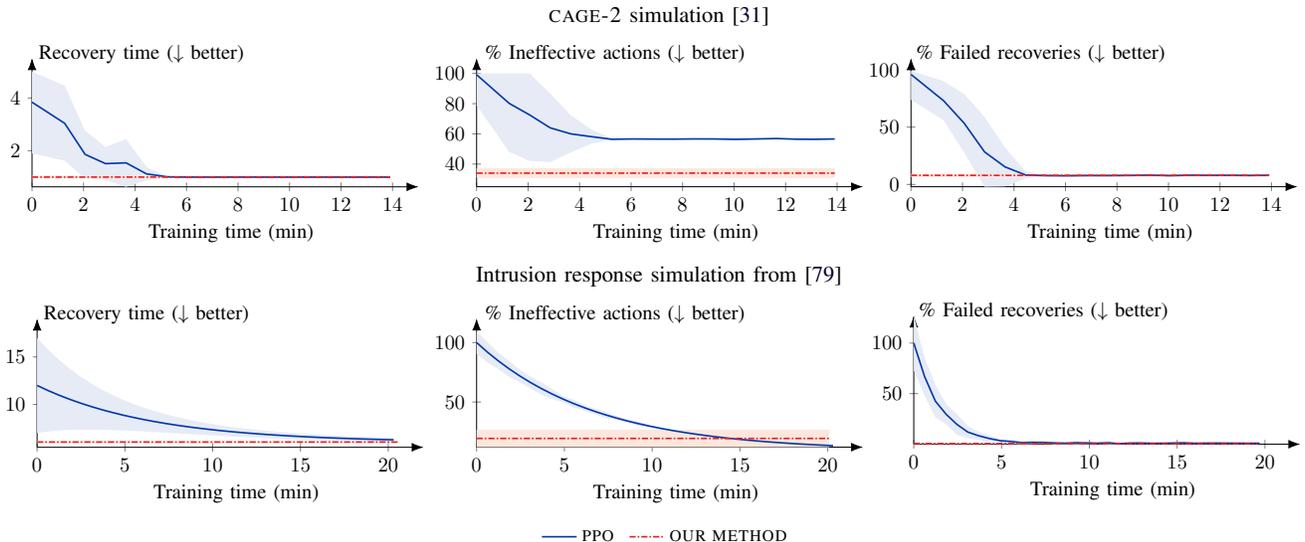}    
  }
  \caption{Comparison between our method (red curves) and the \textsc{ppo} reinforcement learning method (blue curves) \cite{ppo}. The first row of plots relates to the \textsc{cage-2} simulation \cite{cage_challenge_2_announcement} and the second row relates to the network intrusion simulation from \cite{hammar_stadler_tnsm}. Columns relate to different evaluation metrics ($\downarrow$ better). Curves show the mean value from evaluations with $5$ random seeds; shaded areas indicate standard deviations. The x-axes indicate the training time required by \textsc{ppo} for each simulation. In contrast, our method requires no incident-specific training.}
  \label{fig:cage_eval}
\end{figure*}

\subsection{Comparison with Proximal Policy Optimization}
Numerous reinforcement learning approaches have been proposed for incident response, including policy optimization methods \cite{vyas2023automated}, tree search \cite{hammar2024optimaldefenderstrategiescage2}, stochastic approximation \cite{hammar_stadler_tnsm}, and Q-learning \cite{tabular_Q_andy}; see \cite{kim_phd_thesis} for an extensive review of the state of the art. Among these methods, variants of proximal policy optimization (\textsc{ppo}) \cite{ppo} dominate recent work. We therefore use \textsc{ppo} as a representative baseline for comparison. 

\vspace{2mm}

\noindent\textit{\textbf{Experiment setup.}} We evaluate our method against \textsc{ppo} on two simulated incidents: an advanced persistent threat from the \textsc{cage-2} simulation \cite{cage_challenge_2_announcement}, and a network intrusion scenario from \cite{hammar_stadler_tnsm}. The evaluation uses the same metrics as in the comparison with frontier \textsc{llm}s\footnote{To align the \textsc{cage-2} scenario with our evaluation metrics, we exclude decoy-related actions as they target prevention rather than response.}. The hyperparameters of \textsc{ppo} and the prompt templates that we use for the evaluation are available in \appendixref{appendix:hyperparameters} and \cite{llm_source_kim}, respectively.

\vspace{2mm}

\noindent\textit{\textbf{Evaluation results.}} The results are presented in \figref{fig:cage_eval}. Both our method and \textsc{ppo} achieve similar performance in terms of recovery time and failed recoveries across the two simulations. The only notable performance gap is in the percentage of ineffective actions for the \textsc{cage-2} simulation, where our method performs better. The key difference between our method and \textsc{ppo} lies in their training requirements: \textsc{ppo} requires incident-specific training (approximately $10$–$20$ minutes of training per incident) to reach good performance. In contrast, our method does not require such training to achieve good performance.

\subsection{Discussion of the Evaluation Results}
Our experimental results demonstrate a trade-off between generality, computational cost, and deployment practicality. Compared to frontier \textsc{llm}s, our method is significantly more lightweight, i.e., it requires fewer parameters and runs efficiently on commodity hardware, yet it achieves consistently better performance across all evaluation metrics. This performance advantage is primarily driven by our fine-tuning and planning steps, as shown in the ablation study.

When compared to reinforcement learning methods such as \textsc{ppo}, our method is more computationally costly at inference time due to the overhead of planning; see \figref{fig:planning_scale}. However, this overhead is offset by a major advantage: our method requires no incident-specific training. In contrast, \textsc{ppo} must be retrained for each new incident, which is impractical.

Compared to incident response playbooks \cite{playbook_response}, our method provides two clear advantages. First, it does not rely on domain experts for configuration. Second, it generates more precise and context-specific response actions. In particular, current playbooks often prescribe vague actions that are not directly executable \cite{10.1145/3491102.3517559, 10646756}. By contrast, our method produces executable response actions tailored to the system logs.

On the other hand, the main concern of our method compared to playbooks is the risk of hallucination. While our method reduces this risk through fine-tuning, information retrieval, and planning, it does not eliminate it entirely. Hence, response actions generated by our method should be subject to human validation before execution in most cases.

\vspace{2mm}

\noindent\textit{\textbf{Takeaways.}}
In summary, our main experimental findings are:
\begin{itemize}
    \item Our method consistently outperforms frontier \textsc{llm}s across all evaluation metrics, while being significantly more lightweight and able to run on commodity hardware.
    \item Fine-tuning and decision-theoretic planning are key drivers of performance, \textsc{rag} is less important.
    \item Compared to reinforcement learning methods, our method has higher overhead but avoids incident-specific training.
    \item In contrast to response playbooks, our method does not rely on domain experts for configuration and generates more precise and actionable response plans.
\end{itemize}
\section{Conclusion}
We introduce a novel method that enables the effective use of a large language model (\textsc{llm}) to provide decision support for incident response planning. Our method uses the \textsc{llm} for translating system logs into effective response plans while addressing its limitations through fine-tuning, information retrieval, and decision-theoretic planning. We prove that our method produces incident responses with bounded hallucination probability; see \propref{theorem:hallucination_bound} and \propref{cor:hallucination_bound}. Under certain assumptions, this bound can be made arbitrarily small at the expense of increased planning time. We evaluate our method on logs from incidents reported in the literature. The results show that our method a) achieves up to $22$\% shorter recovery times than frontier \textsc{llm}s and b) generalizes to a broad range of incident types and response actions.

\vspace{2mm}
\noindent\textit{\textbf{Future work.}} A primary direction for future work is to conduct evaluations in operational settings, where security operators use our method for decision support. Such studies would provide insights into the practical utility of our method and how to improve it further. From a theoretical standpoint, a promising direction of future work is to tighten the hallucination-probability bound stated in \propref{cor:hallucination_bound}. A possible approach to tighten this bound is to leverage conformal-abstention techniques \cite{yadkori2024mitigatingllmhallucinationsconformal}. Another direction for future work is to extend the system model in \sectionref{sec:decision_planning} to include additional performance metrics beyond recovery time. Moreover, due to the generality of our method, it is possible to extend our planning procedure [cf.~\myalgref{alg:our_method}] in many ways, e.g., by incorporating rollout techniques \cite{bertsekas2021rollout} or tree search \cite{hammar2024optimaldefenderstrategiescage2}. Similarly, the information-retrieval step of our method can be expanded to integrate information from several sources.

\section*{Acknowledgment}
This research is supported by the Swedish Research Council under contract 2024-06436.

\appendices

\section{Proof of \Propref{theorem:hallucination_bound}}\label{app:hallucination_bound_proof}
We start by noting that the planning problem in \sectionref{sec:decision_planning} can be viewed as a stochastic shortest path problem on the graph of recovery states, where the goal is to reach the state $\mathbf{s}_t=(1,1,1,1,1,1)$ as quickly as possible. Consequently, the problem is well-defined under standard assumptions, see e.g., \cite{10.5555/1396348} for details. The main approach for proving \propref{theorem:hallucination_bound} is to bound the difference in estimated recovery time of a non-hallucinated action and a hallucinated action. To this end, we start by stating and proving the following lemma.
\begin{lemma}\label{lemma:one}
Given the conditions of \propref{theorem:hallucination_bound}, we have
\begin{align*}
\norm{\tilde{J} - J}_{\infty} \leq \eta \norm{\tilde{J}}_{\infty}\norm{J}_{\infty},
\end{align*}
where $J$ is the recovery time-to-go function [cf. \defref{def:time_to_go}] and $\tilde{J}$ is the time-to-go function estimated by the \textsc{llm}.
\end{lemma}
\begin{proof}
  We first note that the vocabulary (i.e., the set of tokens) of any \textsc{llm} is finite. Therefore, the set of feasible response actions $\mathcal{A}$ is finite. As a consequence, the state predictions $p_{\bm{\theta}^{\prime}}(\mathbf{s}^{\prime} \mid \mathbf{s}, \mathbf{a}, \mathbf{I})$ define a transition probability matrix between (non-terminal) recovery states. We denote this matrix by $\Tilde{\mathbf{F}}$ and the corresponding matrix of the real system by $\mathbf{F}$, where $\Tilde{\mathbf{F}}_{\mathbf{s}\mathbf{s}^{\prime}}$ denotes the transition probability between the non-terminal states $\mathbf{s}$ and $\mathbf{s}^{\prime}$. Similarly, we use $\mathbf{J}$ and $\Tilde{\mathbf{J}}$ to denote the vectors obtained by applying the functions $J$ and $\Tilde{J}$ to the set of non-terminal recovery states $\tilde{\mathcal{S}}$, i.e., all states for which $\mathbf{s} \neq (1,1,1,1,1,1)$, where $\mathbf{J}_{\mathbf{s}}$ denotes the expected recovery time-to-go from the non-terminal state $\mathbf{s}$.

Since the goal is to minimize the recovery time, we can express the recovery time-to-go function recursively by defining a stage cost of $1$ for each response action taken. Using this formulation of the recovery time-to-go, we have
\begin{equation}
\begin{split}  
\mathbf{J} &= \mathbf{1} + \mathbf{F}\mathbf{J},\\
\Tilde{\mathbf{J}} &= \mathbf{1} + \Tilde{\mathbf{F}}\Tilde{\mathbf{J}},\label{eq:matrix_bellman}
\end{split}
\end{equation}
where $\mathbf{1}$ denotes the vector of all ones. Given these Bellman equations, the difference $\Tilde{\mathbf{J}}-\mathbf{J}$ can be written as
\begin{align*}
  \Tilde{\mathbf{J}}-\mathbf{J} &= \left(\mathbf{1} + \Tilde{\mathbf{F}}\Tilde{\mathbf{J}}\right) - \left(\mathbf{1} + \mathbf{F}\mathbf{J}\right)\\
                                &= \tilde{\mathbf{F}}\tilde{\mathbf{J}}-\mathbf{F}\mathbf{J}\\
                                &=\tilde{\mathbf{F}}\tilde{\mathbf{J}}-\mathbf{F}\mathbf{J} + \tilde{\mathbf{F}}\mathbf{J}- \tilde{\mathbf{F}}\mathbf{J}\\
                              &=\tilde{\mathbf{F}}(\tilde{\mathbf{J}}-\mathbf{J}) + (\tilde{\mathbf{F}}-\mathbf{F})\mathbf{J}.
\end{align*}
Solving for $(\tilde{\mathbf{J}}-\mathbf{J})$ gives
\begin{align}
\tilde{\mathbf{J}}-\mathbf{J} &= \tilde{\mathbf{F}}(\tilde{\mathbf{J}}-\mathbf{J}) + (\tilde{\mathbf{F}}-\mathbf{F})\mathbf{J} \nonumber\\
  \implies (\mathds{1} - \tilde{\mathbf{F}})(\tilde{\mathbf{J}}-\mathbf{J}) &= (\tilde{\mathbf{F}}-\mathbf{F})\mathbf{J} \nonumber\\
\implies \tilde{\mathbf{J}}-\mathbf{J}  &= (\mathds{1} - \tilde{\mathbf{F}})^{-1}(\tilde{\mathbf{F}}-\mathbf{F})\mathbf{J},\label{eq:inverse_pf}
\end{align}
where $\mathds{1}$ denotes the identity matrix.

Since $\norm{J}_{\infty}$ and $\norm{\tilde{J}}_{\infty}$ are assumed finite, $\norm{\mathbf{J}}_{\infty}$ and $\norm{\Tilde{\mathbf{J}}}_{\infty}$ are also finite. As a consequence, the linear systems in (\ref{eq:matrix_bellman}) have unique solutions. Consequently, the inverse in (\ref{eq:inverse_pf}) exists. Taking the supremum norm on both sides of the final expression in (\ref{eq:inverse_pf}), we have
\begin{align}
  \norm{\tilde{\mathbf{J}}-\mathbf{J}}_{\infty}  &= \norm{(\mathds{1} - \tilde{\mathbf{F}})^{-1}(\tilde{\mathbf{F}}-\mathbf{F})\mathbf{J}}_{\infty} \nonumber\\
  &\leq \norm{(\mathds{1} - \tilde{\mathbf{F}})^{-1}}_{\infty}\norm{(\tilde{\mathbf{F}}-\mathbf{F})\mathbf{J}}_{\infty}, \label{eq:mult_bound}
\end{align}
where the last inequality follows from the sub-multiplicative property of the supremum norm. Hence, it suffices to show that the right-hand side in (\ref{eq:mult_bound}) is bounded by $\eta \norm{\tilde{J}}_{\infty}\norm{J}_{\infty}$. We start by showing that $\norm{(\mathds{1} - \tilde{\mathbf{F}})^{-1}}_{\infty}=\norm{\tilde{\mathbf{J}}}_{\infty}$. In view of (\ref{eq:matrix_bellman}), we have
\begin{align}
\tilde{\mathbf{J}} &= \mathbf{1} + \tilde{\mathbf{F}}\tilde{\mathbf{J}}\nonumber\\
\implies  (\mathds{1}-\tilde{\mathbf{F}})\tilde{\mathbf{J}} &= \mathbf{1}\nonumber\\
\implies   \tilde{\mathbf{J}}&=  (\mathds{1}-\tilde{\mathbf{F}})^{-1}\mathbf{1} \nonumber\\
  \implies   \norm{\tilde{\mathbf{J}}}_{\infty}&=  \norm{(\mathds{1}-\tilde{\mathbf{F}})^{-1}\mathbf{1}}_{\infty}. \nonumber
\end{align}
Since $\norm{\tilde{J}}_{\infty}$ is assumed finite, the expected time to reach the terminal state $\mathbf{s}=(1,1,1,1,1,1)$ from any non-terminal state $\mathbf{s} \in \tilde{\mathcal{S}}$ is finite. As a consequence, the spectral radius of the transition matrix between the non-terminal states, i.e., $\tilde{\mathbf{F}}$, must be strictly less than $1$. Therefore, we can expand $(\mathds{1}-\tilde{\mathbf{F}})^{-1}$ using the Neumann series representation as
\begin{align*}
(\mathds{1}-\tilde{\mathbf{F}})^{-1} = \sum_{k=0}^{\infty}\tilde{\mathbf{F}}^k.
\end{align*}
Because the matrix $\tilde{\mathbf{F}}$ is non-negative, all of its powers are also non-negative. As a consequence, $(\mathds{1}-\tilde{\mathbf{F}})^{-1}$ is non-negative. For any non-negative matrix $\mathbf{A}$, we have $\norm{\mathbf{A}}_{\infty}=\norm{\mathbf{A}\mathbf{1}}_{\infty}$. Consequently, we obtain
\begin{align}
\norm{\tilde{\mathbf{J}}}_{\infty}&=  \norm{(\mathds{1}-\tilde{\mathbf{F}})^{-1}\mathbf{1}}_{\infty}\nonumber\\ 
 &=\norm{(\mathds{1}-\tilde{\mathbf{F}})^{-1}}_{\infty}. \label{eq:pff_bound_1}
\end{align}
Now consider the second factor in the right-hand side of (\ref{eq:mult_bound}), i.e., $\norm{(\tilde{\mathbf{F}}-\mathbf{F})\mathbf{J}}_{\infty}$. Fix any recovery state $\mathbf{s}$. We have
\begin{align*}
  \left|\bigg(\left(\tilde{\mathbf{F}} - \mathbf{F}\right)\mathbf{J}\bigg)_{\mathbf{s}}\right| &= \left|\sum_{\mathbf{s}^{\prime} \in \tilde{\mathcal{S}}} \left(\tilde{\mathbf{F}}_{\mathbf{s}\mathbf{s}^{\prime}} - \mathbf{F}_{\mathbf{s}\mathbf{s}^{\prime}}\right)\mathbf{J}_{\mathbf{s}^{\prime}}\right|\\
  &\numleq{a} \sum_{\mathbf{s}^{\prime} \in \tilde{\mathcal{S}}} \left|\tilde{\mathbf{F}}_{\mathbf{s}\mathbf{s}^{\prime}} - \mathbf{F}_{\mathbf{s}\mathbf{s}^{\prime}}\right|\cdot |\mathbf{J}_{\mathbf{s}^{\prime}}|\\
  &\leq \left(\sum_{\mathbf{s}^{\prime} \in \tilde{\mathcal{S}}} \left|\tilde{\mathbf{F}}_{\mathbf{s}\mathbf{s}^{\prime}} - \mathbf{F}_{\mathbf{s}\mathbf{s}^{\prime}}\right|\right)\norm{\mathbf{J}}_{\infty}\\
&\leq \eta \norm{\mathbf{J}}_{\infty},
\end{align*}
where we use the triangle inequality to move the absolute value inside the sum and then the fact that $|ab|=|a||b|$ to obtain (a). Since this bound holds for any state $\mathbf{s}$, we have
\begin{align*}
\norm{(\tilde{\mathbf{F}}-\mathbf{F})\mathbf{J}}_{\infty} \leq \eta \norm{\mathbf{J}}_{\infty}.
\end{align*}
Substituting this bound and (\ref{eq:pff_bound_1}) into (\ref{eq:mult_bound}) yields
\begin{align*}
  \norm{\tilde{\mathbf{J}}-\mathbf{J}}_{\infty} &\leq \norm{(\mathds{1} - \tilde{\mathbf{F}})^{-1}}_{\infty}\norm{(\tilde{\mathbf{F}}-\mathbf{F})\mathbf{J}}_{\infty}\\
                                                &\leq \eta \norm{\tilde{\mathbf{J}}}_{\infty}\norm{\mathbf{J}}_{\infty}.\\
\end{align*}
Since the recovery time-to-go from the terminal state is $0$ [cf. \defref{def:time_to_go}], we have $\norm{\tilde{\mathbf{J}}-\mathbf{J}}_{\infty}=\norm{\tilde{J}-J}_{\infty}$, $\norm{\tilde{\mathbf{J}}}_{\infty}=\norm{\tilde{J}}_{\infty}$, and $\norm{\mathbf{J}}_{\infty}=\norm{J}_{\infty}$. The proof is thus complete.
\end{proof}
Given \lemmaref{lemma:one}, we are now ready to derive the proof of \propref{theorem:hallucination_bound}. The event that a hallucinated action $\hat{\mathbf{a}}$ is selected over a non-hallucinated action $\tilde{\mathbf{a}}$ in (\ref{eq:rollout}) implies
\begin{align*}
\tilde{Q}(\tilde{\mathbf{s}}, \hat{\mathbf{a}}) \leq \tilde{Q}(\tilde{\mathbf{s}}, \tilde{\mathbf{a}}).
\end{align*}
To show that this inequality cannot hold under the proposition's assumptions, we start by bounding the difference between $\tilde{Q}$ and $Q$, where $Q(\mathbf{s}, \mathbf{a})$ is the true expected recovery time-to-go when taking response action $\mathbf{a}$ in state $\mathbf{s}$ and $\Tilde{Q}(\mathbf{s}, \mathbf{a})$ is the \textsc{llm}'s estimate. We have
\begin{align*}
  &|\tilde{Q}(\mathbf{s}, \mathbf{a})-Q(\mathbf{s}, \mathbf{a})| = \\
  &\bigg\vert \sum_{\mathbf{s}^{\prime} \in \tilde{\mathcal{S}}}p_{\bm{\theta}^{\prime}}(\mathbf{s}^{\prime} \mid \mathbf{s}, \mathbf{a},\mathbf{I})\tilde{J}(\mathbf{s}^{\prime})  - 
    \sum_{\mathbf{s}^{\prime}\in \tilde{\mathcal{S}}}P(\mathbf{s}^{\prime} \mid \mathbf{s}, \mathbf{a}, \mathbf{I})J(\mathbf{s}^{\prime}) \bigg\vert\\
&= \bigg\vert \sum_{\mathbf{s}^{\prime}\in \tilde{\mathcal{S}}}p_{\bm{\theta}^{\prime}}(\mathbf{s}^{\prime} \mid \mathbf{s}, \mathbf{a}, \mathbf{I})\tilde{J}(\mathbf{s}^{\prime})  - 
  \sum_{\mathbf{s}^{\prime} \in \tilde{\mathcal{S}}}P(\mathbf{s}^{\prime} \mid \mathbf{s}, \mathbf{a}, \mathbf{I})J(\mathbf{s}^{\prime}) + \\
& \quad\quad \sum_{\mathbf{s}^{\prime} \in \tilde{\mathcal{S}}}p_{\bm{\theta}^{\prime}}(\mathbf{s}^{\prime} \mid \mathbf{s}, \mathbf{a}, \mathbf{I})J(\mathbf{s}^{\prime})-\sum_{\mathbf{s}^{\prime} \in \tilde{\mathcal{S}}}p_{\bm{\theta}^{\prime}}(\mathbf{s}^{\prime} \mid \mathbf{s}, \mathbf{a}, \mathbf{I})J(\mathbf{s}^{\prime}) \bigg\vert\\
  &=\bigg\vert \sum_{\mathbf{s}^{\prime} \in \tilde{\mathcal{S}}}p_{\bm{\theta}^{\prime}}(\mathbf{s}^{\prime} \mid \mathbf{s}, \mathbf{a}, \mathbf{I})\left(\tilde{J}(\mathbf{s}^{\prime})- J(\mathbf{s}^{\prime})\right) - \\
  &  \quad\quad\sum_{\mathbf{s}^{\prime} \in \tilde{\mathcal{S}}}\Big(P(\mathbf{s}^{\prime} \mid \mathbf{s}, \mathbf{a}, \mathbf{I}) - p_{\bm{\theta}^{\prime}}(\mathbf{s}^{\prime} \mid \mathbf{s}, \mathbf{a}, \mathbf{I})\Big)J(\mathbf{s}^{\prime})\bigg\vert\\
  &\numleq{a} \bigg\vert\sum_{\mathbf{s}^{\prime} \in \tilde{\mathcal{S}}} p_{\bm{\theta}^{\prime}}(\mathbf{s}^{\prime} \mid \mathbf{s}, \mathbf{a}, \mathbf{I})\left(\tilde{J}(\mathbf{s}^{\prime})- J(\mathbf{s}^{\prime})\right)\bigg\vert + \\
  &  \quad\quad\bigg\vert\sum_{\mathbf{s}^{\prime} \in \tilde{\mathcal{S}}}\Big(P(\mathbf{s}^{\prime} \mid \mathbf{s}, \mathbf{a}, \mathbf{I}) - p_{\bm{\theta}^{\prime}}(\mathbf{s}^{\prime} \mid \mathbf{s}, \mathbf{a}, \mathbf{I})\Big)J(\mathbf{s}^{\prime})\bigg\vert\\
  &\numleq{b} \sum_{\mathbf{s}^{\prime} \in \tilde{\mathcal{S}}}\bigg\vert p_{\bm{\theta}^{\prime}}(\mathbf{s}^{\prime} \mid \mathbf{s}, \mathbf{a}, \mathbf{I})\left(\tilde{J}(\mathbf{s}^{\prime})- J(\mathbf{s}^{\prime})\right)\bigg\vert + \\
  &  \quad\quad\sum_{\mathbf{s}^{\prime} \in \tilde{\mathcal{S}}}\bigg\vert\Big(P(\mathbf{s}^{\prime} \mid \mathbf{s}, \mathbf{a}, \mathbf{I}) - p_{\bm{\theta}^{\prime}}(\mathbf{s}^{\prime} \mid \mathbf{s}, \mathbf{a}, \mathbf{I})\Big)J(\mathbf{s}^{\prime})\bigg\vert\\    
  &\numleq{c} \sum_{\mathbf{s}^{\prime} \in \tilde{\mathcal{S}}}\bigg\vert p_{\bm{\theta}^{\prime}}(\mathbf{s}^{\prime} \mid \mathbf{s}, \mathbf{a}, \mathbf{I})\left(\tilde{J}(\mathbf{s}^{\prime}) - J(\mathbf{s}^{\prime})\right)\bigg\vert + \\
  &\quad\quad\sum_{\mathbf{s}^{\prime} \in \tilde{\mathcal{S}}}\bigg\vert p_{\bm{\theta}^{\prime}}(\mathbf{s}^{\prime} \mid \mathbf{s}, \mathbf{a}, \mathbf{I}) - P(\mathbf{s}^{\prime} \mid \mathbf{s}, \mathbf{a}, \mathbf{I})\bigg\vert \norm{J}_{\infty}\\
&\leq \norm{\tilde{J}-J}_{\infty} + \eta \norm{J}_{\infty}\\
&\numleq{d} \underbrace{\eta \norm{\tilde{J}}_{\infty}\norm{J}_{\infty} + \eta \norm{J}_{\infty}}_{=\Delta},
\end{align*}
where (a)-(b) follow from the triangle inequality; (c) uses the fact that $\norm{\mathbf{a}\mathbf{b}}_{\infty} \leq \norm{\mathbf{a}}_{\infty}\norm{\mathbf{b}}_{\infty}$; and (d) follows from \lemmaref{lemma:one}. This bound implies that
\begin{align}\label{eq:pf_h_bound_399}
&Q(\mathbf{s}, \mathbf{a}) - \Delta \leq \Tilde{Q}(\mathbf{s}, \mathbf{a}) \leq Q(\mathbf{s}, \mathbf{a}) + \Delta.
\end{align}  
Now, if a hallucinated action $\hat{\mathbf{a}}$ is selected over a non-hallucinated action $\tilde{\mathbf{a}}$ in (\ref{eq:rollout}), we must have
\begin{align*}
&\tilde{Q}(\mathbf{s}, \hat{\mathbf{a}}) \leq \tilde{Q}(\mathbf{s}, \tilde{\mathbf{a}}).
\end{align*}
Combining this inequality with (\ref{eq:pf_h_bound_399}), we have
\begin{equation}\label{eq:pf_h_bound_3}
\begin{split}  
 Q(\mathbf{s}, \hat{\mathbf{a}}) - \Delta &\leq \tilde{Q}(\mathbf{s}, \hat{\mathbf{a}}) \\
 &\leq \tilde{Q}(\mathbf{s}, \tilde{\mathbf{a}}) \\
                                                                 &\leq Q(\mathbf{s}, \tilde{\mathbf{a}}) + \Delta\\
\implies Q(\mathbf{s}, \hat{\mathbf{a}})-Q(\mathbf{s}, \tilde{\mathbf{a}}) &\leq 2\Delta.
\end{split}
\end{equation}
Next, since $\hat{\mathbf{a}}$ is hallucinated and $\tilde{\mathbf{a}}$ is not, we have
\begin{align*}
Q(\mathbf{s}, \hat{\mathbf{a}}) &= 1 + \mathbb{E}_{\mathbf{s}^{\prime}}[J(\mathbf{s}^{\prime}) \mid \hat{\mathbf{a}}, \mathbf{s}, \mathbf{I}] = 1+ J(\mathbf{s}),\\
Q(\mathbf{s}, \tilde{\mathbf{a}}) &= 1 + \mathbb{E}_{\mathbf{s}^{\prime}}[J(\mathbf{s}^{\prime}) \mid \tilde{\mathbf{a}}, \mathbf{s}, \mathbf{I}] \leq  1 + J(\mathbf{s}) - \delta.
\end{align*}
Substituting $Q(\mathbf{s}, \hat{\mathbf{a}})=1+ J(\mathbf{s})$ into the inequality, we obtain
\begin{align*}
  Q(\mathbf{s}, \hat{\mathbf{a}}) &\geq Q(\mathbf{s}, \tilde{\mathbf{a}}) + \delta\\
  \implies \delta &\leq Q(\mathbf{s}, \hat{\mathbf{a}}) - Q(\mathbf{s}, \tilde{\mathbf{a}}) \leq 2\Delta\\
                                                         &=2\left(\eta \norm{\tilde{J}}_{\infty}\norm{J}_{\infty} + \eta \norm{J}_{\infty}\right)\\
                                                         &=2\eta\norm{J}_{\infty}\left(\norm{\tilde{J}}_{\infty} + 1\right).
\end{align*}
Since the conditions of the proposition state that
\begin{align*}
\delta > 2\eta\norm{J}_{\infty}\left(\norm{\tilde{J}}_{\infty} + 1\right),
\end{align*}  
we conclude that whenever a non-hallucinated action exists, it will be selected by the minimization (\ref{eq:rollout}). \qed

\section{Notation}\label{appendix:notation}
Our notation is summarized in \tableref{tab:notation}.
\begin{table}
  \centering
  \scalebox{0.72}{
    \begin{tabular}{ll} \toprule
\rowcolor{lightgray}
      {\textit{Notation(s)}} & {\textit{Description}} \\ \midrule
      $\mathbf{a}$ & Response action; cf. \sectionref{sec:decision_planning}.  \\
      $\mathbf{s}, \tilde{\mathbf{s}}$ & Recovery state [cf.~(\ref{eq:state_def})] and predicted state.\\
      $\mathbf{I}$ & Initial information about the incident (e.g., logs).\\      
      $T$ & Recovery time; cf. \sectionref{sec:decision_planning}.  \\
      $p_{\bm{\theta}}, \bm{\theta}$ & the token distribution of an \textsc{llm} and its parameters; cf. (\ref{eq:next_token_pred}).  \\
      $\bm{\theta}^{\prime}$ & Fine-tuned parameter vector of an \textsc{llm}; cf. \sectionref{sec:fine_tuning}.  \\      
      $\mathcal{D}$ & Instruction dataset for fine-tuning; cf. \sectionref{sec:fine_tuning}. \\
      $\mathbf{x}, \mathbf{y}$ & Instruction and desired output; cf. \sectionref{sec:fine_tuning}. \\
      $N$ & Number of candidate actions to evaluate; cf. \sectionref{sec:decision_planning}. \\
      $M$ & Number of samples to estimate expected values in \myalgref{alg:our_method}. \\      
      $\tilde{\mathbf{a}}_t$ & The response action selected after planning; cf.~(\ref{eq:rollout}).  \\
      $J, \Tilde{J}$ & Recovery time-to-go functions (true and estimated); cf.\sectionref{sec:decision_planning}.  \\
      $Q, \Tilde{Q}$ & Q-functions (true and estimated by \textsc{llm}); cf. \sectionref{sec:decision_planning}.  \\
      $\mathcal{S}, \mathcal{A}$ & Sets of recovery states and response actions; cf. \sectionref{sec:decision_planning}.  \\
      $\tilde{\mathcal{S}}$ & Set of non-terminal recovery states; cf. \sectionref{sec:decision_planning}.  \\
      $\mathcal{A}^N_t$ & The set of $N$ candidate actions at time step $t$; cf. \sectionref{sec:decision_planning}.  \\                       
    \bottomrule\\
  \end{tabular}}
  \caption{Notation.}\label{tab:notation}
\end{table}
\section{Experimental Setup}\label{appendix:hyperparameters}
All computations are performed using $4\times $\textsc{quadro rtx 8000} \textsc{gpus}. The hyperparameters that we use for fine-tuning and for instantiating \textsc{ppo} are listed in \tableref{tab:hyperparams}. Parameters not listed in \tableref{tab:hyperparams} are set to default values.

\begin{table}
  \centering
  \scalebox{0.7}{
    \begin{tabular}{ll} \toprule
\rowcolor{lightgray}
    {\textit{Parameter(s)}} & {\textit{Value(s)}} \\ \midrule
      \textsc{lora} rank $r$ & $64$\\
      \textsc{lora} $\alpha$ & $128$\\
      \textsc{lora} dropout & $0.05$\\
      Learning rate & $0.00095$\\
      Batch size & $5$ \\
      Gradient accumulation steps & $16$ \\
      Temperature & $0.6$ \\
      Number of training epochs & $2$ \\
      Quantization & $4$ bit \\
    \midrule
    \textsc{ppo} \cite[Alg. 1]{ppo} &   \\
    \midrule      
    Learning rate, \# hidden layers  & $5148 \cdot 10^{-5}$, $1$, \\
    \# Neurons/layer & $64$\\
    \# Steps between updates, & $2048$,\\      
    Batch size, discount factor $\gamma$ & $16$, $0.99$\\
    \textsc{gae} $\lambda$, clip range, entropy coefficient & $0.95$, $0.2$, $2\cdot 10^{-4}$\\
    Value coefficient, max gradient norm & $0.102$, $0.5$\\
    Feature representation & the original cyborg features \cite{cyborg} \& \\
                      & one-hot encoded scan-state \& \\
                      & decoy-state for each node\\      
    \bottomrule\\
  \end{tabular}}
  \caption{Hyperparameters.}\label{tab:hyperparams}
\end{table}

\section{Dataset Generation}\label{app:example_dataset}
To generate the dataset of instruction-answer pairs for fine-tuning, we use a combination of log data from our testbed and synthetic data generated by frontier \textsc{llm}s. Specifically, we first run the attacks listed in \tableref{tab:attacker_actions} on our testbed. Such runs generate system measurements and logs (e.g., \textsc{snort} alerts \cite{snort}), based on which we construct $500$ instruction-answer pairs. Each instruction consists of log data and previously applied response actions, as well as a task to either generate a response action or predict the current recovery state; cf.~\defref{def:recovery_state}. Similarly, each answer either consists of the true recovery state or the optimal response action, both of which are manually selected based on knowledge about the incident.

Since these $500$ instruction-answer pairs are too few for effective fine-tuning, we then expand the dataset using synthetic data generated by prompting \textsc{gemini 2.5} \cite{comanici2025gemini25pushingfrontier} and \textsc{openai o3} \cite{openai2024gpt4technicalreport} with our testbed examples to generate new examples of similar structure but for different types of systems and attacks, yielding a total dataset of size $68,000$. This dataset covers a diverse range of attacks, system architectures, and log data. Each instruction-answer pair follows a specific \textsc{json} structure. \Figref{fig:dataset} shows the distributions of token counts and \textsc{mitre att\&ck} tactics \cite{strom2018mitre} in our dataset. We see that most incidents are described by around $1200$ tokens, and the most common attacker tactics are \textsc{initial access} and \textsc{execution}. The prompt templates that we use are available at \cite{llm_source_kim}.

Our approach of combining testbed examples with synthetic examples is inspired by the studies presented in \cite{wang2023farcamelsgoexploring} and \cite{yu2025finemedlmo1enhancingmedicalreasoning}, which successfully used similar approaches to generate fine-tuning datasets for other domains, e.g., the medical domain.

\begin{figure}
  \centering
  \scalebox{0.7}{
   \input{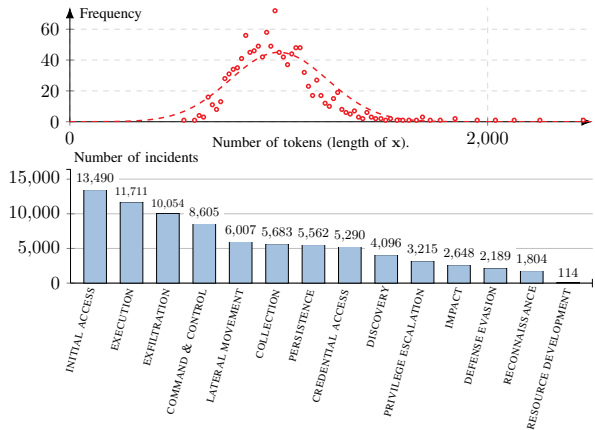}    
  }
  \caption{Distributions of the number of tokens (upper plot) and \textsc{mitre att\&ck tactics} \cite{strom2018mitre} (lower plot) in our dataset of instruction-answer pairs ($\mathbf{x}$, $\mathbf{y}$), which we use for fine-tuning the \textsc{llm}. The dataset is available at \cite{llm_source_kim}.}
  \label{fig:dataset}
\end{figure}

\begin{table}
  \centering
\scalebox{0.78}{  
\begin{tabular}{lll} \toprule
\rowcolor{lightgray}
  {\textit{Type}} & {\textit{Actions}} & {\textsc{mitre att\&ck} technique} \\ \midrule
  Reconnaissance  & \tcpp \syn scan, \udp scan & \textsc{t1046} service scanning\\
                  & \tcpp \xmas scan & \textsc{t1046} service scanning \\
                  & \vulscan & \textsc{t1595} active scanning \\
                  & ping-scan & \textsc{t1018} system discovery\\\midrule
  Brute-force & \telnet, \ssh & \textsc{t1110} brute force\\
                  & \ftp, \cassandra & \textsc{t1110} brute force\\
                  &  \irc, \mongo, \mysql & \textsc{t1110} brute force\\
                  & \smtp, \postgres & \textsc{t1110} brute force\\\midrule
  Exploit & \cve-2017-7494 & \textsc{t1210} service exploitation\\
                  &\cve-2015-3306 & \textsc{t1210} service exploitation\\
                  & \cve-2010-0426 & \textsc{t1068} privilege escalation\\
                  & \cve-2015-5602 & \textsc{t1068} privilege escalation\\
                  & \cve-2015-1427 & \textsc{t1210} service exploitation\\
                  & \cve-2014-6271 & \textsc{t1210} service exploitation\\
                  & \cve-2016-10033 & \textsc{t1210} service exploitation\\
                  & \textsc{sql} injection (\cwe-89) & \textsc{t1210} service exploitation \\
  \bottomrule\\
\end{tabular}}
\caption{Attacker actions executed on our testbed to generate the initial examples for our training dataset, which we use for fine-tuning the \textsc{llm}. Actions are mapped to the corresponding vulnerabilities they exploit, as indicated by the \cve{} \cite{cve} and \cwe{} \cite{cwe} identifiers. The actions are also linked to the corresponding attack techniques in the \textsc{mitre att\&ck} taxonomy \cite{strom2018mitre}.}\label{tab:attacker_actions}
\end{table}
\bibliographystyle{IEEEtran}
\bibliography{references,url}

\begin{thebibliography}{10}
\providecommand{\url}[1]{#1}
\csname url@samestyle\endcsname
\providecommand{\newblock}{\relax}
\providecommand{\bibinfo}[2]{#2}
\providecommand{\BIBentrySTDinterwordspacing}{\spaceskip=0pt\relax}
\providecommand{\BIBentryALTinterwordstretchfactor}{4}
\providecommand{\BIBentryALTinterwordspacing}{\spaceskip=\fontdimen2\font plus
\BIBentryALTinterwordstretchfactor\fontdimen3\font minus
  \fontdimen4\font\relax}
\providecommand{\BIBforeignlanguage}[2]{{%
\expandafter\ifx\csname l@#1\endcsname\relax
\typeout{** WARNING: IEEEtran.bst: No hyphenation pattern has been}%
\typeout{** loaded for the language `#1'. Using the pattern for}%
\typeout{** the default language instead.}%
\else
\language=\csname l@#1\endcsname
\fi
#2}}
\providecommand{\BIBdecl}{\relax}
\BIBdecl

\bibitem{287145}
D.~W. Woods, R.~B{\"o}hme, J.~Wolff, and D.~Schwarcz, ``Lessons lost: Incident
  response in the age of cyber insurance and breach attorneys,'' in \emph{32nd
  USENIX Security Symposium (USENIX Security 23)}.\hskip 1em plus 0.5em minus
  0.4em\relax Anaheim, CA: USENIX Association, Aug. 2023, pp. 2259--2273.

\bibitem{ibm2024costofdatabreach}
{IBM Security} and P.~Institute, ``Cost of a data breach report 2024,'' IBM,
  Cambridge, MA, Tech. Rep.~19, 2024, based on breaches at 524 organizations
  across 17 industries in 16 countries between March 2023 and February 2024.

\bibitem{10.1145/3491102.3517559}
R.~Stevens, D.~Votipka, J.~Dykstra, F.~Tomlinson, E.~Quartararo, C.~Ahern, and
  M.~L. Mazurek, ``How ready is your ready? assessing the usability of incident
  response playbook frameworks,'' in \emph{Proceedings of the 2022 CHI
  Conference on Human Factors in Computing Systems}, ser. CHI '22.\hskip 1em
  plus 0.5em minus 0.4em\relax New York, NY, USA: Association for Computing
  Machinery, 2022.

\bibitem{10646756}
D.~Schlette, P.~Empl, M.~Caselli, T.~Schreck, and G.~Pernul, ``Do you play it
  by the books? a study on incident response playbooks and influencing
  factors,'' in \emph{2024 IEEE Symposium on Security and Privacy (SP)}, 2024,
  pp. 3625--3643.

\bibitem{10.5555/3600270.3602070}
J.~Wei, X.~Wang, D.~Schuurmans, M.~Bosma, B.~Ichter, F.~Xia, E.~H. Chi, Q.~V.
  Le, and D.~Zhou, ``Chain-of-thought prompting elicits reasoning in large
  language models,'' in \emph{Proceedings of the 36th International Conference
  on Neural Information Processing Systems}, ser. NIPS '22.\hskip 1em plus
  0.5em minus 0.4em\relax Red Hook, NY, USA: Curran Associates Inc., 2022.

\bibitem{castro2025largelanguagemodelsautonomous}
S.~R. Castro, R.~Campbell, N.~Lau, O.~Villalobos, J.~Duan, and A.~A. Cardenas,
  ``Large language models are autonomous cyber defenders,'' 2025,
  \url{https://arxiv.org/abs/2505.04843}.

\bibitem{rigaki2023cage}
M.~Rigaki, O.~Lukáš, C.~A. Catania, and S.~Garcia, ``Out of the cage: How
  stochastic parrots win in cyber security environments,'' 2023,
  \url{https://arxiv.org/abs/2308.12086}.

\bibitem{10991969}
H.~Mohammadi, J.~J. Davis, and M.~Kiely, ``Leveraging large language models for
  autonomous cyber defense: Insights from {CAGE-2} simulations,'' \emph{IEEE
  Intelligent Systems}, pp. 1--8, 2025.

\bibitem{yan2024dependingshouldmentoringllm}
Y.~Yan, Y.~Zhang, and K.~Huang, ``Depending on yourself when you should:
  Mentoring {LLM} with {RL} agents to become the master in cybersecurity
  games,'' 2024, \url{https://arxiv.org/abs/2403.17674}.

\bibitem{hays2024employingllmsincidentresponse}
S.~Hays and J.~White, ``Employing {LLM}s for incident response planning and
  review,'' 2024, \url{https://arxiv.org/abs/2403.01271}.

\bibitem{lin2025ircopilotautomatedincidentresponse}
X.~Lin, J.~Zhang, G.~Deng, T.~Liu, X.~Liu, C.~Yang, T.~Zhang, Q.~Guo, and
  R.~Chen, ``{IRCopilot}: Automated incident response with large language
  models,'' 2025, \url{https://arxiv.org/abs/2505.20945}.

\bibitem{10540988}
J.~F. Loevenich, E.~Adler, R.~Mercier, A.~Velazquez, and R.~R.~F. Lopes,
  ``Design of an autonomous cyber defence agent using hybrid {AI} models,'' in
  \emph{2024 International Conference on Military Communication and Information
  Systems (ICMCIS)}, 2024, pp. 1--10.

\bibitem{hussey2025instana}
S.~Hussey. (2025, June) Resolve incidents faster with {IBM} {Instana}
  intelligent incident investigation powered by agentic {AI}. Accessed:
  2025-07-08,
  \url{https://www.ibm.com/new/announcements/resolve-incidents-faster-with-ibm-instana-intelligent-incident-investigation-powered-by-agentic-ai}.

\bibitem{openai2024gpt4technicalreport}
OpenAI, J.~Achiam, S.~Adler \emph{et~al.}, ``{GPT-4} technical report,'' 2024,
  \url{https://arxiv.org/abs/2303.08774}.

\bibitem{NEURIPS2024_3c1e1fdf}
G.~Sriramanan, S.~Bharti, V.~S. Sadasivan, S.~Saha, P.~Kattakinda, and
  S.~Feizi, ``{LLM}-check: Investigating detection of hallucinations in large
  language models,'' in \emph{Advances in Neural Information Processing
  Systems}, A.~Globerson, L.~Mackey, D.~Belgrave, A.~Fan, U.~Paquet,
  J.~Tomczak, and C.~Zhang, Eds., vol.~37.\hskip 1em plus 0.5em minus
  0.4em\relax Curran Associates, Inc., 2024, pp. 34\,188--34\,216.

\bibitem{comanici2025gemini25pushingfrontier}
G.~Comanici, E.~Bieber, M.~Schaekermann \emph{et~al.}, ``Gemini 2.5: Pushing
  the frontier with advanced reasoning, multimodality, long context, and next
  generation agentic capabilities,'' 2025,
  \url{https://arxiv.org/abs/2507.06261}.

\bibitem{geminiteam2024geminifamilyhighlycapable}
G.~Team, R.~Anil, S.~Borgeaud \emph{et~al.}, ``Gemini: A family of highly
  capable multimodal models,'' 2024, \url{https://arxiv.org/abs/2312.11805}.

\bibitem{ppo}
J.~Schulman, F.~Wolski, P.~Dhariwal, A.~Radford, and O.~Klimov, ``Proximal
  policy optimization algorithms,'' \emph{CoRR}, 2017,
  \url{http://arxiv.org/abs/1707.06347}.

\bibitem{llm_source_kim}
K.~Hammar, T.~Alpcan, and E.~C. Lupu, ``{Supplementary material of the paper
  "Incident Response Planning Using a Lightweight Large Language Model with
  Reduced Hallucination"},'' 2025, code for fine-tuning:
  \url{https://github.com/Limmen/llm_recovery}, code for our testbed:
  \url{https://github.com/Limmen/csle}, dataset:
  \url{https://huggingface.co/datasets/kimhammar/CSLE-IncidentResponse-V1},
  video demonstration: \url{https://www.youtube.com/watch?v=e7ckmv5p6cI},
  fine-tuned {LLM} and prompts:
  \url{https://huggingface.co/kimhammar/LLMIncidentResponse}.

\bibitem{10.1504/IJICS.2007.012248}
N.~Stakhanova, S.~Basu, and J.~Wong, ``A taxonomy of intrusion response
  systems,'' \emph{Int. J. Inf. Comput. Secur.}, vol.~1, no. 1/2, p. 169–184,
  Jan. 2007.

\bibitem{tansu_response_2003}
T.~Alpcan and T.~Basar, ``A game theoretic approach to decision and analysis in
  network intrusion detection,'' in \emph{42nd IEEE International Conference on
  Decision and Control (IEEE Cat. No.03CH37475)}, vol.~3, 2003, pp. 2595--2600
  Vol.3.

\bibitem{playbook_response}
A.~Applebaum, S.~Johnson, M.~Limiero, and M.~Smith, ``Playbook oriented cyber
  response,'' in \emph{2018 National Cyber Summit (NCS)}, 2018, pp. 8--15.

\bibitem{splunk_playbook}
Splunk, ``Automate incident response with playbooks and actions in {Splunk}
  mission control,'' 2025,
  \url{https://help.splunk.com/en/splunk-enterprise-security-7/mission-control/investigate-and-respond-to-threats/automate-incident-response/automate-incident-response-with-playbooks-and-actions-in-splunk-mission-control}.

\bibitem{cisa_playbook}
CISA, ``Cybersecurity incident \& vulnerability response playbooks,'' 2021,
  \url{https://www.cisa.gov/resources-tools/resources/federal-government-cybersecurity-incident-and-vulnerability-response-playbooks}.

\bibitem{oasis_playbook}
OASIS, ``Cacao security playbooks version 2.0,'' 2023,
  \url{https://docs.oasis-open.org/cacao/security-playbooks/v2.0/security-playbooks-v2.0.html}.

\bibitem{dsn24_hammar_stadler}
K.~Hammar and R.~Stadler, ``Intrusion tolerance for networked systems through
  two-level feedback control,'' in \emph{2024 54th Annual IEEE/IFIP
  International Conference on Dependable Systems and Networks (DSN)}, 2024, pp.
  338--352.

\bibitem{singh2024hierarchicalmultiagentreinforcementlearning}
A.~V. Singh, E.~Rathbun, E.~Graham, L.~Oakley, S.~Boboila, A.~Oprea, and
  P.~Chin, ``Hierarchical multi-agent reinforcement learning for cyber network
  defense,'' 2024, \url{https://arxiv.org/abs/2410.17351}.

\bibitem{10955193}
K.~Hammar, T.~Li, R.~Stadler, and Q.~Zhu, ``Adaptive security response
  strategies through conjectural online learning,'' \emph{IEEE Transactions on
  Information Forensics and Security}, vol.~20, pp. 4055--4070, 2025.

\bibitem{5270307}
S.~A. Zonouz, H.~Khurana, W.~H. Sanders, and T.~M. Yardley, ``{RRE}: A
  game-theoretic intrusion response and recovery engine,'' in \emph{2009
  IEEE/IFIP International Conference on Dependable Systems \& Networks}, 2009,
  pp. 439--448.

\bibitem{nework_security_alpcan}
T.~Alpcan and T.~Basar, \emph{Network Security: A Decision and Game-Theoretic
  Approach}, 1st~ed.\hskip 1em plus 0.5em minus 0.4em\relax USA: Cambridge
  University Press, 2010.

\bibitem{cage_challenge_2_announcement}
CAGE, ``{TTCP CAGE} challenge 2,'' in \emph{AAAI-22 Workshop on Artificial
  Intelligence for Cyber Security (AICS)}, 2022,
  \url{https://github.com/cage-challenge/cage-challenge-2}.

\bibitem{tifs_25_HLALB}
K.~Hammar, Y.~Li, T.~Alpcan, E.~C. Lupu, and D.~Bertsekas, ``Adaptive network
  security policies via belief aggregation and rollout,'' 2025,
  \url{https://arxiv.org/abs/2507.15163}.

\bibitem{vyas2023automated}
S.~Vyas, J.~Hannay, A.~Bolton, and P.~P. Burnap, ``Automated cyber defence: A
  review,'' 2023, code: \url{https://github.com/john-cardiff/-cyborg-cage-2}.

\bibitem{hammar2024optimaldefenderstrategiescage2}
K.~Hammar, N.~Dhir, and R.~Stadler, ``Optimal defender strategies for {CAGE-2}
  using causal modeling and tree search,'' 2024,
  https://arxiv.org/abs/2407.11070.

\bibitem{li2024conjectural}
T.~Li, K.~Hammar, R.~Stadler, and Q.~Zhu, ``Conjectural online learning with
  first-order beliefs in asymmetric information stochastic games,'' in
  \emph{2024 IEEE 63rd Conference on Decision and Control (CDC)}, 2024, pp.
  6780--6785.

\bibitem{tabular_Q_andy}
A.~Applebaum, C.~Dennler, P.~Dwyer, M.~Moskowitz, H.~Nguyen, N.~Nichols,
  N.~Park, P.~Rachwalski, F.~Rau, A.~Webster, and M.~Wolk, ``Bridging automated
  to autonomous cyber defense: Foundational analysis of tabular {Q}-learning,''
  in \emph{Proceedings of the 15th ACM Workshop on Artificial Intelligence and
  Security}, 2022.

\bibitem{ramamurthy2025generalautonomouscybersecuritydefense}
A.~Ramamurthy and N.~Dhir, ``General autonomous cybersecurity defense: Learning
  robust policies for dynamic topologies and diverse attackers,'' 2025,
  \url{https://arxiv.org/abs/2506.22706}.

\bibitem{huang2025intentbasedontologydrivenautonomicsecurity}
Z.~Huang, J.~Robin, N.~Herbaut, N.~B. Rabah, and B.~L. Grand, ``Toward an
  intent-based and ontology-driven autonomic security response in security
  orchestration automation and response,'' 2025,
  \url{https://arxiv.org/abs/2507.12061}.

\bibitem{10.1145/3538969.3538976}
A.~Shaked, Y.~Cherdantseva, and P.~Burnap, ``Model-based incident response
  playbooks,'' in \emph{Proceedings of the 17th International Conference on
  Availability, Reliability and Security}, ser. ARES '22.\hskip 1em plus 0.5em
  minus 0.4em\relax New York, NY, USA: Association for Computing Machinery,
  2022.

\bibitem{10.1145/3688810}
M.~Akbari~Gurabi, L.~Nitz, A.~Bregar, J.~Popanda, C.~Siemers, R.~Matzutt, and
  A.~Mandal, ``Requirements for playbook-assisted cyber incident response,
  reporting and automation,'' \emph{Digital Threats}, vol.~5, no.~3, Oct. 2024.

\bibitem{pentest_gpt}
G.~Deng, Y.~Liu, V.~Mayoral-Vilches, P.~Liu, Y.~Li, Y.~Xu, T.~Zhang, Y.~Liu,
  M.~Pinzger, and S.~Rass, ``{PentestGPT}: Evaluating and harnessing large
  language models for automated penetration testing,'' in \emph{33rd USENIX
  Security Symposium (USENIX Security 24)}.\hskip 1em plus 0.5em minus
  0.4em\relax Philadelphia, PA: USENIX Association, Aug. 2024, pp. 847--864.

\bibitem{rodriguez2025frameworkevaluatingemergingcyberattack}
M.~Rodriguez, R.~A. Popa, F.~Flynn, L.~Liang, A.~Dafoe, and A.~Wang, ``A
  framework for evaluating emerging cyberattack capabilities of {AI},'' 2025,
  \url{https://arxiv.org/abs/2503.11917}.

\bibitem{DBLP:conf/ndss/DengLCBWLW025}
J.~Deng, X.~Li, Y.~Chen, Y.~Bai, H.~Weng, Y.~Liu, T.~Wei, and W.~Xu,
  ``{RACONTEUR:} {A} knowledgeable, insightful, and portable {LLM}-powered
  shell command explainer,'' in \emph{32nd Annual Network and Distributed
  System Security Symposium, {NDSS} 2025, San Diego, California, USA, February
  24-28, 2025}.\hskip 1em plus 0.5em minus 0.4em\relax The Internet Society,
  2025.

\bibitem{DBLP:conf/ndss/StafeevRSKP25}
A.~Stafeev, T.~Recktenwald, G.~D. Stefano, S.~Khodayari, and G.~Pellegrino,
  ``Yurascanner: Leveraging {LLM}s for task-driven web app scanning,'' in
  \emph{32nd Annual Network and Distributed System Security Symposium, {NDSS}
  2025, San Diego, California, USA, February 24-28, 2025}.\hskip 1em plus 0.5em
  minus 0.4em\relax The Internet Society, 2025.

\bibitem{299549}
P.~Liu, J.~Liu, L.~Fu, K.~Lu, Y.~Xia, X.~Zhang, W.~Chen, H.~Weng, S.~Ji, and
  W.~Wang, ``Exploring {ChatGPT{\textquoteright}s} capabilities on
  vulnerability management,'' in \emph{33rd USENIX Security Symposium (USENIX
  Security 24)}.\hskip 1em plus 0.5em minus 0.4em\relax Philadelphia, PA:
  USENIX Association, Aug. 2024, pp. 811--828.

\bibitem{google_llm_recovery}
M.~Allamanis, M.~Arjovsky, C.~Blundell \emph{et~al.}, ``From naptime to big
  sleep: Using large language models to catch vulnerabilities in real-world
  code,'' 2024,
  \url{https://googleprojectzero.blogspot.com/2024/10/from-naptime-to-big-sleep.html}.

\bibitem{DBLP:conf/ndss/0012XW00S025}
Y.~Liu, Y.~Xue, D.~Wu, Y.~Sun, Y.~Li, M.~Shi, and Y.~Liu, ``{PropertyGPT}:
  {LLM}-driven formal verification of smart contracts through
  retrieval-augmented property generation,'' in \emph{32nd Annual Network and
  Distributed System Security Symposium, {NDSS} 2025, San Diego, California,
  USA, February 24-28, 2025}.\hskip 1em plus 0.5em minus 0.4em\relax The
  Internet Society, 2025.

\bibitem{DBLP:conf/ndss/GohilDNSR25}
V.~Gohil, M.~DeLorenzo, V.~V. A. S.~V. Nallam, J.~See, and J.~Rajendran,
  ``{LLMPirate}: {LLM}s for black-box hardware {IP} piracy,'' in \emph{32nd
  Annual Network and Distributed System Security Symposium, {NDSS} 2025, San
  Diego, California, USA, February 24-28, 2025}.\hskip 1em plus 0.5em minus
  0.4em\relax The Internet Society, 2025.

\bibitem{DBLP:conf/ndss/YangL0L25}
Y.~Yang, J.~Liu, K.~Chen, and M.~Lin, ``The midas touch: Triggering the
  capability of {LLM}s for {RM-API} misuse detection,'' in \emph{32nd Annual
  Network and Distributed System Security Symposium, {NDSS} 2025, San Diego,
  California, USA, February 24-28, 2025}.\hskip 1em plus 0.5em minus
  0.4em\relax The Internet Society, 2025.

\bibitem{10.1145/3597503.3639121}
C.~S. Xia, M.~Paltenghi, J.~Le~Tian, M.~Pradel, and L.~Zhang, ``Fuzz4all:
  Universal fuzzing with large language models,'' in \emph{Proceedings of the
  IEEE/ACM 46th International Conference on Software Engineering}, ser. ICSE
  '24.\hskip 1em plus 0.5em minus 0.4em\relax New York, NY, USA: Association
  for Computing Machinery, 2024.

\bibitem{299896}
X.~Ma, L.~Luo, and Q.~Zeng, ``From one thousand pages of specification to
  unveiling hidden bugs: Large language model assisted fuzzing of matter {IoT}
  devices,'' in \emph{33rd USENIX Security Symposium (USENIX Security
  24)}.\hskip 1em plus 0.5em minus 0.4em\relax Philadelphia, PA: USENIX
  Association, Aug. 2024, pp. 4783--4800.

\bibitem{DBLP:conf/ndss/LiuY0L25}
J.~Liu, Y.~Yang, K.~Chen, and M.~Lin, ``Generating {API} parameter security
  rules with {LLM} for {API} misuse detection,'' in \emph{32nd Annual Network
  and Distributed System Security Symposium, {NDSS} 2025, San Diego,
  California, USA, February 24-28, 2025}.\hskip 1em plus 0.5em minus
  0.4em\relax The Internet Society, 2025.

\bibitem{net_llm}
D.~Wu, X.~Wang, Y.~Qiao, Z.~Wang, J.~Jiang, S.~Cui, and F.~Wang, ``{NetLLM}:
  Adapting large language models for networking,'' in \emph{Proceedings of the
  ACM SIGCOMM 2024 Conference}, ser. ACM SIGCOMM '24.\hskip 1em plus 0.5em
  minus 0.4em\relax New York, NY, USA: Association for Computing Machinery,
  2024, p. 661–678.

\bibitem{ARAZZI2025100765}
M.~Arazzi, D.~{R. Arikkat}, S.~Nicolazzo, A.~Nocera, R.~{Rehiman K.A.}, V.~P.,
  and M.~Conti, ``{NLP}-based techniques for cyber threat intelligence,''
  \emph{Computer Science Review}, vol.~58, p. 100765, 2025.

\bibitem{DBLP:conf/ndss/HuL024}
P.~Hu, R.~Liang, and K.~Chen, ``Degpt: Optimizing decompiler output with
  {LLM},'' in \emph{31st Annual Network and Distributed System Security
  Symposium, {NDSS} 2024, San Diego, California, USA, February 26 - March 1,
  2024}.\hskip 1em plus 0.5em minus 0.4em\relax The Internet Society, 2024.

\bibitem{ross_security}
R.~J. Anderson, \emph{Security Engineering: A Guide to Building Dependable
  Distributed Systems}, 1st~ed.\hskip 1em plus 0.5em minus 0.4em\relax USA:
  John Wiley \& Sons, Inc., 2001.

\bibitem{Ganin2016}
A.~A. Ganin, E.~Massaro, A.~Gutfraind, N.~Steen, J.~M. Keisler, A.~Kott,
  R.~Mangoubi, and I.~Linkov, ``Operational resilience: concepts, design and
  analysis,'' \emph{Scientific Reports}, vol.~6, no.~1, p. 19540, Jan 2016.

\bibitem{273869}
L.~Li, X.~Zhang, X.~Zhao, H.~Zhang, Y.~Kang, P.~Zhao, B.~Qiao, S.~He, P.~Lee,
  J.~Sun, F.~Gao, L.~Yang, Q.~Lin, S.~Rajmohan, Z.~Xu, and D.~Zhang, ``Fighting
  the fog of war: Automated incident detection for cloud systems,'' in
  \emph{2021 USENIX Annual Technical Conference (USENIX ATC 21)}.\hskip 1em
  plus 0.5em minus 0.4em\relax USENIX Association, Jul. 2021, pp. 131--146.

\bibitem{wannacry_nhgs}
A.~Morse, ``Investigation: Wannacry cyber attack and the {NHS},'' 2017,
  national Audit Office UK.

\bibitem{tonmoy2024comprehensivesurveyhallucinationmitigation}
S.~M. T.~I. Tonmoy, S.~M.~M. Zaman, V.~Jain, A.~Rani, V.~Rawte, A.~Chadha, and
  A.~Das, ``A comprehensive survey of hallucination mitigation techniques in
  large language models,'' 2024, \url{https://arxiv.org/abs/2401.01313?}

\bibitem{deepseekai2025deepseekr1incentivizingreasoningcapability}
DeepSeek-AI, D.~Guo, D.~Yang \emph{et~al.}, ``{DeepSeek-R1}: Incentivizing
  reasoning capability in {LLM}s via reinforcement learning,'' 2025,
  \url{https://arxiv.org/abs/2501.12948}.

\bibitem{ayala-bechard-2024-reducing}
O.~Ayala and P.~Bechard, ``Reducing hallucination in structured outputs via
  retrieval-augmented generation,'' in \emph{Proceedings of the 2024 Conference
  of the North American Chapter of the Association for Computational
  Linguistics: Human Language Technologies (Volume 6: Industry Track)},
  Y.~Yang, A.~Davani, A.~Sil, and A.~Kumar, Eds.\hskip 1em plus 0.5em minus
  0.4em\relax Mexico City, Mexico: Association for Computational Linguistics,
  Jun. 2024, pp. 228--238.

\bibitem{DBLP:conf/iclr/0002WSLCNCZ23}
X.~Wang, J.~Wei, D.~Schuurmans, Q.~V. Le, E.~H. Chi, S.~Narang, A.~Chowdhery,
  and D.~Zhou, ``Self-consistency improves chain of thought reasoning in
  language models,'' in \emph{The Eleventh International Conference on Learning
  Representations, {ICLR} 2023, Kigali, Rwanda, May 1-5, 2023}.\hskip 1em plus
  0.5em minus 0.4em\relax OpenReview.net, 2023.

\bibitem{chen2023universalselfconsistencylargelanguage}
X.~Chen, R.~Aksitov, U.~Alon, J.~Ren, K.~Xiao, P.~Yin, S.~Prakash, C.~Sutton,
  X.~Wang, and D.~Zhou, ``Universal self-consistency for large language model
  generation,'' 2023, \url{https://arxiv.org/abs/2311.17311}.

\bibitem{weng-etal-2023-large}
Y.~Weng, M.~Zhu, F.~Xia, B.~Li, S.~He, S.~Liu, B.~Sun, K.~Liu, and J.~Zhao,
  ``Large language models are better reasoners with self-verification,'' in
  \emph{Findings of the Association for Computational Linguistics: EMNLP 2023},
  H.~Bouamor, J.~Pino, and K.~Bali, Eds.\hskip 1em plus 0.5em minus 0.4em\relax
  Singapore: Association for Computational Linguistics, Dec. 2023, pp.
  2550--2575.

\bibitem{10.5555/3666122.3668141}
A.~Madaan, N.~Tandon, P.~Gupta, S.~Hallinan, L.~Gao, S.~Wiegreffe, U.~Alon,
  N.~Dziri, S.~Prabhumoye, Y.~Yang, S.~Gupta, B.~P. Majumder, K.~Hermann,
  S.~Welleck, A.~Yazdanbakhsh, and P.~Clark, ``Self-refine: iterative
  refinement with self-feedback,'' in \emph{Proceedings of the 37th
  International Conference on Neural Information Processing Systems}, ser. NIPS
  '23.\hskip 1em plus 0.5em minus 0.4em\relax Red Hook, NY, USA: Curran
  Associates Inc., 2023.

\bibitem{lewis2020retrieval}
P.~Lewis, E.~Perez, A.~Piktus, F.~Petroni, V.~Karpukhin, N.~Goyal,
  H.~K{\"u}ttler, M.~Lewis, W.-t. Yih, T.~Rockt{\"a}schel \emph{et~al.},
  ``Retrieval-augmented generation for knowledge-intensive {NLP} tasks,''
  \emph{Advances in neural information processing systems}, vol.~33, pp.
  9459--9474, 2020.

\bibitem{cve}
MITRE, ``{CVE} database,'' 2022, \url{https://cve.mitre.org/}.

\bibitem{kaloroumakis2021d3fend}
P.~E. Kaloroumakis and M.~J. Smith, ``Toward a knowledge graph of cybersecurity
  countermeasures,'' The MITRE Corporation, Annapolis Junction, MD, Technical
  Report, 2021, approved for Public Release; Distribution Unlimited.

\bibitem{GARCIA2014100}
S.~García, M.~Grill, J.~Stiborek, and A.~Zunino, ``An empirical comparison of
  botnet detection methods,'' \emph{Computers \& Security}, vol.~45, pp.
  100--123, 2014.

\bibitem{8418627}
D.~Y. Huang, M.~M. Aliapoulios, V.~G. Li, L.~Invernizzi, E.~Bursztein,
  K.~McRoberts, J.~Levin, K.~Levchenko, A.~C. Snoeren, and D.~McCoy, ``Tracking
  ransomware end-to-end,'' in \emph{2018 IEEE Symposium on Security and Privacy
  (SP)}, 2018, pp. 618--631.

\bibitem{snort}
M.~Roesch, ``Snort - lightweight intrusion detection for networks,'' in
  \emph{Proceedings of the 13th USENIX Conference on System Administration},
  ser. LISA '99.\hskip 1em plus 0.5em minus 0.4em\relax USA: USENIX
  Association, 1999, p. 229–238.

\bibitem{icissp18}
I.~Sharafaldin, A.~{Habibi Lashkari}, and A.~A. Ghorbani, ``Toward generating a
  new intrusion detection dataset and intrusion traffic characterization,'' in
  \emph{Proceedings of the 4th International Conference on Information Systems
  Security and Privacy - ICISSP}, INSTICC.\hskip 1em plus 0.5em minus
  0.4em\relax SciTePress, 2018, pp. 108--116.

\bibitem{ait_ids_1}
M.~Landauer, F.~Skopik, and M.~Wurzenberger, ``Introducing a new alert data set
  for multi-step attack analysis,'' in \emph{Proceedings of the 17th Cyber
  Security Experimentation and Test Workshop}, ser. CSET '24.\hskip 1em plus
  0.5em minus 0.4em\relax New York, NY, USA: Association for Computing
  Machinery, 2024, p. 41–53.

\bibitem{wazuh}
{Wazuh Inc}, ``Wazuh - the open source security platform,'' 2022.

\bibitem{alienvaultOTX}
{AT\&T Cybersecurity}, ``{AlienVault Open Threat Exchange (OTX)},''
  \url{https://otx.alienvault.com}, 2021, \url{https://otx.alienvault.com}.

\bibitem{strom2018mitre}
B.~E. Strom, A.~Applebaum, D.~P. Miller, K.~C. Nickels, A.~G. Pennington, and
  C.~B. Thomas, ``{Mitre ATT\&CK}: Design and philosophy,'' in \emph{Technical
  report}.\hskip 1em plus 0.5em minus 0.4em\relax MITRE, 2018.

\bibitem{9866880}
M.~Landauer, F.~Skopik, M.~Frank, W.~Hotwagner, M.~Wurzenberger, and A.~Rauber,
  ``Maintainable log datasets for evaluation of intrusion detection systems,''
  \emph{IEEE Transactions on Dependable and Secure Computing}, vol.~20, no.~4,
  pp. 3466--3482, 2023.

\bibitem{hammar_stadler_tnsm}
K.~Hammar and R.~Stadler, ``Intrusion prevention through optimal stopping,''
  \emph{IEEE Transactions on Network and Service Management}, vol.~19, no.~3,
  pp. 2333--2348, 2022.

\bibitem{kim_phd_thesis}
K.~Hammar, ``Optimal security response to network intrusions in it systems,''
  Ph.D. dissertation, KTH Royal Instistute of Technology, 2024.

\bibitem{yadkori2024mitigatingllmhallucinationsconformal}
Y.~A. Yadkori, I.~Kuzborskij, D.~Stutz, A.~György, A.~Fisch, A.~Doucet,
  I.~Beloshapka, W.-H. Weng, Y.-Y. Yang, C.~Szepesvári, A.~T. Cemgil, and
  N.~Tomasev, ``Mitigating {LLM} hallucinations via conformal abstention,''
  2024, \url{https://arxiv.org/abs/2405.01563}.

\bibitem{bertsekas2021rollout}
D.~Bertsekas, \emph{Rollout, Policy Iteration, and Distributed Reinforcement
  Learning}, ser. Athena scientific optimization and computation series.\hskip
  1em plus 0.5em minus 0.4em\relax Athena Scientific, 2021.

\bibitem{10.5555/1396348}
D.~P. Bertsekas, \emph{Dynamic Programming and Optimal Control, Vol. II},
  3rd~ed.\hskip 1em plus 0.5em minus 0.4em\relax Athena Scientific, 2007.

\bibitem{cyborg}
M.~Standen, M.~Lucas, D.~Bowman, T.~J. Richer, J.~Kim, and D.~Marriott,
  ``Cyborg: {A} gym for the development of autonomous cyber agents,''
  \emph{CoRR}, 2021, \url{https://arxiv.org/abs/2108.09118}.

\bibitem{wang2023farcamelsgoexploring}
Y.~Wang, H.~Ivison, P.~Dasigi, J.~Hessel, T.~Khot, K.~R. Chandu, D.~Wadden,
  K.~MacMillan, N.~A. Smith, I.~Beltagy, and H.~Hajishirzi, ``How far can
  camels go? exploring the state of instruction tuning on open resources,''
  2023, \url{https://arxiv.org/abs/2306.04751}.

\bibitem{yu2025finemedlmo1enhancingmedicalreasoning}
H.~Yu, T.~Cheng, Y.~Cheng, and R.~Feng, ``{FineMedLM-o1}: Enhancing the medical
  reasoning ability of {LLM} from supervised fine-tuning to test-time
  training,'' 2025, \url{https://arxiv.org/abs/2501.09213}.

\bibitem{cwe}
MITRE, ``{CWE} list,'' 2023, \url{https://cwe.mitre.org/index.html}.

\end{thebibliography}
\end{document}